\documentclass[10pt]{article}
\usepackage[a4paper,left=1in, top=1in, right=1in, bottom=1in]{geometry}

\usepackage{amsfonts}
\usepackage{amssymb}
\usepackage{url}
\usepackage{booktabs}
\usepackage{graphicx}
\graphicspath{{Figures/}}
\usepackage{amsthm,bm}
\usepackage{pbox}
\usepackage{mathrsfs}
\usepackage{enumitem}
\usepackage{subfigure}
\usepackage{color, colortbl}

\usepackage{arcs}
\usepackage{tipa}

\usepackage{framed}
\usepackage[framemethod=tikz]{mdframed}

\usepackage[many]{tcolorbox}
\tcolorboxenvironment{theorem}{
  breakable,
  colback=gray!20,
  boxrule=0pt,
  boxsep=1pt,
  left=1em,right=1em,top=1em,bottom=1em,
  arc=10pt,
  before skip=\topsep,
  after skip=\topsep,
}

\definecolor{Gray}{gray}{0.9}
\definecolor{LightCyan}{rgb}{0.88,1,1}
\newcolumntype{g}{>{\columncolor{Gray}}c}

\usepackage[hidelinks]{hyperref}
\hypersetup{
	colorlinks, 
	linkcolor=black, 
	anchorcolor=blue, 
	citecolor=blue
}

\usepackage{cite}

\newtheorem{theorem}{Theorem}
\newtheorem{corollary}{Corollary}
\newtheorem{lemma}{Lemma}
\newtheorem{remark}{Remark}

\newtheorem{definition}{Definition}
\newtheorem{assumption}{Assumption}
\newtheorem{proposition}{Proposition}

\usepackage{mathtools, cuted}

\usepackage[lined,ruled]{algorithm2e}
\usepackage{algorithmic}

\usepackage{bbm}
\usepackage{dsfont}

\usepackage{accents}
\DeclareRobustCommand{\ubar}[1]{\underaccent{\bar}{#1}}
\begin{document}

\title{Mechanism Design for Online Resource Allocation:\\ A Unified Approach\thanks{The conference version of this paper was presented at ACM SIGMETRICS 2020, and the journal version was published by Proceedings of ACM on Measurement and Analysis of Computing Systems, vol. 4, no. 2, June 2020. (\textit{Corresponding Author: Xiaoqi Tan})}
}

\author{Xiaoqi Tan\thanks{Department of Computing Science, University of Alberta. Email: {\tt xiaoqi.tan@ualberta.ca.} This work was done while the author was with the University of Toronto. }
\and Bo Sun\thanks{Hong Kong University of Science and Technology. Email: {\tt eebosun@ust.hk}} 
\and Alberto Leon-Garcia\thanks{University of Toronto. Email: {\tt alberto.leongarcia@utoronto.ca}} 
\and Yuan Wu\thanks{University of Macau. Email: {\tt yuanwu@um.edu.mo}} 
\and Danny H.K. Tsang\thanks{Hong Kong University of Science and Technology. Email: {\tt eetsang@ust.hk}} 
}

\date{}

\begin{titlepage}
\maketitle
\thispagestyle{empty}

\begin{abstract}
This paper concerns the mechanism design for online resource allocation in a strategic setting.  In this setting, a single supplier allocates capacity-limited resources to requests that arrive in a sequential and arbitrary manner. Each request is associated with an agent who may act selfishly to misreport the requirement and valuation of her request. The supplier charges payment from agents whose requests are satisfied, but incurs a load-dependent supply cost. The goal is to design an incentive compatible  online mechanism, which determines not only the resource allocation of each request, but also the payment of each agent,  so as to (approximately) maximize  the social welfare (i.e., aggregate valuations minus supply cost). We study this problem under the framework of competitive analysis. The major contribution of this paper is the development of a unified approach that achieves the best-possible competitive ratios for setups with different supply costs. Specifically, we show that when there is no supply cost or the supply cost function is linear, our model is essentially a standard 0-1 knapsack problem, for which our approach achieves logarithmic competitive ratios that match the state-of-the-art (which is optimal). For the more challenging setup when the supply cost is  strictly-convex, we provide online mechanisms, for the first time, that lead to the optimal competitive ratios as well.  To the best of our knowledge, this is the first approach that unifies the characterization of optimal competitive ratios in online resource allocation for different setups including zero, linear and strictly-convex supply costs.
\end{abstract}

\end{titlepage}


\section{Introduction}
We study the mechanism design for online resource allocation problems. A single supplier who allocates capacity-limited resources (e.g., computing cycles, network bandwidth, energy, etc.) to requests that arrive in a sequential and arbitrary manner. We consider a strategic setting where each request is owned by a self-interested agent who may deliberately misreport the resource requirement and value of the request.  A request is satisfied if the required resource is allocated to the corresponding agent. The supplier charges payment from agents whose requests are satisfied, and affords a supply cost which is a function of the total resource allocated. The goal is to design not only the resource allocation of each request, but also the payment of each agent, so that agents are well-incentivized to follow their true preferences (i.e., incentive compatible \cite{AGT}) and meanwhile, the social welfare (i.e., the aggregate value minus the supply cost ) can be approximately maximized.

A prominent application of this model is the market-based resource allocation in cloud computing \cite{value_based_scheduling}. Here, the resource may represent computing cycles that can run at different speeds. Cloud service providers charge money from  customers who purchase their services\footnote{{\color{black} In practice,  major cloud service providers such as Amazon Web Services, Microsoft Azure, and Google Cloud often use fixed pricing schemes. However, dynamic pricing is also adopted under some circumstances. For example, customers of Amazon EC2 Spot Instances (\url{https://aws.amazon.com/ec2/spot/}) are charged based on spot prices that are adjusted in real-time. Typically, the spot prices are determined based on multiple factors such as long-term trend in demand and supply of spot instance capacity.}}, but must pay a considerable amount of power and cooling cost to maintain their data centers. Moreover, such operational costs are often load-dependent and  is usually an increasing function of the total resource allocated (e.g., CPU). Another application of the investigated model may arise in the context of network routing with congestion cost \cite{network_RA, congestion_games_2005,  flow_routing_2005}. Each incoming user wants to own some pair of connections with a valuation if any feasible connection is established. Congestion cost occurs when a link is occupied by multiple users, and the cost can be modelled as a traffic-dependent increasing function \cite{network_RA, diseconomy_cost}. The target is to maximize the total valuation of  routed connections minus the congestion cost.  In reality, such congestion costs can either reflect the total energy needed to support the network, or simply a virtual cost to penalize the degradation of the quality-of-service (e.g., latency \cite{flow_routing_2005, congestion_games_2005}). The authors of \cite{diseconomy_cost} indicated that the power-rate curve for networking or computing devices (e.g., CPU, communication links, and edge routers, etc) exhibit \textit{diseconomy-of-scale} properties (i.e., a convex curve), and thus polynomial functions are commonly used to model the power requirements of the network \cite{network_RA}.

Motivated by the above applications, this paper primarily focuses on setups when the supply cost function is strictly-convex and differentiable, i.e., the marginal cost is strictly increasing. 
A basic challenge for online resource allocation in such setups is  as follows: if the limited resource is allocated too aggressively, then an excessive portion of the resource may be allocated to earlier requests with low values. This will increase the total cost for the supplier and thus increase the payment for future agents, which will consequently prevent the later agents from purchasing the resources even if their valuations are higher than the earlier ones. On the other hand, if the payment is set too high, a majority of the requests may be declined, leading to a poor performance as well. This paper aims to tackle this challenge by designing an online mechanism that leads to the best-possible performance in social welfare.  More importantly, our optimal results for strictly-convex costs extend to setups with no supply cost (i.e., zero cost) or linear supply cost,  leading to a  unified approach for online resource allocation with zero, linear or strictly-convex supply costs.

\subsection{Related Results}
\label{section_literature}
Different variants of online resource allocation problems  have been studied in both the non-strategic setting, e.g., \cite{b_matching, Adword2009, covering_packing, Buchbinder2009, knapsack2008, online_scheduling},  and the strategic setting, e.g.,  \cite{Bartal2003, Blum2011, Buchbinder2015, Huang2015, posted_price_MAMA}.

In the non-strategic setting, the main focus is to design online algorithms with a competitive fraction of optimum.  Classic problems in this stream of studies include online bipartite matching \cite{b_matching}, Adwords problems \cite{Adword2009, Mehta2013,  Adwords2007},  online covering and packing problems \cite{online_routing_FOCS, Buchbinder2009}, online knapsack problems \cite{knapsack2008, Zhang2017}, and one-way trading problems \cite{trading2001}. For example, the authors of \cite{Adwords2007} introduced the Adwords problem and  characterized the optimal competitive ratio  by assuming each bid is much smaller than the total budget (i.e., the infinitesimal assumption). Similarly, the authors of \cite{knapsack2008} proved that, when the weight of each item is much smaller than the capacity of the knapsack, and the value-to-weight ratio of each item is bounded within $ [L,U] $, no online algorithm can achieve a competitive ratio tighter than $ 1+\ln(U/L) $.  In addition to the above classic problems, new variants of online resource allocation problems have also been reported, e.g.,  \cite{covering_packing, concave_return, Lin2019, willma}. For example, a generalization of the online matching and Adwords problem was studied in \cite{concave_return}, where each incoming agent has a concave function representing the return or utility of this agent, in contrast to the budgeted linear utility function studied by almost all previous online matching problems. In particular, the authors of \cite{concave_return} derived a differential equation based on the primal-dual analysis, and then characterized the optimal competitive ratio by analyzing the differential equation via variations of calculus. Lin et al. \cite{Lin2019} recently proposed an optimal online algorithm for a generalized one-way trading problem, where the online decisions are made in each round to maximize the per-round concave revenue function under inventory (budget) constraints. Recently, an extended online matching problem was presented by \cite{willma}, in which the items can be sold at multiple feasible prices instead of a single price.  Based on this setup, Ma et al. \cite{willma} designed an online algorithm with the best-possible competitive ratios that depend on the sets of multiple prices given in advance. 

In the strategic setting, agents are self-interested and thus a careful payment rule must be designed along with the online allocation algorithms so as to guarantee the incentive compatibility. A classic problem in this setting is the  online combinatorial auction (CA) problem, e.g., \cite{AGT, CA_unlimited_supply_2008,  Bartal2003, Buchbinder2015, Blum2011, Huang2015}. Existing studies studying the online CAs focus on two cases: one is with stringent supply constraints (i.e., limited supply) and the other is unlimited supply, namely, the supplier can produce as many copies of items as possible with no cost.  For example,  the authors of \cite{Bartal2003} studied an online CA problem and proposed an  $ O(\log(v_{\max}/v_{\min})) $-competitive online algorithm when there are $\Omega(\log(v_{\max}/v_{\min}))$ copies of each item, where the customers' valuation is assumed to be in the range of $ [v_{\min},v_{\max} ] $. Considering the shortcomings of both the limited-supply case and the unlimited-supply case in modelling real-world problems, 
Blum et al. \cite{Blum2011} studied online CAs in a more general case with increasing production cost.  In this case, the supplier can produce additional number of items while paying an increasing marginal cost per copy\footnote{The limited-supply case can be considered a special case of the model in \cite{Blum2011} with a $ 0 $-to-$\infty  $ step function representing the marginal supply cost.}. Blum et al. proposed a pricing scheme called \textit{twice-the-index}, and gave constant competitive ratios for simple marginal cost functions such as linear and lower-degree polynomial. Huang et al. \cite{Huang2015} later studied a similar problem and achieved a tighter competitive ratio. In particular, for polynomial cost functions $ f(y) = y^s $, Huang et al. \cite{Huang2015} proved that an optimal $ s^{\frac{s}{s-1}} $-competitive online mechanism can be designed when each agent buys infinitesimally small units of items. However, the optimal competitive ratio in \cite{Huang2015} was obtained without considering the capacity limits, and thus the results cannot capture how the limited supply affects the optimal online mechanisms. Tan et al. \cite{posted_price_MAMA} later studied this problem with both polynomial supply costs $ f(y) = y^s $ and stringent supply constraints. When the valuation-to-demand ratio of agents (i.e., valuation divided by resource demand) are upper bounded by $ \bar{p} $, Tan et al. \cite{posted_price_MAMA} proved that the optimal competitive ratio $ s^{\frac{s}{s-1}} $ which is firstly characterized by \cite{Huang2015} becomes non-achievable if $ \bar{p} $ exceeds a certain threshold. Moreover, this threshold depends on the supply capacity. In this regard, the results in \cite{Huang2015,posted_price_MAMA} show that both the supply costs and the supply constraints can influence the optimal competitive ratio.

\subsection{Our Contributions and Techniques}
\label{section_contribution}
In this paper, we consider a \textit{significant} generalization of the online resource allocation problem by allowing \textit{arbitrary} strictly-convex supply costs and stringent supply constraints. In this setting, additional units of resources can be supplied at strictly-increasing marginal costs, and the total resource allocated must stay within the capacity limit. To the best of our knowledge, characterizing the optimal competitive ratios in such settings is an open question.  This paper fills this gap and proposes an optimal online mechanism  with the best-possible competitive ratios. Meanwhile, for some special cases when the supply costs are not strictly-convex, e.g., zero and linear  supply cost,  we prove that our design still holds (with some minor modifications) and achieves logarithmic competitive ratios which have been proven as optimal by existing literature.  Therefore, our proposed approach unifies the characterization of optimal competitive ratios in online resource allocation for different setups including zero, linear and strictly-convex supply costs.

Close to our work is the design of online primal-dual algorithms, e.g., \cite{covering_packing, online_routing_FOCS, Buchbinder2009, concave_return, Huang2017, Huang2015, willma, Gupta2013} (in particular, see  \cite{OPD2009} for a comprehensive survey). The key is to construct a dual feasible objective that is close to the primal feasible objective at each round when there is a new arrival of agent, and then use weak duality to bound the performance of the online algorithm.  In this paper, the construction of feasible primal and dual solutions is based on a universal monotone ``pricing function", denoted by $ \phi $,  which is designed by solving an ordinary differential equation (ODE) \cite{ODE1973}. 
Unlike techniques adopted by \cite{willma, Huang2015, concave_return}, our derived ODE contains two boundary conditions imposed by the supply constraints, and thus is essentially a two-point boundary value problem (BVP) \cite{uniqueness_book1993, ODE1973, Perko2001}. We prove that the existence of any $ \alpha $-competitive algorithm is equivalent to the existence of a strictly-increasing solution to the BVP, for which the uniqueness and existence conditions are rigorously proved. To the best of our knowledge, this is the first time that the design and analysis of online algorithms are based on the existence and uniqueness of solutions to BVPs parameterized by the competitive ratios. We believe that our proposed design principle sheds light on new directions for solving other online optimization problems.

\section{Problem Setup}
In this section, we present the problem setup with the resource allocation model, and then introduce the  assumptions  and definitions that will be used in our online mechanism design. 

\subsection{The Model}
\label{section_model}
We consider a single supplier who allocates a single type of resource to a set of requests $ \mathcal{N}=\{1,\cdots,N\} $.  Each request $ n\in\mathcal{N} $ is owned by an agent (which is also indexed by $ n $ for simplicity), and can be represented by a private type  
$ \bm{\theta}_n = \big(r_n,v_n\big) $, 
where  $ r_n $ represents the resource requirement and $ v_n $ denotes the valuation of agent $ n $ if the request is satisfied. In practice, we can interpret $ v_n $  as the \textit{willingness-to-pay} of agent $ n $ for obtaining the required resource. 

Let us define a binary variable $ x_n = \{0,1\} $ for each request $ n\in\mathcal{N} $, and assume $ x_n = 1 $ if request $ n $ is satisfied  and $ x_n = 0 $ otherwise.  We assume that agents have quasi-linear utilities \cite{AGT}, i.e., the utility of agent $ n $ is denoted by 
$ U_n = v_nx_n - \pi_n $, where $ \pi_n $ denotes the payment made by agent $ n $.  Intuitively, agents whose requests are not satisfied make zero payment, i.e., $ \pi_n = 0 $ if $ x_n = 0 $. The supplier collects payment $ \pi_n $ from agent $ n $, and pays a total supply cost of $ f(\sum_n r_nx_n) $, where $ \sum_n r_nx_n $ denotes the total resource allocated and $ f $ represents the supply cost function. Therefore, the utility of the supplier (i.e., the profit) can be denoted by 
$ U_s = \sum_{n\in\mathcal{N}}\pi_n - f(\sum_n r_nx_n) $. We consider limited supply with a stringent capacity limit. Without loss of generality, we normalize the capacity limit to be 1, and thus $ \sum_n r_nx_n \leq 1 $. Meanwhile, the requirements $ \{r_n\}_{\forall n} $ denote the proportions of the normalized capacity limit accordingly. Our mechanism design relies on some properties of the cost function $ f $, which is discussed later in Section \ref{section_assumptions}.

We focus on social welfare maximization in an online setting, where agents arrive one-by-one in a sequential and arbitrary manner, breaking ties arbitrarily. We denote the sequence of agent arrivals by $ \mathcal{A} \triangleq  (\bm{\theta}_1,\bm{\theta}_2,\cdots,\bm{\theta}_N)$,
where we assume without loss of generality that the arrival times  are in ascending order. In the following we refer to  $ \mathcal{A} $  as the arrival instance.
If we assume a complete knowledge of  $ \mathcal{A} $, then the social welfare maximization in the offline setting can be written as follows:
\begin{subequations}\label{SWM}
	\begin{alignat}{3}
	& \underset{\{x_n\}_{\forall n}}{\textsf{maximize}}\qquad        & &  \sum_{n\in\mathcal{N}} v_n x_n - f\left(\sum_{n\in\mathcal{N}} r_nx_n\right)\\
	& \textsf{subject to} & & \sum_{n\in\mathcal{N}} r_nx_n \leq 1,
	\label{total_load_original} \\ 
	& & &  x_n = \{0,1\},\forall n\in\mathcal{N}. \label{binary_x_agent_n}
	\end{alignat}
\end{subequations}
In Problem \eqref{SWM}, the objective is the aggregate utilities of all the agents and the supplier, where the payment terms cancel out.

\begin{remark}\label{model_remark}
	If there is no supply cost, i.e., $ f = 0 $, then Problem \eqref{SWM} is a standard 0-1 knapsack problem \cite{knapsack2008}. Meanwhile, if we consider allocation of bundles of resources with a supply cost function for each type of resource, Problem \eqref{SWM} can be extended to model the standard online combinatorial auction problem \cite{AGT, Blum2011}. Therefore, the studied model is a generalization of a variety of classic  online resource allocation problems. In this paper, we mainly focus on the basic model presented by Problem \eqref{SWM}. For discussions of extending our model to consider multiple types of resources, as well as generalizing Problem \eqref{SWM} to model online resource allocation with multiple time slots (e.g., job scheduling in cloud computing), please refer to  Section \ref{extension_multiple_resources}.
\end{remark}

\subsection{Assumptions}
\label{section_assumptions}
We make the following  assumptions throughout the paper. 

\begin{assumption}\label{Assumption_PUV}
	For each request $ n\in\mathcal{N} $, the \textit{valuation density}, defined as $ v_n/r_n $, is lower and upper bounded as follows: 
	\begin{align}
	\ubar{p} \leq \frac{v_n}{r_n} \leq  \bar{p}, \forall n\in\mathcal{N}. 
	\end{align} 
\end{assumption}
Assumption \ref{Assumption_PUV} is similar to the one in \cite{knapsack2008, willma, Bartal2003}. Specifically, one can interpret $ \ubar{p} $ as the lowest selling price per unit of resource, which is set by the supplier in advance. Thus, the request of any agent with a valuation density lower than $ \ubar{p} $ will automatically be rejected. In comparison, the upper bound $ \bar{p} $ can be considered an inherent attribute of the arrival instance with rational agents, and in reality it exists naturally. 

\begin{assumption}\label{Assumption_d_n}
	The resource requirement of each request is very small compared to the total capacity limit, i.e., $ r_n\ll 1, \forall  n\in\mathcal{N} $. 
\end{assumption}

Assumption \ref{Assumption_d_n} is common in designing online algorithms, e.g., \cite{Adwords2007, knapsack2008, Huang2015, concave_return}, since it allows us to focus on the infinitesimal nature of our problem with mathematical convenience. Meanwhile, Assumption \ref{Assumption_d_n} is reasonable in many real-world large-scale systems. 

\begin{definition}[Information Setup]
	We define all the information known to the supplier a priori as a setup $ \mathcal{S} $,  denoted by
	\begin{align}
	\mathcal{S} \triangleq \{f,\ubar{p},\bar{p}\}.
	\end{align}
\end{definition}

In the following, we say $ \mathcal{S} $ is a \textit{nice setup} if $ f $ is monotonically non-decreasing with zero startup cost (i.e, $ f(0)=0 $) and satisfies $ f'(0) <\ubar{p}\leq \bar{p} $. Meanwhile, we refer to a nice setup $ \mathcal{S} $ as a \textit{convex setup} if $ f $ is also \textit{strictly-convex and differentiable}. Given a nice setup $ \mathcal{S} $, the supplier only knows the information given  by $ \mathcal{S} $, and does not know anything about the arrival instance $ \mathcal{A} $ such as the total number of requests (i.e., $ N $). 

\subsection{Definitions}
\label{competitive_analysis}
Below we give some definitions regarding incentive compatibility, online mechanisms, and competitive ratios.

(\textbf{Incentive Compatibility}) We consider the design of direct revelation mechanisms under a
strategic setting, where each agent $ n\in\mathcal{N} $ participates by declaring the type $ \hat{\bm{\theta}}_n = (\hat{r}_n,\hat{v}_n)$ at time $ \hat{a}_n $. Since agents are self-interested, the reported preference $ \hat{\bm{\theta}}_n $ of agent $ n\in\mathcal{N} $ may or may not be her true preference $ \bm{\theta}_n $. A mechanism is \textit{incentive compatible  (IC)} if all agents achieve the best outcomes (i.e., the maximum utilities) by acting according to their true preferences.   

(\textbf{Online Mechanism}) Given a convex setup $ \mathcal{S} $, we aim to design an IC online mechanism to (approximately) maximize the social welfare. In the online setting, agents arrive one-by-one in a sequential order (where ties are broken arbitrarily).  At each round when there is a new arrival of agent $ n\in\mathcal{N} $, the supplier needs to make two types of irrevocable decisions as follows: i) the \textit{allocation rule}, namely the decisions about whether request $ n $ should be satisfied  or not (i.e., $ x_n $); and ii) the \textit{payment rule}, namely, how much money should be charged from agent $ n $ if the request is satisfied. The resource allocation rule and the payment rule constitute an \textit{online mechanism}. In the strategic setting, agents may deliberately misreport their type information to be better off, and thus the supplier needs to carefully design the payment rule such that agents are well-incentivized to follow their true preferences\footnote{In a non-strategic environment, only an allocation rule is needed,  and thus in this case an online mechanism is simply an online algorithm.}. 

(\textbf{Competitive Ratio}) The performance of an online mechanism can be quantified by the standard competitive analysis framework \cite{Borodin1998}. Given an arrival instance $ \mathcal{A} $, let us denote the optimal social welfare by  $ S_{\textsf{offline}}(\mathcal{A}) $, which is the optimal objective value of Problem \eqref{SWM} in the offline setting. Let  $ S_{\textsf{online}}(\mathcal{A}) $ denote the social welfare achieved by an online mechanism based on the knowledge of $ \mathcal{S} $ only. 
The competitive ratio of an online mechanism is defined as 
\begin{equation}\label{definition_of_alpha}
\alpha \triangleq \max_{\normalfont\text{all possible }\mathcal{A}} \frac{ S_{\normalfont\textsf{offline}}(\mathcal{A})}{S_{\normalfont\textsf{online}}(\mathcal{A})}, 
\end{equation}
where $ \alpha\geq 1 $ and the closer to 1 the better. Note that in the online setting the supplier does not known any future information except $ \mathcal{S} $, and thus the competitive ratio $ \alpha $ should depend on $ \mathcal{S} $ only. 

Our target is to design an IC online mechanism such that $ S_{\text{online}}(\mathcal{A}) $ is as close to $ S_{\textsf{offline}}(\mathcal{A}) $ as possible for all possible $ \mathcal{A} $'s, namely, an online mechanism with an competitive ratio that is as small as possible. To define how small the competitive ratio could possible be, we give the following definition regarding the optimal competitive ratio for a given nice setup $ \mathcal{S} $.  

\begin{definition}[Optimal Competitive Ratio]\label{optimal_competitive_ratio}
	A competitive ratio is optimal if no other online algorithms can achieve a smaller one under Assumption 1 and Assumption 2. In particular, given a nice setup $ \mathcal{S} $, the optimal competitive ratio $ \alpha_*(\mathcal{S}) $ is defined as 
	\begin{equation}\label{alpha_*_S}
	\alpha_*(\mathcal{S}) \triangleq \inf \max_{\normalfont\text{all possible }\mathcal{A}} \frac{ S_{\normalfont\textsf{offline}}(\mathcal{A})}{S_{\normalfont\textsf{online}}(\mathcal{A})},
	\end{equation}
	where the $\inf$ operator is taken w.r.t. all possible online algorithms.
\end{definition}

\section{Online Mechanism Design}
This section presents our proposed online mechanism. We start by giving some preliminaries and notations regarding the primal and dual of Problem \eqref{SWM}, and then describe the design principle and design procedures of our online mechanism. 

\begin{table}[htb]
	\caption{Notations}
	\centering
	\begin{tabular}{c|c}
		\hline
		Symbol  & Description \\
		\hline
		
		\rowcolor{Gray}
		$ v_n $ & value of agent $ n $ \\
		
		$ x_n $ & binary decision variable of agent $ n $\\
		
		\rowcolor{Gray}
		$ r_n $ & resource requirement of agent $ n $ \\
		
		$ f $   & supply cost function
		\\
		
		\rowcolor{Gray}
		
		$ \ubar{p} $ ($ \bar{p} $)    & lower (upper) bound of $ v_n/r_n $ \\
	
		$ \ubar{c} $ ($ \bar{c} $) & minimum (maximum) marginal cost\\
		
		\rowcolor{Gray}
		$ \ubar{\rho} $ ($ \bar{\rho} $)   & maximum resource utilization when $ \frac{v_n}{r_n} = \ubar{p} $ ($ \bar{p} $), $\forall n $\\
		
		$ \phi(y) $  & pricing function with three segments\\
		
		\rowcolor{Gray}
		$ \varphi(y) $  & increasing segment of $ \phi(y) $ when $ y\in [\omega,\bar{\rho}] $\\

		$ \omega $ & resource utilization threshold so that $ \phi(\omega) = \ubar{c} $\\
		
		\rowcolor{Gray}
		$ u $   &  resource utilization threshold so that $ \phi(u) = \bar{c} $\\
		
		$ \hat{\pi}_n $  & final payment made by agent $ n $\\
		
		\rowcolor{Gray}
		$ \hat{p}_n $ & posted price for agent $ n+1 $\\
		
		$ \hat{y}_n $ & total resource utilization after agent $ n $\\
		\hline
	\end{tabular}
	\label{table}
\end{table}

\subsection{Preliminaries}
\label{primal_dual_offline_setting}
We first introduce some notations to help derive the dual of Problem \eqref{SWM}. Given a convex setup $ \mathcal{S} $, we define $ \ubar{c} $ and $ \bar{c} $ as follows:
\begin{equation}\label{c_0_c_1_rho_0_rho_1}
\ubar{c} \triangleq f'(0), \qquad  \bar{c} \triangleq f'(1),
\end{equation}
where $ \ubar{c} $  and $ \bar{c} $ denote the minimum and maximum marginal cost, respectively. Based on the capacity limit, we define the \textit{extended cost function} $ \bar{f}(y) $ as follows:
\begin{equation}\label{cost_function}
\begin{aligned}
\bar{f}(y) = 
\begin{cases}
f(y) & \text{if } y\in [0, 1],\\
+\infty  &  \text{if } y\in (1,+\infty),
\end{cases}
\end{aligned}
\end{equation}
which extends the domain of $ f(y) $ to all $ y\in [0,+\infty) $.
Based on the extended cost function $ \bar{f} $, we define $ F_p(y) $ as follows:
\begin{equation}\label{profit_function}
F_p(y) \triangleq py - \bar{f}(y), y\in [0,+\infty),
\end{equation}
which represents the profit of selling $ y $ units of resources at price $ p $, where $ py $ represents the revenue and $ \bar{f}(y) $ represents the supply cost. In the following, $ F_p(y) $ is referred to as the \textit{profit function}.  

(\textbf{Relaxed Primal Problem}) Our online mechanism is designed based on a principled primal-dual analysis of Problem \eqref{SWM}. We consider the following relaxed welfare maximization problem
\begin{subequations}\label{primal_problem}
	\begin{alignat}{3}
	& \underset{\bm{x},y}{\text{maximize}}\qquad         & &  \sum_{n\in\mathcal{N}} v_n x_n - \bar{f}(y) \\
	&\text{subject to} & &   \sum_{n\in\mathcal{N}}  r_nx_n \leq  y, & & \qquad  (p)\label{tota_load_relaxation}\\
	& & &  x_n \leq 1,\forall n\in\mathcal{N}, & & \qquad  (\gamma_n)\\
	& & & \bm{x} \geq \bm{0}, y\geq 0, \label{lower_box_constraint_agent_n_relaxation}
	\end{alignat}
\end{subequations}
where $ \bm{x} = \{x_n\}_{\forall n} $. In Problem \eqref{primal_problem}, $ (p,\gamma_n)$ are the Lagrange multipliers associated with the corresponding constraint. 

When $ f $ is non-decreasing, Problem \eqref{primal_problem} is equivalent to Problem \eqref{SWM} except the relaxation of $ \{x_n\}_{\forall n} $. The reason is as follows. First, compared to the original cost function $ f $ with the capacity limit constraint \eqref{total_load_original}, the introduction of $ \bar{f} $ in Problem \eqref{primal_problem} is an equivalent transformation. Second, the relaxation  of the equality constraint  \eqref{total_load_original} to an inequality one  in Eq. \eqref{tota_load_relaxation} is lossless since the inequality will always be binding when  $ f $ is non-decreasing. 

(\textbf{Dual Problem}) Problem \eqref{primal_problem} is a continuous optimization problem, whose dual can be expressed  as follows:
\begin{subequations}\label{dual_problem}
	\begin{alignat}{3}\label{objective_dual}
	& \underset{p, \bm{\gamma}}{\textsf{minimize}}\qquad      & & \sum_{n\in\mathcal{N}}\gamma_n +  h\left(p\right)\\
	&\textsf{subject to} & &\ \gamma_n\geq  v_n - pr_n, \forall n\in\mathcal{N}, \label{dual_constraint_gamma}\\
	&    & &  p\geq 0, \boldsymbol{\gamma}\geq \mathbf{0},\label{dual_constraint_variables}
	\end{alignat}
\end{subequations}
where $ \bm{\gamma} = \{\gamma_n\}_{\forall n}$. We interpret $ p $ as the price per unit of resource, and $ \gamma_n $ as the utility of agent $ n $.  In Eq. \eqref{objective_dual},  $ h(p) $ is given by
\begin{equation}\label{def_h_p}
h(p) = \max_{y \geq 0}\ p y    - \bar{f}(y) = \max_{y\geq 0}\ F_p(y),
\end{equation}
which can be written as follows:
\begin{equation}\label{dual_bar_f_star}
\begin{aligned}
h(p) 
= 
\begin{cases}
F_p\left(f'^{-1}(p)\right) & \text{if } p\in [\ubar{c},\bar{c}],\\
F_p\big(1\big) & \text{if } p\in (\bar{c},+\infty).
\end{cases}
\end{aligned}
\end{equation}

The above function $ h(p) $ is known as the convex conjugate of $ \bar{f} $ \cite{conjugate_book_2006}.  An economic interpretation of Eq. \eqref{def_h_p} is that $ h(p) $ represents the optimal profit when the selling price is $ p $, where $ p y $ represents the revenue and $ \bar{f}(y) $ represents the supply cost. In this regard, the dual objective in Eq. \eqref{objective_dual} is the aggregate utilities of all the agents plus the utility of the supplier (i.e., the profit). In comparison, the primal objective in Eq. \eqref{primal_problem} is the aggregate values of all the satisfied requests minus the supply cost. Both of them represent the social welfare of the system.

Given an arrival instance $ \mathcal{A} $, let us denote the optimal primal objective of Problem \eqref{primal_problem} by $ P_*(\mathcal{A}) $. Similarly, denote by $ D_*(\mathcal{A}) $ the optimal dual objective of Problem \eqref{dual_problem}. Then, we have
\begin{equation}\label{S_offline_P_D}
S_{\textsf{offline}}(\mathcal{A})\leq P_*(\mathcal{A})\leq D_*(\mathcal{A}), 
\end{equation}
where the first inequality is due to the relaxation of $ \{x_n\}_{\forall n} $ and the second one is because of weak duality\footnote{Note that when $ f $ is convex, strong duality holds here, i.e., $ P_*(\mathcal{A})= D_*(\mathcal{A}) $. However, throughout the paper we do not need the strong duality.}.

\subsection{Design Principles}
In the online setting, we need to process agents in a sequential manner. For a given arrival instance $ \mathcal{A} $, after processing the $ n $-th request, we denote the objective values of Problem \eqref{primal_problem} and Problem \eqref{dual_problem} by $ P_n(\mathcal{A}) $ and $ D_n(\mathcal{A}) $, respectively. For simplicity, we drop the parenthesis and simply write $ P_n $ and $ D_n $ hereinafter. Meanwhile, we say that $ P_n\ (D_n) $ is feasible if $ P_n\ (D_n) $ is a feasible objective value of Problem \eqref{SWM} (Problem \eqref{dual_problem}) after processing the $ n $-th request. Proposition \ref{OPD_principle} below shows that if the sequences of $ \{P_n\}_{\forall n} $ and $ \{D_n\}_{\forall n} $ are feasible and satisfy a group of inequalities parameterized by $ \alpha $, then the corresponding online algorithm is guaranteed to be $ \alpha $-competitive. 
\begin{proposition}\label{OPD_principle}
	An online algorithm is $ \alpha $-competitive if the following two conditions are satisfied, namely, i) the sequences of $ \{P_n\}_{\forall n} $ and $ \{D_n\}_{\forall n} $ are feasible, and ii) there exists an index $ k\in \mathcal{N} $ such that the following inequalities hold
	\begin{equation}\label{two_groups_of_inequalities}
		\begin{aligned}
			\begin{cases}
				P_k\geq \frac{1}{\alpha} D_k,\\
				P_n - P_{n-1} \geq \frac{1}{\alpha}\left(D_n- D_{n-1}\right), \forall n\in \{k+1,\cdots,N\}.
			\end{cases}
		\end{aligned} 
	\end{equation}
\end{proposition}
\begin{proof}
	The feasibility of $ \{P_n\}_{\forall n} $ is trivial since any online algorithm must first produce a feasible solution to the original problem in Eq. \eqref{SWM}. We next show how the feasibility of $ \{D_n\}_{\forall n} $ and the inequalities in Eq. \eqref{two_groups_of_inequalities} lead to an  $ \alpha $-competitive online algorithm. Note that it suffices to prove $ P_N\geq \frac{1}{\alpha} D_N $ since 
	\begin{align}\label{inequalities_OPD}
	S_{\textsf{online}} = P_N\geq \frac{1}{\alpha} D_N \overset{(i)}{\geq} \frac{1}{\alpha} D_* \overset{(ii)}{\geq} \frac{1}{\alpha} P_* \overset{(iii)}{\geq} \frac{1}{\alpha} S_{\textsf{offline}},
	\end{align}
	where  inequalities  $(ii)$ and  $ (iii) $ directly follow Eq. \eqref{S_offline_P_D}, and inequality $ (i) $ holds since $ D_N $ is a feasible objective value of Problem \eqref{dual_problem}, while $ D_* $ is the optimum (minimum). 
	
	Suppose there exists an index $ k\in\mathcal{N} $ such that the second inequality in Eq. \eqref{two_groups_of_inequalities} hold for all $ n\in \{k+1,\cdots,N\} $, then
	\begin{equation}\nonumber
	\begin{aligned}
	P_N-P_k =  \sum_{n=k+1}^{N}(P_n-P_{n-1}) \geq  \frac{1}{\alpha} \sum_{n=k+1}^{N}(D_n-D_{n-1}) = \frac{1}{\alpha}(D_N- D_k),
	\end{aligned}
	\end{equation}
	which leads to $ P_N\geq \frac{1}{\alpha} D_N $ after substituting $ P_k\geq \frac{1}{\alpha}D_k $ into the above equation.
	We thus complete the proof. 
\end{proof}

We refer to the first inequality in Eq. \eqref{two_groups_of_inequalities} as the \textit{initial inequality}, and the second one in Eq. \eqref{two_groups_of_inequalities} as the \textit{incremental inequality}. Note that when $ P_0 = D_0 = 0 $ and $ k = 1 $,  the two types of inequalities can be combined.  Proposition \ref{OPD_principle} in this case simply follows the standard online primal-dual approach \cite{OPD2009}. 
In this regard, Proposition \ref{OPD_principle} is a more general principle for designing online algorithms.

\subsection{Posted Price Mechanism} \label{section_PM}
Our online mechanism is designed based on  Proposition \ref{OPD_principle}. The key idea is to construct a set of feasible primal and dual solutions  at each round when there is a new arrival  of request $ n\in\mathcal{N} $, and then guarantee that the resulting sequences of $ \{P_n\}_{\forall n} $ and $ \{D_n\}_{\forall n} $ satisfy the inequalities in Eq. \eqref{two_groups_of_inequalities}. 

(\textbf{A Two-Step Design}) Note that the dual variables $ \{\gamma_n\}_{\forall n} $ in Problem \eqref{dual_problem} are associated with each individual agent,   while $ p $ is a global one couples all the agents. Once the final price $ p $ is known to the supplier, we can easily decouple different agents and design feasible primal and dual solutions for each individual request $ n\in\mathcal{N} $. However, when there is no future information, it is impossible to know the exact value of $  p $ a priori. Our idea is to adopt a two-step design procedures as follows:
\begin{itemize}
	\item \textbf{Step-1}: design a trajectory of $ \{\hat{p}_n\}_{\forall n}$, where $ \hat{p}_n $ denotes the supplier's prediction of the final price $ p $ after processing  agent $ n $. 
	\item \textbf{Step-2}: when there is a new  arrival of request $ n\in\mathcal{N} $,  based on $ \hat{p}_{n-1}$, perform the following decision-making:
	\begin{itemize}
		\item Set the dual variable $ \hat{\gamma}_n $ by
		\begin{equation}\label{design_of_gamma}
		\hat{\gamma}_n = \max\big\{v_n - \hat{p}_{n-1} r_n, 0\big\},\forall n\in\mathcal{N}.
		\end{equation}
		\item Set the primal variable $ \hat{x}_n $ by:
		\begin{align*}
		\hat{x}_n = 
		\begin{cases}
		0  &\text{if } v_n - \hat{p}_{n-1} r_n< 0,\\
		1 &\text{if } v_n - \hat{p}_{n-1} r_n\geq 0 \text{ and } \hat{y}_{n-1} +  r_n \leq 1,
		\end{cases}
		\end{align*}
		where $ \hat{y}_{n-1} $ denotes the total resource utilization after processing agent $ n-1 $. Intuitively, we have $ \hat{y}_0 = 0 $.
		
		\item Set the final payment $ \hat{\pi}_n $ for agent $ n $ by
		\begin{equation}\label{payment}
		\hat{\pi}_n = \hat{p}_{n-1} r_n \hat{x}_n,
		\end{equation}
		
		\item Update the total resource utilization by 
		\begin{equation}\label{update_of_omega}
		\hat{y}_n = \hat{y}_{n-1} +  r_n\hat{x}_n.
		\end{equation}
	\end{itemize}
\end{itemize}

Based on the above two-step design procedure, the terminal value of the total resource utilization is $ \hat{y}_N  $, and the terminal value of the predicted final price is $ \hat{p}_N $. Together with $ \{\hat{x}_n\}_{\forall n} $ and $ \{\hat{\gamma}_n\}_{\forall n} $, these variables constitute a complete set of online primal and dual solutions, which are denoted by $ \mathcal{V}_P $ and $ \mathcal{V}_D $ as follows:
\begin{equation*}
\mathcal{V}_P \triangleq \left(\{\hat{x}_n\}_{\forall n}, \hat{y}_N\right),  \mathcal{V}_D \triangleq  \big(\hat{p}_N, \{ \hat{\gamma}_n\}_{\forall n}\big).
\end{equation*}
Note that to differentiate between offline and online settings,  we place a hat on top of variables that denote the decisions made online. 

(\textbf{Pricing Function in Step-1}) 
To enable an online implementation, the price predictions in \textbf{Step-1} must be performed based on causal information only.  One natural way of designing such price predictions (without future information) is to relate $ \hat{p}_n $  to the current total resource utilization as follows:
\begin{equation}\label{pricing_function}
\hat{p}_n =  \phi\left(\hat{y}_n\right), \forall n\in\mathcal{N},
\end{equation}
where  $ \hat{y}_n $ denotes the total resource utilization \textit{after} allocating the required resources to  request $ n $. Based on our interpretation of $ \hat{p}_n $, $ \phi $ is referred to as the \textit{pricing function} hereinafter.  Eq. \eqref{pricing_function} indicates that our prediction of the final price will be updated whenever the total resource utilization changes.   

(\textbf{Posted Price Mechanism:} $ \textsf{PM}_\phi $) 
The above two-step design indicates that the pricing function $ \phi $ plays an important role in influencing $ \mathcal{V}_P $ and $ \mathcal{V}_D $, as well as the sequences of $ \{P_n\}_{\forall n} $ and $ \{D_n\}_{\forall n} $. Therefore,  the design of $ \phi $ is directly related to the inequalities in Proposition \ref{OPD_principle}, and thus determines the competitive ratio of the online mechanism described above. The techniques of how to design $ \phi $ in \textbf{Step-1} constitute the major results of this paper, and the details are deferred to Section \ref{section_main_results}. In the following, we temporarily assume that the pricing function $ \phi $ is given and summarize our proposed online mechanism in Algorithm \ref{PriMe}, dubbed $ \textsf{PM}_\phi $.  

An interesting observation about $ \textsf{PM}_\phi $ is that it can be implemented in a posted price manner. Unlike auctions \cite{CA_PNAS, CA_survey, AGT}, the supplier running posted price \cite{posted_prices_1993,posted_prices_2018} simply publishes the selling price (i.e., line \ref{publish_price}) and does not collect any information from the agents. The decisions are made by each individual agent in the manner of \textit{take-it-or-leave-it} (i.e., line \ref{if_utility_negative}-line \ref{end_if}).  By virtue of posted-price mechanisms \cite{posted_price_EC_2017, posted_prices_1993, posted_prices_2018}, $ \textsf{PM}_\phi $ is IC, privacy-preserving and computationally-efficient. Therefore, we argue that this is an extra advantage of our design, although we do not commit to posted price mechanisms a priori.

\begin{algorithm}[htb]
	\caption{Posted Price Mechanism ($ \textsf{PM}_{\phi} $)}\label{online_mechanism}	
	\begin{algorithmic}[1]
		\STATE \textbf{Inputs:} A given setup $ \mathcal{S} = \{f,\ubar{p},\bar{p}\} $ and $ \phi$. 
		
		\STATE \textbf{Initialize}: $ \hat{y}_0 = 0 $ and $ \hat{p}_0 = \phi(\hat{y}_0)$.  
		
		\WHILE{a new agent $ n $ arrives}
		
		\STATE Supplier publishes the price $ \hat{p}_{n-1} $.\label{publish_price}
		
		\IF {$ v_n -  \hat{p}_{n-1} r_n < 0 $}\label{if_utility_negative} 
		
		\STATE Agent $ n $ leaves (i.e., set $ \hat{x}_n =0 $) \label{rejection_1}
		
		\ELSIF{$\hat{y}_{n-1}+ r_n >1$}
		\STATE Request $ n $ is rejected (i.e., set $ \hat{x}_n =0 $)\label{rejection_2}
		
		\ELSE 
		\STATE Request $ n $ is satisfied (i.e., set $ \hat{x}_n = 1 $)\label{accepted}
		
		\STATE Collect  the payment $ \hat{\pi}_n $ by Eq. \eqref{payment}.\label{payment_line}
		
		\STATE Update the total resource utilization by Eq. \eqref{update_of_omega}.
		
		\STATE Update the price by $
		\hat{p}_n = \phi(\hat{y}_n)$. \label{current_price}
		
		\ENDIF \label{end_if}
		\ENDWHILE
	\end{algorithmic}
	\label{PriMe}
\end{algorithm}

(\textbf{Feasibility and Rationality of} $ \textsf{PM}_\phi $) 
As shown by Proposition \ref{OPD_principle}, the feasibility of $ \mathcal{V}_P $ and $ \mathcal{V}_D $ is crucial to prove the competitive ratio of the designed online mechanism.  In Proposition \ref{feasibility} below, we show that our above design of $ \mathcal{V}_P $ and $ \mathcal{V}_D $ is feasible as long as $ \phi $ is monotone. 
\begin{proposition}\label{feasibility}
	The primal solutions in  $ \mathcal{V}_P $ are always feasible to Problem \eqref{SWM} and  Problem \eqref{primal_problem}. The dual solutions in $ \mathcal{V}_D $ are feasible to Problem \eqref{dual_problem}  as long as $ \phi $ is monotonically non-decreasing. 
\end{proposition}
\begin{proof}
	It is obvious that the design of $ \left(\{\hat{x}_n\}_{\forall n},  \hat{y}_N\right) $ in $ \mathcal{V}_P $ are feasible to Problem \eqref{SWM}. For each agent $ n\in\mathcal{N} $, our design of $ \hat{\gamma}_n $ in Eq. \eqref{design_of_gamma} indicates that 
	$ \hat{\gamma}_n \geq v_n - \hat{p}_{n-1} r_n $. 
	If $ \phi $ is non-decreasing, we have $
	\hat{p}_{n-1} \leq \hat{p}_n \leq \hat{p}_N, \forall n\in\mathcal{N}$.
	Therefore,  we have  
	$ \hat{\gamma}_n \geq v_n - \hat{p}_N r_n $ holds for all $ n\in\mathcal{N} $. We thus complete the proof.
\end{proof}

Before leaving this section, it is worth mentioning that in economics and game theory \cite{economics1995, AGT}, our design of a dual feasible $ \hat{\gamma}_n $ in Eq. \eqref{design_of_gamma} is known to guarantee the \textit{individual rationality} of the mechanism, namely, no agent suffers from negative utility by participating in the mechanism (i.e., line \ref{rejection_1} in Algorithm \ref{PriMe}). Note that the individual rationality is also equivalent to saying that no agent is forced to participate in the mechanism.

\section{Main Results and Techniques}
\label{section_main_results}
Proposition \ref{feasibility} shows that our previous design of $ \mathcal{V}_P $ and $ \mathcal{V}_D $ is feasible if the pricing function $ \phi $ is monotone.  In this section, we present our major results regarding \textbf{Step-1}, namely, the design of a monotone pricing function $ \phi $ so that $ \textsf{PM}_\phi $ achieves a competitive performance in social welfare. 

\subsection{Sufficient Conditions}
To aid our following presentation, let us define $ \ubar{\rho} $ and $ \bar{\rho} $ as follows:
\begin{equation}\label{def_of_rho}
\ubar{\rho} \triangleq \arg\max_{y \geq 0}\ \ubar{p} y    - \bar{f}(y),\ \bar{\rho} \triangleq \arg\max_{y \geq 0}\ \bar{p}y   - \bar{f}(y).
\end{equation}
For a convex setup when $ \bar{f}(y) $ is strictly-convex and differentiable in $y\in [0,1]  $, $ \ubar{\rho} $ and $ \bar{\rho} $ can be respectively written as: 
\begin{equation*}
\begin{aligned}
\ubar{\rho} = 
\begin{cases}
f'^{-1}\big(\ubar{p}\big)  &\text{if } \ubar{p}\in (\ubar{c},\bar{c}),\\
1  &\text{if } \ubar{p}\geq \bar{c},
\end{cases}
; \
\bar{\rho} = 
\begin{cases}
f'^{-1}\left(\bar{p}\right)  &\text{if } \bar{p}\in (\ubar{c},\bar{c}),\\
1  &\text{if } \bar{p}\geq \bar{c}. 
\end{cases}
\end{aligned}
\end{equation*}
The definitions of $ \ubar{\rho} $ and $ \bar{\rho} $ can be interpreted as follows. Suppose there are infinitely-many identical agents whose valuation densities are all $ \ubar{p} \ (\bar{p}) $, then the maximum (optimal) resource utilization  level is $ \ubar{\rho}\ (\bar{\rho})$. Both $ \ubar{\rho} $ and $ \bar{\rho} $ are capped by 1 due to the capacity limit. 

Based on the above definitions of $ \ubar{\rho} $ and $ \bar{\rho} $, we next give Theorem \ref{sufficiency} which shows the sufficient conditions for $ \phi $ so that $ \textsf{PM}_\phi $ is IC and achieves a bounded competitive ratio.
\begin{theorem}\label{sufficiency}
	Given a convex setup $ \mathcal{S} $, $\normalfont \textsf{PM}_\phi $ is IC and $ \alpha $-competitive  if   $ \phi $ is given by
	\begin{equation}\label{sufficient_phi}
	\begin{aligned}
	\phi(y) = 
	\begin{cases}
	\ubar{p} &\text{if } y\in [0,\omega),\\
	\varphi(y)  &\text{if } y\in [\omega, \bar{\rho}],\\
	+\infty &\text{if } y\in (\bar{\rho},+\infty),
	\end{cases}
	\end{aligned}
	\end{equation}
	where $ \omega $ is a resource utilization threshold that satisfies
	\begin{equation}\label{flat_sufficiency_omega}
	F_{\ubar{p}}(\omega)   \geq    \frac{1}{\alpha}  h\big(\ubar{p}\big) \text{ and }
	0\leq \omega \leq \ubar{\rho},
	\end{equation} 
	and $ \varphi(y) $ is  an increasing function that satisfies
	\begin{equation}\label{ODI_principle_sufficiency}
	\begin{aligned}
	\begin{cases}
	\varphi'(y) \leq   \alpha \cdot\frac{\varphi(y) - f'(y)}{h'\left(\varphi(y)\right)}, y\in (\omega,\bar{\rho});\\
	\varphi(\omega) = \ubar{p}, \varphi(\bar{\rho})\geq \bar{p}.
	\end{cases}
	\end{aligned}
	\end{equation}
	In Eq. \eqref{flat_sufficiency_omega}, $ F_{\ubar{p}} $ is the profit function defined in Eq. \eqref{profit_function} and $ h $ is given by Eq. \eqref{def_h_p}. In Eq. \eqref{ODI_principle_sufficiency}, $ h' $ represents the derivative of $ h $. 
\end{theorem}

Theorem  \ref{sufficiency} shows that  $ \textsf{PM}_\phi $ with any pricing function given by Eq. \eqref{sufficient_phi} is IC and $ \alpha $-competitive, provided that $ \omega $ and $ \varphi $ satisfy certain conditions.  The sufficient conditions in Theorem \ref{sufficiency} are derived based on Proposition \ref{OPD_principle}. In particular, Eq. \eqref{flat_sufficiency_omega} and Eq. \eqref{ODI_principle_sufficiency} correspond to the initial inequality and the incremental inequality in Eq. \eqref{two_groups_of_inequalities}, respectively. A rigorous proof is given in Appendix \ref{proof_sufficiency_alpha_competitive}, while discussions of intuitions are given in Section \ref{proof_of_main_results_sufficiency}.

We note that a pricing function $ \phi $ given by Eq. \eqref{sufficient_phi} consists of three segments, namely, the \textit{flat-segment} $ [0,\omega] $, the \textit{increasing-segment} $ [\omega,\bar{\rho}] $, and the \textit{infinite-segment} $ (\bar{\rho},+\infty) $. Since $ \omega $ is a resource utilization threshold that  separates the first two segments, and plays a critical role in shaping the curvature of $ \phi $, we refer to $ \omega $ as the \textit{critical threshold} hereinafter.  Recall that the valuation densities of all the agents are lower bounded by $ \ubar{p} $, and thus the incoming requests will always be satisfied  when the total resource utilization is below $ \omega $, regardless of their valuations. 

\subsection{Necessary Conditions}
\label{section_necessary_condition}
An interesting result proved by this paper is as follows: existence of a pricing function given by Eq. \eqref{sufficient_phi} is not only sufficient to guarantee a bounded competitive ratio for $ \textsf{PM}_\phi $, but also necessary to the existence of any $ \alpha $-competitive online algorithm. The result is given by the following Theorem \ref{necessity}.  
\begin{theorem}\label{necessity}
Given a convex setup $ \mathcal{S} $, if there exists an $ \alpha $-competitive online algorithm, then there must exist a critical threshold $\omega $ which satisfies  Eq. \eqref{flat_sufficiency_omega} so that the following claims hold simultaneously: 
\begin{itemize}
	\item There exists a case when the total resource utilization is $ \omega $ and all the accepted agents have the same valuation density $ \ubar{p} $.
	\item There exists a strictly-increasing function $ \psi(p) $ that satisfies
	\begin{equation}\label{BVP_necessary_psi}
	\begin{aligned}
	\begin{cases}
	\psi'(p) = \frac{1}{\alpha}\cdot \frac{h'(p)}{p - f'(\psi(p))}, p\in (\ubar{p},\bar{p}),\\
	\psi(\ubar{p}) = \omega, \psi(\bar{p})\leq \bar{\rho}.
	\end{cases}
	\end{aligned}
	\end{equation}
	\item There exists a strictly-increasing function $ \varphi(y) $ that satisfies
	\begin{equation}\label{BVP_necessary_varphi}
	\begin{aligned}
	\normalfont\textsf{BVP}(\omega,\alpha)
	\begin{cases}
	\varphi'(y) =  \alpha \cdot\frac{\varphi(y) - f'(y)}{h'\left(\varphi(y)\right)}, y\in (\omega,\bar{\rho}),\\
	\varphi(\omega) = \ubar{p}, \varphi(\bar{\rho})\geq \bar{p}.
	\end{cases}
	\end{aligned}
	\end{equation}
	\item  $ \psi $ and $ \varphi $ are inverse to each other, i.e.,  $ \psi = \varphi^{-1} $ or $ \varphi = \psi^{-1} $.
\end{itemize}
\end{theorem}
\begin{proof}
	The proof of this theorem is based on constructing a resource utilization level $ \omega$ and a strictly-increasing function $ \psi $ for any $ \alpha $-competitive online algorithm under two special arrival instances. The complete proof is deferred to  Appendix \ref{proof_of_necessity_alpha_competitive}. 
\end{proof}

The first necessary condition in Theorem \ref{necessity} argues that the existence of a critical threshold $ \omega $ is directly related to the existence of any $ \alpha $-competitive algorithm.  Therefore, selling the resource at the lowest price $ \ubar{p} $ during $ [0,\omega] $ is necessary for $ \textsf{PM}_\phi $ to achieve a bounded competitive ratio. Meanwhile, the value of the critical threshold $ \omega $ must stay within a certain range defined by Eq. \eqref{flat_sufficiency_omega}. 

The second and the third conditions in Theorem \ref{necessity} are related to ODEs with two boundary conditions, which are often termed as first-order two-point BVPs in mathematics \cite{ODE1973, ODE_book}. For a given setup, Eq. \eqref{BVP_necessary_varphi} is related to the critical threshold $ \omega $ and the competitive ratio $ \alpha $, and thus we refer to Eq. \eqref{BVP_necessary_varphi} by $ \textsf{BVP}(\omega,\alpha) $ hereinafter. Note that $ \textsf{BVP}(\omega,\alpha) $ is equivalent to Eq. \eqref{ODI_principle_sufficiency} after enforcing the equality of the differential inequality. Meanwhile, the sufficient conditions in Theorem \ref{sufficiency} only require an increasing $ \varphi $ that satisfies Eq. \eqref{ODI_principle_sufficiency}, while a strictly-increasing $ \varphi $ is needed in Theorem \ref{necessity} to guarantee the existence of an arbitrary $ \alpha $-competitive online algorithm.

(\textbf{Principles of Optimal Design}) Based on Theorem \ref{necessity}, if we can find a competitive ratio parameter $ \alpha $ such that 
\begin{itemize}[leftmargin=*]
	\item for some $ \omega $ that satisfies Eq. \eqref{flat_sufficiency_omega}, we can find a strictly-increasing solution to  $\textsf{BVP}(\omega,\alpha) $;
	\item for all $ \omega$ that satisfies Eq. \eqref{flat_sufficiency_omega},  there exists no strictly-increasing solution to $\textsf{BVP}(\omega,\alpha-\epsilon) $, $\forall  \epsilon>0 $,
\end{itemize}
then this $ \alpha $ must be the optimal (minimum) competitive ratio achievable by all online algorithms\footnote{Note that the same design principle can also be applied for the BVP in Eq. \eqref{BVP_necessary_psi}. Here, we choose to deal with $ \textsf{BVP}(\omega,\alpha) $ because in the following we need to compute $ \varphi $ as part of our pricing function $ \phi $.}.  In the next two subsections, we use this principle to characterize the optimal competitive ratios for different setups.

\subsection{Optimal Designs for Convex Setups}
Theorem \ref{major_results} below  characterizes the design of a pricing function to achieve the optimal competitive ratio for a given convex setup.

\begin{theorem}\label{major_results}
Given a convex setup $ \mathcal{S}  $, there exists a unique  optimal critical threshold $ \omega_* $ so that the following claims hold:
\begin{itemize} 
	\item The optimal competitive ratio $ \alpha_*(\mathcal{S}) $ is given by 
	\begin{align}\label{optimal_alpha_convex}
	\alpha_*(\mathcal{S}) = \frac{h\big(\ubar{p}\big)}{F_{\ubar{p}}(\omega_*)} 
	=
	\begin{cases}
	\frac{\ubar{p}f'^{-1}(\ubar{p})-f(f'^{-1}(\ubar{p}))}{\ubar{p} \omega_* - f(\omega_*)}  &\text{if }\ \ubar{p}\in (\ubar{c},\bar{c}),\\
	\frac{\ubar{p}-f(1)}{\ubar{p} \omega_* - f(\omega_*)} &\text{if }\  \ubar{p}\in [\bar{c},+\infty).
	\end{cases}
	\end{align} 
	
	\item There exists a unique optimal pricing function $ \phi_* $ given by 
	\begin{align}\label{optimal_pricing_convex}
	\phi_*(y) = 
	\begin{cases}
	\ubar{p} &\text{if } y\in [0,\omega_*),\\
	\varphi_*(y)  &\text{if } y\in [\omega_*, \bar{\rho}],\\
	+\infty &\text{if } y\in (\bar{\rho},+\infty),
	\end{cases}
	\end{align}
	so that $ \textsf{PM}_{\phi_*} $ is IC and $ \alpha_*(\mathcal{S}) $-competitive. In Eq. \eqref{optimal_pricing_convex}, $ \varphi_* $ is the unique strictly-increasing solution to $\normalfont \textsf{BVP}(\omega_*,\alpha_*(\mathcal{S})) $.
\end{itemize}
\end{theorem}

Theorem \ref{major_results} illustrates our major results regarding the convex setup, namely, the existence and uniqueness of the optimal pricing function $ \phi_* $ so that $ \textsf{PM}_{\phi_*} $ achieves the optimal competitive ratio $ \alpha_*(\mathcal{S}) $. The proof of Theorem \ref{major_results} is based on analyzing the solution structures of $ \textsf{BVP}(\omega,\alpha) $, and the details are given in Section \ref{proof_of_major_results}. We emphasize that Theorem \ref{major_results} shows the existence and uniqueness of the optimal critical threshold $ \omega_* $ without discussing how to quantify it. For the details about how to calculate  $ \omega_* $, please refer to Theorem \ref{existence_uniqueness_opt_case_1} in Section \ref{proof_of_major_results} as well.  

\begin{corollary}\label{property_of_alpha_S}
Given a convex setup $ \mathcal{S}  $, we have:
\begin{itemize}
	\item $\alpha_*(\mathcal{S})  $ is strictly decreasing in $ \ubar{p}\in (\ubar{c},\bar{p}] $ for a given $ \bar{p}\in (\ubar{c},+\infty) $. 
	\item  $\alpha_*(\mathcal{S})  $ is strictly increasing in $ \bar{p}\in [\ubar{p},+\infty) $ for a given $ \ubar{p}\in (\ubar{c},+\infty) $. 
	\item $ \alpha_*(\mathcal{S}) = 1 $ and $ \omega_* = \ubar{\rho} $ when $ \ubar{p} = \bar{p} \in  (\ubar{c},+\infty)$.  
\end{itemize}
\end{corollary}

In Eq. \eqref{optimal_alpha_convex}, $ \alpha_*(\mathcal{S}) $ explicitly depends on $ f $ and $ \ubar{p} $, and implicitly depends on $ \bar{p} $ through the optimal critical threshold $ \omega_* $. This leads to the monotonicity of $ \alpha_*(\mathcal{S}) $ w.r.t. to $ \ubar{p} $ and $ \bar{p} $ in Corollary \ref{property_of_alpha_S}.  In particular, the third bullet in Corollary \ref{property_of_alpha_S} can be interpreted as follows: when $ \ubar{p} = \bar{p} \in  (\ubar{c},+\infty)$, the agents are identical in terms of their valuation densities, and thus it makes no difference to know all the future arrival information, i.e., $ \alpha_*(\mathcal{S}) = 1 $. The proof of Corollary \ref{property_of_alpha_S} is deferred to Appendix \ref{proof_of_property_of_alpha_S}. 

Before leaving this subsection, we give the following remark regarding the uniqueness of $ \phi_* $.
\begin{remark}[Uniqueness]
	We emphasize that the property of uniqueness in Theorem \ref{major_results} does not mean that there exists only one $ \alpha_*({\mathcal{S}}) $-competitive online mechanism/algorithm. Instead, Theorem \ref{major_results} only argues that $ \textsf{PM}_{\phi_*} $ can achieve the optimal competitive ratio $ \alpha_*({\mathcal{S}}) $ with a unique pricing function $ \phi_* $. 
\end{remark}

\subsection{Optimal Designs for Linear Supply Costs}
For a given convex setup $ \mathcal{S} $, Theorem \ref{major_results} shows that the optimal competitive ratio $ \alpha_*(\mathcal{S}) $ and  the optimal pricing function $ \phi_* $ are directly related to the optimal critical threshold $ \omega_* $.  It is worth emphasizing that $\omega_* $ is a design parameter which cannot be given in analytical forms, so does $ \phi_* $ (in fact, $ f $ itself is arbitrary, and thus this is not surprising). Nevertheless, for some special cost functions, Corollary \ref{major_results_log} below shows that logarithmic competitive ratios can be obtained via  analytical designs of  $ \phi_* $ and $ \omega_* $.

\begin{corollary}\label{major_results_log}
	Given a nice setup $ \mathcal{S} = \{f,\ubar{p},\bar{p}\} $, if the cost function $ f(y) = qy $, where $ q\geq 0 $,  then there exists a unique pricing function $ \phi_* $ given by
	\begin{align}\label{optimal_pricing_function_log}
	\phi_*(y) = 
	\begin{cases}
	\ubar{p} &\text{if } y\in [0,\omega_*),\\
	\big(\ubar{p}-q\big)\cdot\exp\left(y/\omega_* - 1\right) + q  &\text{if } y\in [\omega_*, 1],\\
	+\infty &\text{if } y\in (1,+\infty),
	\end{cases}
	\end{align}
	such that $ \textsf{PM}_{\phi_*} $ is IC and $ \alpha_*(\mathcal{S}) $-competitive, where  $ \alpha_*(\mathcal{S}) $ is given by
	\begin{equation*}
	\alpha_*(\mathcal{S}) = 1 + \ln\Big(\frac{\bar{p}-q}{\ubar{p}-q}\Big).
	\end{equation*}	
	In Eq. \eqref{optimal_pricing_function_log}, the optimal critical threshold $ \omega_* = \frac{1}{\alpha_*(\mathcal{S})} $. 
\end{corollary}

We note that the logarithmic competitive ratios in Corollary \ref{major_results_log} are not new and have been discussed  in the literature, e.g., \cite{ knapsack2008}. Based on \cite{knapsack2008}, such logarithmic competitive ratios are optimal, unless extra assumptions are introduced. Here, based on our above sufficient and necessary conditions for convex setups, we provide new proofs, which are simple and intuitive,  for the  results in Corollary \ref{major_results_log}. The details about the proof are deferred to Appendix \ref{proof_of_major_results_appendix}. 

Summarizing our results in Theorem \ref{major_results} and Corollary \ref{major_results_log}, we argue that $ \textsf{PM}_\phi $ is a unified mechanism for online resource allocation with or without supply costs. In particular, we obtain optimal competitive ratios for different setups including zero cost (i.e, without supply cost, which corresponds to $ q=0 $ in Corollary \ref{major_results_log}), linear cost, and strictly-convex cost.

\section{Interpretation, Intuition and Generalization of Theorem \ref{sufficiency}}
\label{proof_of_main_results_sufficiency}
In this section, we give a geometric interpretation of Theorem \ref{sufficiency} and discuss the intuitions of our geometric interpretation via worst-case analysis. We also show that the sufficient conditions in Theorem \ref{sufficiency} can be generalized in multiple directions.

\subsection{A Geometric Interpretation of Theorem \ref{sufficiency}}
\label{geometric_analysis}
Since $ \phi(\bar{\rho}) = \varphi(\bar{\rho})\geq \bar{p}  $, the highest-possible resource utilization level under $ \textsf{PM}_{\phi} $ is $ \phi^{-1}(\bar{p}) $.  
Let us denote the final resource utilization level under $ \textsf{PM}_\phi $ by $ \rho\in [0,\phi^{-1}(\bar{p})] $. Intuitively, if $ \rho\in [0,\omega] $, then the decisions made by $ \textsf{PM}_\phi $ and its offline counterpart are the same as long as $ \phi $ has the flat-segment $ [0,\omega] $, namely, both are to satisfy all the requests. Hence, if $ \rho \in [0,\omega] $, the competitive ratio of $ \textsf{PM}_\phi $ is 1. We next focus on the more general case when  $ \rho\in [\omega,\phi^{-1}(\bar{p})] $. 

For any $ \rho \in [\omega,\phi^{-1}(\bar{p})] $, let us denote the final price by $ p = \phi(\rho) $. The pricing function $ \phi(y) $ satisfying Eq. \eqref{ODI_principle_sufficiency}   indicates that for any given $ \rho\in [\omega,\phi^{-1}(\bar{p})]  $, we have
\begin{equation}\label{integral_version}
\int_\omega^\rho \left(\phi(y) - f'(y)\right)dy \geq  \int_{\ubar{p}}^p \frac{1}{\alpha}  h'\left(\phi(y)\right) d\phi(y).
\end{equation}
Meanwhile, based on Eq. \eqref{flat_sufficiency_omega}, we have 
\begin{equation}\label{initial_version}
F_{\ubar{p}}(\omega) = \ubar{p}\omega - f(\omega)\geq \frac{1}{\alpha} h(\ubar{p}),
\end{equation}
where $ F_{\ubar{p}}(\omega) $ is the profit function defined by Eq. \eqref{profit_function}.
The combination of Eq. \eqref{integral_version} and Eq. \eqref{initial_version} leads to the following condition
\begin{equation}\label{integral_version_expanding}
\ubar{p}\omega+\int_{\omega}^{\rho}\phi(y) dy - f\left(\rho\right) \geq  \frac{1}{\alpha} h(\phi(\rho)).
\end{equation}
Eq. \eqref{integral_version_expanding} must hold for all possible $ \rho $'s, i.e., all $ \rho\in [\omega,\phi^{-1}(\bar{p})] $,  and thus we have the following expression of $ \alpha $:
\begin{equation}\label{alpha_larger_than_a_ratio}
\alpha = \max_{\rho\in [\omega,\phi^{-1}(\bar{p})]}\  \frac{h(\phi(\rho))}{\ubar{p}\omega+\int_{\omega}^{\rho}\phi(y) dy - f\left(\rho\right)}.
\end{equation}

\begin{figure}
	\centering
	\subfigure[$ \rho\in (\omega,u{]} $]{\includegraphics[width=5 cm]{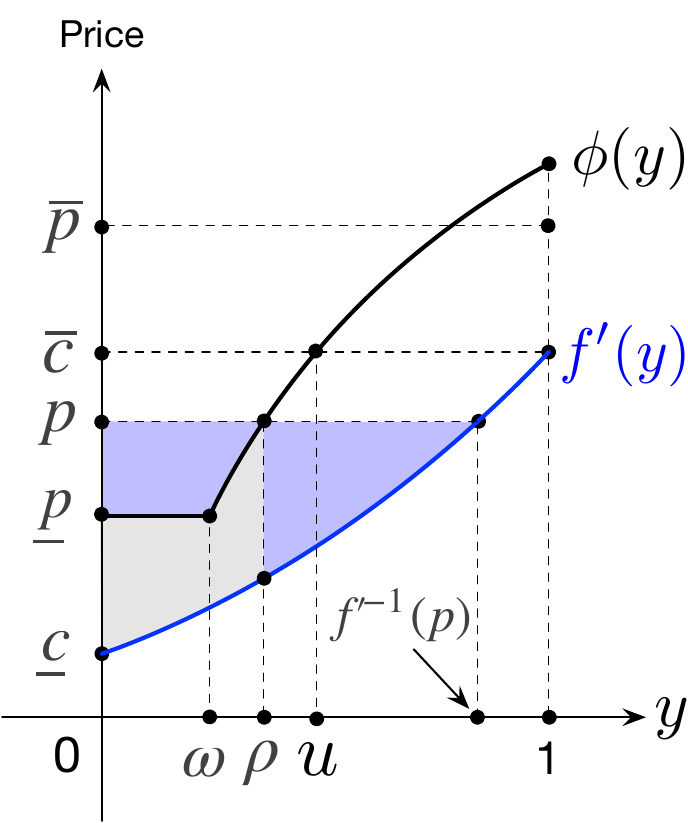}}
	\qquad
	\qquad 
	\subfigure[$ \rho\in (u,1{]} $]{\includegraphics[width= 5 cm]{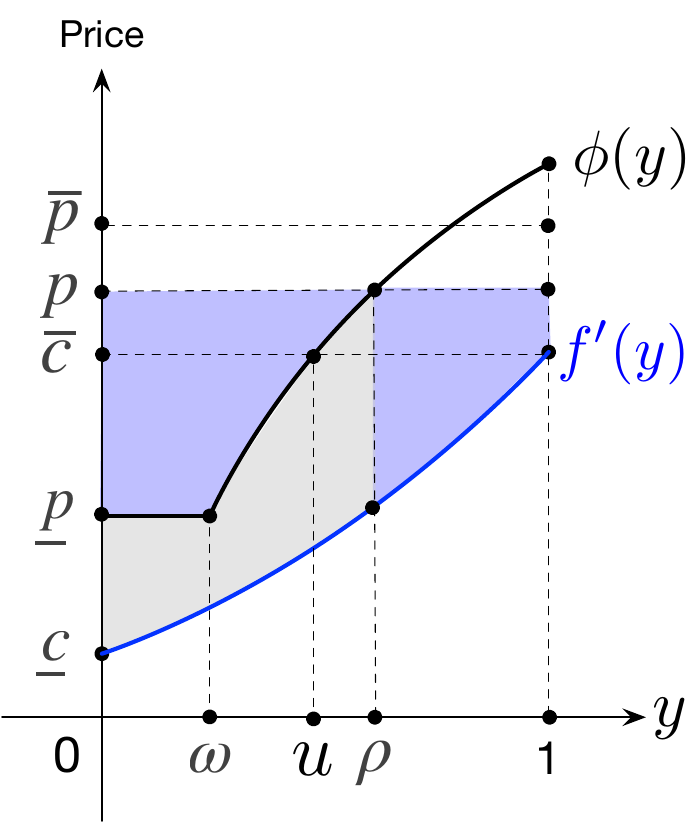}}
	\caption{Illustration of a feasible pricing function $ \phi $. The parameter $ u $ is a resource utilization level such that $ \phi(u) = \bar{c} $.}
	\label{fig_rationality_sufficiency}
\end{figure}

To illustrate the geometric meaning of Eq. \eqref{alpha_larger_than_a_ratio}, let us assume $ \bar{p}\geq \bar{c} $ for simplicity.  
We illustrate the geometric interpretation of Eq. \eqref{alpha_larger_than_a_ratio} in Fig. \ref{fig_rationality_sufficiency}.  Specifically,  Fig. \ref{fig_rationality_sufficiency}(a) shows the case when $ p = \phi(\rho)\leq \bar{c} $, i.e., $ \rho\in [\omega,u] $, where $ u $ is a resource utilization level such that $ \phi(u) = \bar{c} $ or equivalently, $ u  = \phi^{-1}(\bar{c}) $; Fig. \ref{fig_rationality_sufficiency}(b) shows the case when $ p = \phi(\rho)> \bar{c} $, i.e., $ \rho\in [u,1] $. Based on the expression of $ h(p) $ in Eq. \eqref{dual_bar_f_star}, for both subfigures in Fig. \ref{fig_rationality_sufficiency}, the area of the grey region and the two blue regions (which is referred to as `\textsf{\textbf{grey+blue}}') represents the numerator of the fraction in Eq. \eqref{alpha_larger_than_a_ratio}, while the denominator is the area of the grey region only (which is referred to as `\textsf{\textbf{grey}}'). Therefore, based on Eq. \eqref{alpha_larger_than_a_ratio}, an $ \alpha $-competitive $ \textsf{PM}_\phi $ is equivalent to a pricing function $ \phi $ so that the ratio between `\textsf{\textbf{grey+blue}}' and  `\textsf{\textbf{grey}}' is less than or equal to $ \alpha $ for all possible $ \rho$'s in $[\omega,\phi^{-1}(\bar{p})] $. 

In summary, given a pricing function $ \phi $ that satisfies the sufficient conditions in Theorem \ref{sufficiency}, when the final resource utilization level $ \rho\in [0,\omega] $, the competitive ratio of $ \textsf{PM}_\phi $ is 1; when $ \rho\in [\omega,\phi^{-1}(\bar{p}) ] $, the competitive ratio of $ \textsf{PM}_\phi $ equals the maximum ratio between `\textsf{\textbf{grey+blue}}' and   `\textsf{\textbf{grey}}' for all $ \rho\in [\omega,\phi^{-1}(\bar{p}) ] $. Since there is no prior information of $ \rho $,  the competitive ratio of $ \textsf{PM}_\phi $ is dominated by the latter case, leading to the expression of $ \alpha $ in Eq. \eqref{alpha_larger_than_a_ratio}.

\subsection{Intuitions via  Worst-Case Analysis}
\label{section_worst_case}
Our above analysis only visualizes the geometric meaning of Eq. \eqref{alpha_larger_than_a_ratio}, but reveals little intuition and rationality about why the maximum ratio between `\textsf{\textbf{grey+blue}}' and  `\textsf{\textbf{grey}}' determines the competitive ratio of $ \textsf{PM}_\phi $. Below we show that,  Eq. \eqref{alpha_larger_than_a_ratio} can be traced back to the original definition of $ \alpha $ in Eq. \eqref{definition_of_alpha}, based on which the rationality of Theorem \ref{sufficiency} can be demonstrated by the performance of $ \textsf{PM}_\phi $ under a special arrival instance in the worst-case scenario. 

To be more specific, let us again focus on the case when the final resource utilization level under $ \textsf{PM}_\phi $ is $ \rho\in [\omega,\bar{\rho}] $. For any given $ \rho\in [\omega,\bar{\rho}] $, let us assume the interval  $[0,\rho] $ is discretized into $ K+B $ blocks with each block $ \Delta $ long, i.e., $ \Delta = \rho/(K+B) $, where we assume $ \Delta $ is infinitesimally small. Let us consider a special arrival instance, denoted by $ \mathcal{A}_\rho $, which consists of three groups of agents as follows: \textit{the first group of agents are similar to our construction of $ \mathcal{A}_{\ubar{p}} $ in Section \ref{section_necessary_condition}, where we have $ K $ identical agents with valuation density $ \ubar{p} $ and requirement $ \Delta $. Here, we assume $ K\Delta =  \omega $. The second group of agents are indexed by $ b\in \{1,2,\cdots,B\} $, and the third group of  agents are also identical and are indexed by $ i $, where $ i = \{1,2,\cdots, I\} $ and $ I \geq 1/\Delta $.  For each request $b\in \{1,2,\cdots,B\} $ in the second group, we assume the valuation of agent $ b $ is given by $ v_b = \phi(\omega+b\Delta)\cdot \Delta $, where $ \phi(\omega+b\Delta) $ and $ \Delta $ denote the valuation density of agent $ b $ and the resource requirement of her request, respectively. Similarly, for each request $ i\in \{1,2,\cdots, I\} $ in the third group, we assume the valuation of agent $ i $ is given by $ v_i = \phi(\rho)\Delta$, where $ \phi(\rho) $ and $ \Delta $ denote the valuation density of agent $ i $ and the requirement of her request, respectively.} 

Given an arrival instance $ \mathcal{A}_\rho $ with $ \rho\in [\omega,\bar{\rho}] $, let us denote $ p = \phi(\rho) $. The optimal social welfare in hindsight, denoted by $ S_{\textsf{offline}}(\mathcal{A}_\rho) $, is to reject all the requests in the first two groups but satisfy all the requests in the third group until reaching the resource utilization level  $ y $, at which the marginal cost equals either $ p $ (i.e., $ y =  f'^{-1}(p) $) or $ \ubar{c} $ (i.e., $ y = 1 $).  Thus, $ S_{\textsf{offline}}(\mathcal{A}_\rho) $ is given by
\begin{equation}\label{S_offline_A_rho}
\begin{aligned} S_{\textsf{offline}}(\mathcal{A}_\rho) 
=
\begin{cases}
p f'^{-1}(p) - f\big(f'^{-1}(p)\big) &\text{if } p\in (\ubar{c},\bar{c}],\\
p -f(1) &\text{if } p\in (\bar{c},+\infty),
\end{cases}
\end{aligned}
\end{equation}
which equals the numerator of the fraction in Eq. \eqref{alpha_larger_than_a_ratio} based on the expression of $ h(p) $ in Eq. \eqref{def_h_p}. 
Note that  $ I \geq 1/\Delta  $ guarantees that the final resource utilization level $ y =  f'^{-1}(p) $ or $ y = 1 $ can be reached by the offline optimal strategy.

In comparison, given the arrival instance $ \mathcal{A}_\rho $, our online mechanism $ \textsf{PM}_\phi $ will satisfy all the requests in the first two groups and reject all in the third group. Thus, the social welfare is given by
\begin{equation}\label{S_online_A_rho}
S_{\textsf{online}}(\mathcal{A}_\rho) =  \ubar{p} K\Delta + \sum_{b=1}^{B} \phi(\omega+b\Delta)\cdot\Delta - f\left(\rho\right),
\end{equation}
which equals  the denominator of the fraction in Eq. \eqref{alpha_larger_than_a_ratio} since $ K\Delta = \omega $  and $ \Delta $ is infinitesimally small. 

Based on \eqref{alpha_larger_than_a_ratio}-\eqref{S_online_A_rho}, the competitive ratio of $ \textsf{PM}_\phi $ is given by
\begin{equation*}
\alpha = \max_{\rho\in [\omega,\bar{\rho}]}\  \frac{S_{\textsf{offline}}(\mathcal{A}_\rho) }{S_{\textsf{online}}(\mathcal{A}_\rho)} =   \max_{\normalfont\text{all possible }\mathcal{A}} \frac{S_{\textsf{offline}}(\mathcal{A})}{S_{\textsf{online}}(\mathcal{A})},
\end{equation*}
where the second equality is from  the original definition of $ \alpha $ in Eq. \eqref{definition_of_alpha}. 
It becomes clear now that the sufficient conditions in Theorem \ref{sufficiency} lead to the expression of $ \alpha $ in Eq. \eqref{alpha_larger_than_a_ratio}, which in principle captures the worst-case performance ratio between the optimal strategy in hindsight and our proposed online mechanism $ \textsf{PM}_\phi $ under a continuum of special arrival instances $ \{\mathcal{A}_\rho\}_{ \{\forall \rho\in [\omega,\bar{\rho}]\}} $. In this regard, the arrival instance $ \mathcal{A}_\rho $ is not constructed randomly; it indeed  captures the worst-case scenario in the context. 

\subsection{Generalizations of Theorem \ref{sufficiency}}
\label{section_generalization}

\subsubsection{General Increasing Pricing Functions}
Based on Eq. \eqref{alpha_larger_than_a_ratio}, we can generalize Theorem \ref{sufficiency} to the following Corollary \ref{main_result_any} with a broader class of competitive pricing functions.

\begin{corollary}\label{main_result_any}
Given a convex setup $ \mathcal{S} $, if   $ \phi $ is given by
\begin{align}\label{pricing_function_any}
	\phi(y) = 
	\begin{cases}
		\ubar{p} &\text{if } y\in [0,\omega),\\
		\varphi(y)  &\text{if } y\in [\omega, 1],\\
		+\infty &\text{if } y\in (1,+\infty),
	\end{cases}
\end{align}
where $ \omega \in (0,\ubar{\rho}] $, and $ \varphi $ is  increasing in $ [\omega,1] $ with $ \varphi(\omega) = \ubar{p} $ and $ \varphi(1)\geq \bar{c} $, then $ \textsf{PM}_\phi $ is IC and $ \alpha(\omega,\varphi) $-competitive, where $ \alpha(\omega,\varphi) $ is
\begin{align}
	\alpha(\omega,\varphi) = \max\Bigg\{\frac{h\big(\ubar{p}\big)}{F_{\ubar{p}}(\omega)},  \frac{h(\bar{p})}{\ubar{p}\omega+\int_{\omega}^{\bar{\rho}_\varphi}\varphi(y) dy - f\left(\bar{\rho}_\varphi\right)},  \max_{\rho\in [\omega,\bar{\rho}_{\varphi}]}\  \frac{h\big(\varphi(\rho)\big)}{\ubar{p}\omega+\int_{\omega}^{\rho}\varphi(y) dy - f\left(\rho\right)} \Bigg\}.\label{alpha_S}
\end{align}
In Eq. \eqref{alpha_S},  $ \bar{\rho}_{\varphi} $ is the maximum resource utilization level defined as:
\begin{equation*}
	\begin{aligned}
		\bar{\rho}_{\varphi} \triangleq 
		\begin{cases}
			\varphi^{-1}(\bar{p})  &\text{if } \bar{p} \leq \varphi(1),\\
			1  &\text{if } \bar{p} > \varphi(1).
		\end{cases}
	\end{aligned} 
\end{equation*}
\end{corollary}

Corollary \ref{main_result_any} argues that any continuous pricing function with a flat-segment and an increasing-segment can lead to a bounded competitive ratio for $ \textsf{PM}_\phi $.  Recall that when $ \omega = 0 $, our previous necessary conditions in Theorem \ref{necessity} argue that it is impossible to achieve a bounded competitive ratio. This can be illustrated by Eq. \eqref{alpha_S}, where the first term in the bracket is unbounded since $ F_{\ubar{p}}(0) = 0 $. 
Note that the third term in Eq. \eqref{alpha_S} is derived from Eq. \eqref{alpha_larger_than_a_ratio}. Meanwhile, when $ \varphi(1)\geq \bar{p} $, the second term in the bracket of Eq. \eqref{alpha_S} is not needed  as it is contained by the third term.  The proof of Corollary \ref{main_result_any} is given in Appendix \ref{proof_of_main_result_any}.

\begin{remark}
	Based on Theorem \ref{major_results} and Corollary \ref{main_result_any}, we argue that $ \alpha_*(\mathcal{S}) $ is the optimal solution to the following optimization problem:
	\begin{alignat*}{3}
	\alpha_*(\mathcal{S}) =\ & \underset{\omega,\varphi}{\normalfont\textsf{minimize}}\qquad   & &  \alpha(\omega,\varphi) \\
	&\normalfont\textsf{subject to} & & \omega\in [0,\ubar{\rho}], \\
	&  & &  \varphi'(y)\geq 0, \forall y\in [\omega,1],\\
	& & & \varphi(\omega) = \ubar{p}, \varphi(1) \geq \bar{c}.
	\end{alignat*}
	In this regard, our previous design of $ \omega_* $ and  $ \varphi_* $ in Theorem \ref{major_results} essentially provides a method of solving the above  optimization problem in functional spaces.
\end{remark}

\subsubsection{Pricing Functions for Multiple Time Slots}\label{extension_multiple_resources}
Our mechanism also extends to more general settings when there exists multiple time slots. Specifically, let us consider the following model:
\begin{subequations}\label{SWM_multiple_extension_paper}
	\begin{alignat}{3}
	& \underset{\bm{x,y}}{\textsf{maximize}}\qquad       & &  \sum_{n\in\mathcal{N}} v_n x_n - \sum_{t\in\mathcal{T}} \bar{f}_t(y_t) \label{objective_multiple_paper}\\
	& \textsf{subject to} & & \sum_{n\in\mathcal{N}}   r_n^t x_n = y_t, \forall t\in\mathcal{T}, \label{total_load_multiple_paper} \\
	& & & x_n\in  \{0,1\},\forall n\in\mathcal{N}.
	\end{alignat}
\end{subequations}
In Problem \eqref{SWM_multiple_extension_paper}, we consider a discrete time system and index different time slots by $ t\in\mathcal{T} \triangleq \{1,2,\cdots, T\} $. In Eq. \eqref{objective_multiple_paper}, the supply cost at $ t\in\mathcal{T} $ is directly written as an extended cost function $\bar{f}_t$, i.e.,  $ \bar{f}_t(y) = 
f_t(y) $ if $ y\in [0,1] $, and 
$ \bar{f}_t(y) = +\infty $ if $ y\in (1,+\infty) $. In Eq. \eqref{total_load_multiple_paper}, $ r_n^t $ denotes the requirement of agent $ n $ at $ t\in\mathcal{T}_n$, where $ \mathcal{T}_n $ represents the time duration of request $ n $. 
We assume $ r_n^t = 0 $ if $ t\notin \mathcal{T}_n $ and  $ r_n^t > 0 $ if $ t\in \mathcal{T}_n $ so that we can simply denote the resource requirement of request $ n $ by $ \{r_n^t\}_{\forall t\in\mathcal{T}} $.  
Meanwhile, similar to our previous basic resource allocation model, each agent $ n\in\mathcal{N} $ can be represented by 
$ \bm{\theta}_n = \big(\{r_n^t\}_{\forall t\in\mathcal{T}_n},v_n\big) $, 
where  $ v_n $ denotes the valuation of agent $ n $ if all the requirements $ \{r_n^t\}_{\forall t\in\mathcal{T}_n} $ are satisfied.

While most of our previous definitions and notations can be reused by simply adding a time index $ t\in\mathcal{T} $, e.g., $ h_t(p), \ubar{c}_t $, and $ \bar{c}_t $, we redefine the lower and upper bounds of the valuation density by
\begin{equation*}
\min_{n\in \mathcal{N}:r_n^t \neq 0} \Big\{\frac{v_n}{r_n^t}\Big\} \geq \ubar{p}_t, \max_{n\in\mathcal{N}:r_n^t \neq 0} \Big\{\frac{v_n}{r_n^t}\Big\} \leq \bar{p}_t, \forall t\in\mathcal{T},
\end{equation*}
where $ \ubar{p}_t $ and $ \bar{p}_t $ correspond to $ \ubar{p} $ and $ \bar{p} $ in Assumption \ref{Assumption_PUV}, respectively. Based on $ \ubar{p}_t $ and $ \bar{p}_t $, we can define  $\ubar{\rho}_t$ and $\bar{\rho}_t$ in the similar way as $\ubar{\rho} $ and $\bar{\rho}$ in Eq. \eqref{def_of_rho}. Here we omit the details for brevity.

Following the principle discussed in Section \ref{section_PM}, we can design a  pricing function  $ \hat{p}_t^{(n)} = \phi_t(y_t^{(n)}) $ for each time slot $ t\in\mathcal{T} $, where $ y_t^{(n)} $ denotes the total resource utilization after processing agent $ n $, similar to our definition of $ \hat{y}_n $ in Eq. \eqref{update_of_omega}. Based on the pricing functions $ \{\phi_t\}_{\forall t} $, we can set the utility of agent $ n\in\mathcal{N} $ by
\begin{equation*}
\hat{\gamma}_n = \max\Big\{v_n - \sum\nolimits_{t\in\mathcal{T}_n} r_n^t \cdot \phi_t\big(y_t^{(n-1)}\big), 0\Big\}, \forall n\in\mathcal{N},
\end{equation*}
where $ \sum_{t\in\mathcal{T}_n} r_n^t \cdot \phi_t(y_t^{(n-1)})  $ denotes the payment of agent $ n $ if $ \hat{\gamma}_n > 0 $. 
The mechanism $ \textsf{PM}_\phi $ can thus be generalized to a posted price mechanism $ \textsf{PM}_{\bm{\phi}} $ with a vector of pricing functions $ \bm{\phi} \triangleq \{\phi_t\}_{\forall t} $. 

Based on the above discussions, we give a generalized version of Theorem \ref{sufficiency} in the following Theorem \ref{sufficiency_multiple_paper}.

\begin{theorem}\label{sufficiency_multiple_paper}
Given a convex setup $ \mathcal{S} = \{f_t,\ubar{p}_t,\bar{p}_t\}_{\forall t} $, $ \textsf{PM}_{\bm{\phi}} $ is IC and $ \max_t\{\alpha_t\} $-competitive  if for each $ t\in\mathcal{T} $,  $ \phi_t $ is given by
\begin{equation}\label{sufficient_phi_multiple_paper}
	\begin{aligned}
		\phi_t(y) = 
		\begin{cases}
			\ubar{p}_t &\text{if } y\in [0,\omega_t),\\
			\varphi_t(y)  &\text{if } y\in [\omega_t, 1],\\
			+\infty &\text{if } y\in (1,+\infty),
		\end{cases}
	\end{aligned}
\end{equation}
where $ \omega_t $ is the critical threshold that satisfies
\begin{equation}\label{flat_sufficiency_omega_multiple_paper}
	F_{\ubar{p}_t}(\omega_t)  \geq    \frac{1}{\alpha_t} h_t(\ubar{p}_t) \text{ and }
	0\leq \omega_t \leq \ubar{\rho}_t,
\end{equation} 
and $ \varphi_t(y) $ is  an increasing function that satisfies
\begin{equation}\label{ODI_principle_sufficiency_multiple_paper}
	\begin{aligned}
		\begin{cases}
			\varphi_t'(y) \leq   \alpha_t \cdot\frac{\varphi_t(y) - f_t'(y)}{h'\left(\varphi_t(y)\right)}, y\in (\omega_t,1)\\
			\varphi_t(\omega_t) = \ubar{p}_t, \varphi_t(1)\geq \bar{p}_t + \sum_{\hat{t}\in \mathcal{T}\backslash\{t\}}h_{\hat{t}}\big(\bar{p}_{\hat{t}}\big).
		\end{cases}
	\end{aligned}
\end{equation} 
\end{theorem}

Note that Eq. \eqref{ODI_principle_sufficiency_multiple_paper} is the same as Eq. \eqref{ODI_principle_sufficiency} except the second boundary condition, which depends on parameters related to all the time slots $ \mathcal{T} $. Our previous analysis regarding the  basic resource allocation model shows that the optimal competitive ratio $ \alpha_*(\mathcal{S}) $ depends on the boundary conditions in Eq. \eqref{BVP_necessary_varphi}. Therefore, in the case with multiple time slots, the final competitive ratio $ \alpha = \max_t\{\alpha_t\} $  depends on $ |\mathcal{T}| $ as well.  For the proof of Theorem \ref{sufficiency_multiple_paper}, as well as discussions of the properties of $ \{\phi_t\}_{\forall t} $ based on Theorem \ref{sufficiency_multiple_paper}, please refer to  Appendix \ref{section_extension}. 

\begin{remark}[Applications in Cloud Computing]
	Problem \eqref{SWM_multiple_extension_paper} can be used to model the online resource allocation in cloud computing \cite{value_based_scheduling, Zhang2017}.  
	A cloud service provider allocates a single type of resources (e.g., CPU) to a set of jobs $ \mathcal{N}=\{1,\cdots,N\} $ that arrive in a sequential order, and each job $ n\in\mathcal{N} $ is active in duration $ \mathcal{T}_n \subset\mathcal{T} $. Based on this interpretation, Problem \eqref{SWM_multiple_extension_paper} can be regarded as maximizing the social welfare of online resource allocation in cloud computing with server costs (e.g., energy costs). 
	In particular, the arrival and departure times of job $ n $ can be taken into account by $ \mathcal{T}_n $, and agents are flexible to set their resource requirements in each time slot based on the length of their active durations.
\end{remark}

We end this section with the following remark regarding the setting with multiple types of resources.

\begin{remark}[Multiple Types of Resources]
	Problem \eqref{SWM_multiple_extension_paper} can also be interpreted in the context of resource allocation with multiple types of resources.  Specifically, we can consider that in each time slot we have a different type of resource so that $ \mathcal{T} $ denotes the set of resource types. Based on this notational system, $ r_n^t $ denotes the resource requirement of agent $ n $ for resource type $ t $, and $ \mathcal{T}_n $ represents the set of resource types required by agent $ n $. For this reason, it is mathematically equivalent to consider multiple types of resources and multiple time slots. In addition, we can also consider the setup with multiple resource types and multiple time slots simultaneously. We argue that our previous design principle still applies to such more general and complex setups. The difference is that we need to design a pricing function for each type of resource at each time slot\footnote{The rationality is that in the worst case  agents may only require a single type of resource at a single time slot.}. 
\end{remark}

\section{Proof of Theorem \ref{major_results}}
\label{proof_of_major_results}
In this section, we provide the proof of Theorem \ref{major_results}. Our proof is organized in three cases based on the relationship between $ \ubar{c},\bar{c} $, $ \ubar{p} $, and $ \bar{p} $. After the proof, we give an example (quadratic supply cost $ f(y) = \frac{1}{2}y^2 $) to show  the calculation of $ \alpha_*(\mathcal{S}) $, the optimal critical threshold $ \omega_* $, and the optimal pricing function $ \phi_* $.

\begin{figure*}
	\centering
	\subfigure[\textbf{Category of the three cases}]{\includegraphics[height = 6cm]{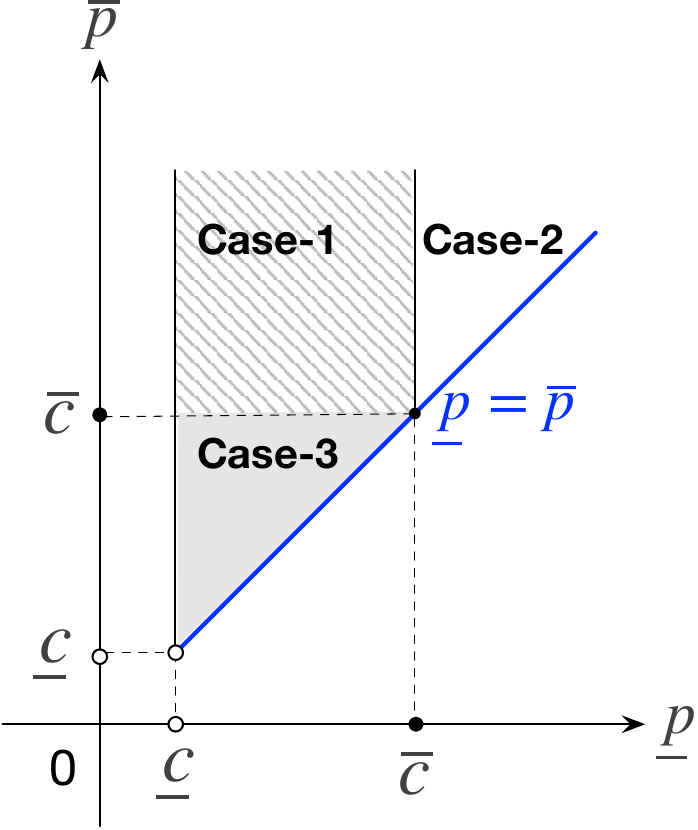}}
	\qquad
	\qquad 
	\subfigure[\textbf{Case-1}: $\ubar{c}<  \ubar{p} <  \bar{c} < \bar{p}$]{\includegraphics[height= 6.2 cm]{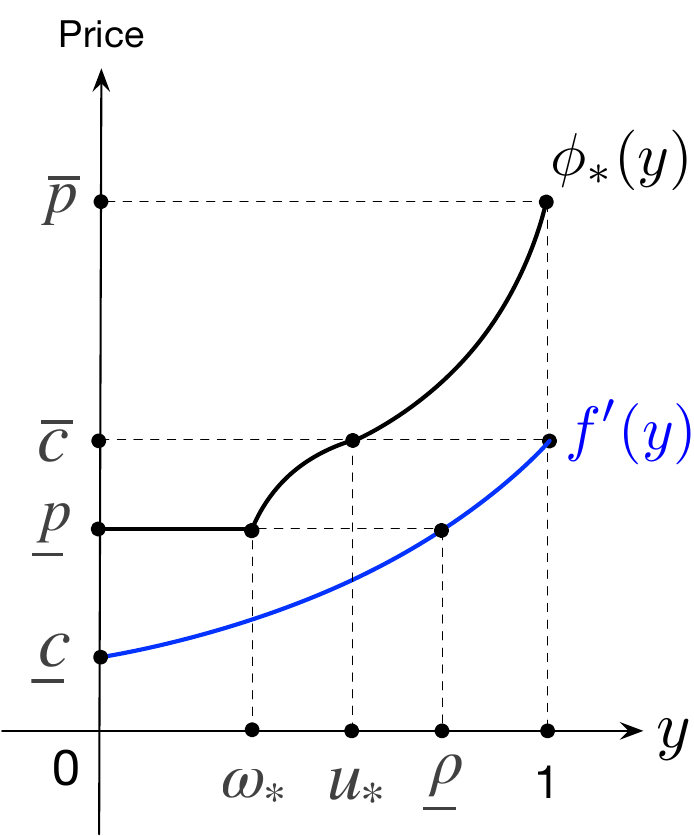}} 
	\\
	\subfigure[\textbf{Case-2}: $\ubar{c}< \bar{c}\leq \ubar{p}\leq \bar{p}$]{\includegraphics[height= 6.2 cm]{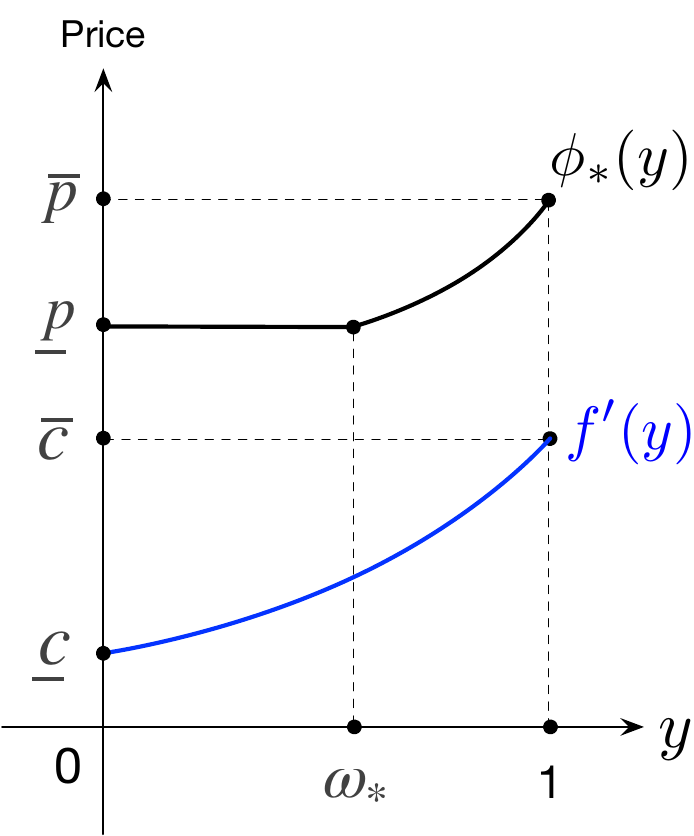}} 
	\qquad\qquad 
	\subfigure[\textbf{Case-3}: $ \ubar{c}<\ubar{p}\leq \bar{p}\leq \bar{c}$]{\includegraphics[height= 6.2 cm]{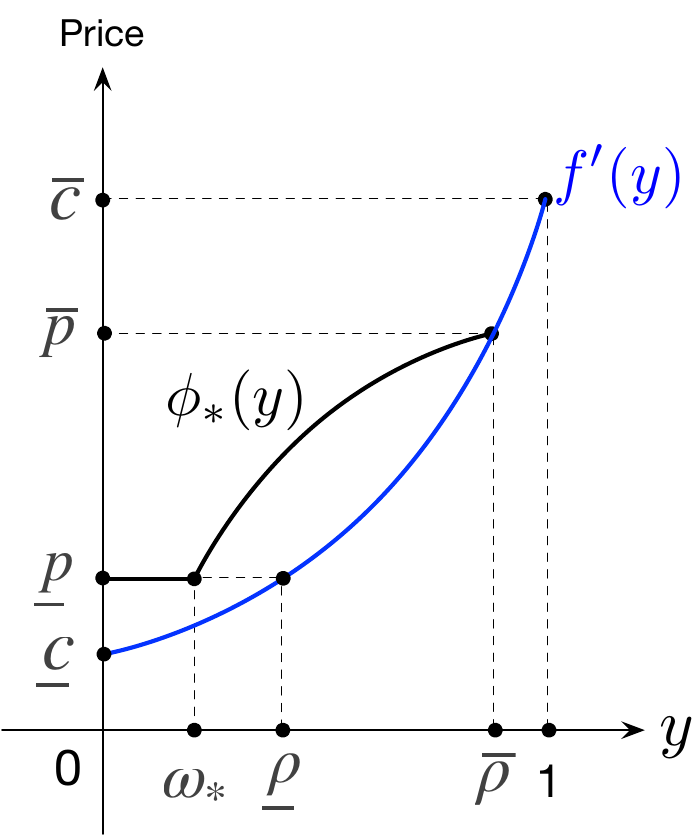}}
	\caption{Category of the three cases for a given setup and illustrations of the optimal pricing functions  in the three cases.}
	\label{three_cases}
\end{figure*}

\subsection{Overview of Our Three-Case Proof}
Our proof heavily follows the sufficient conditions in Theorem \ref{sufficiency} and the necessary conditions in Theorem \ref{necessity}. For the sake of better reference,  here we revisit $ \normalfont\textsf{BVP}(\omega,\alpha)  $ as follows:
\begin{align}\label{BVP_necessary_varphi_revisit}
\normalfont\textsf{BVP}(\omega,\alpha)
\begin{cases}
\varphi'(y) =  \alpha \cdot\frac{\varphi(y) - f'(y)}{h'(\varphi(y))}, y\in (\omega,\bar{\rho}),\\
\varphi(\omega) = \ubar{p}, \varphi(\bar{\rho})\geq \bar{p}.
\end{cases}
\end{align}

Based on the expression of $ h $ in Eq. \eqref{def_h_p}, the denominator of the ODE in Eq. \eqref{BVP_necessary_varphi_revisit} can be written as
\begin{equation}\label{h_prime}
\begin{aligned}
h'\left(\varphi(y)\right) =
\begin{cases}
f'^{-1}\left(\varphi(y)\right) &\text{if } \varphi(y)\in [\ubar{c},\bar{c}],\\
1 &\text{if } \varphi(y)\in (\bar{c},+\infty),
\end{cases}
\end{aligned}
\end{equation} 
Eq. \eqref{h_prime}  indicates that the ODE in Eq. \eqref{BVP_necessary_varphi_revisit} can be equivalently transformed into two ODEs based on whether $ \varphi(y)$ is larger than  $ \bar{c} $ or not. Moreover, notice that the second boundary condition of $ \normalfont\textsf{BVP}(\omega,\alpha)  $ is $ \varphi(\bar{\rho})\geq \bar{p} $, where $ \bar{p} $ can be either less than or larger than $ \bar{c} $.  For this reason, based on the values of $ \ubar{p} $ and $ \bar{p} $, as well as the two boundary conditions in Eq. \eqref{BVP_necessary_varphi_revisit}, we can categorize the whole possibilities into three cases as follows. 
\begin{itemize}
	\item \textbf{Case-1}: $ \ubar{c}<  \ubar{p} <  \bar{c} < \bar{p} $. This case captures the setup when the minimum valuation density is small but the maximum valuation density is large. 
	\item \textbf{Case-2}: $\ubar{c}< \bar{c}\leq \ubar{p}\leq \bar{p}$. This case captures the setup when the minimum valuation density is large. 
	\item \textbf{Case-3}:  $ \ubar{c}<\ubar{p}\leq \bar{p}\leq \bar{c}$.  This case captures the setup when the maximum valuation density is small. 
\end{itemize}

We illustrate the above three cases in Fig. \ref{three_cases}(a). A nice setup $ \mathcal{S} $ must have $ \bar{p}\geq \ubar{p} $, and thus we only focus on the upper-left triangular part.   We next present our proof of Theorem \ref{major_results} in \textbf{Case-1}. The proofs of Theorem \ref{major_results} in \textbf{Case-2} and \textbf{Case-3} are similar and thus are deferred to  Appendix \ref{missing_proofs_case_2_case_3}.

\subsection{Proof of Theorem \ref{major_results} in Case-1}\label{section_Case_II}
In this case, $ \ubar{p}< \bar{c} <  \bar{p}$ indicates that $ \ubar{\rho}< 1 = \bar{\rho} $. Based on the two boundary conditions in Eq. \eqref{BVP_necessary_varphi_revisit}, there must exist a resource utilization threshold $ u\in (\omega,1) $ so that $ \phi(u) = \bar{c} $, as illustrated in Fig. \ref{fig_rationality_sufficiency}(b). Based on Theorem \ref{sufficiency} and Theorem \ref{necessity}, we give the following Corollary  \ref{theorem_case_1_two_BVP} which summarizes the sufficient and necessary conditions in \textbf{Case-1}. 

\begin{corollary}\label{theorem_case_1_two_BVP}
Given a convex setup $ \mathcal{S} $ in \textbf{Case-1}, $ \textsf{PM}_\phi  $ is IC and $ \alpha $-competitive if there exists a pair of resource utilization thresholds $ (\omega,u) \in [F_{\ubar{p}}^{-1}\big(\frac{1}{\alpha} h\big(\ubar{p}\big)\big),\ubar{\rho}]\times(\omega,1) $ such that $ \phi $ is given by
\begin{equation}\label{phi_case_1}
	\begin{aligned}
		\phi(y)=
		\begin{cases}
			\ubar{p} &\text{if } y\in [0,\omega],\\
			\varphi_1(y) &\text{if } y\in [\omega,u],\\
			\varphi_2(y) &\text{if } y\in [u,1],\\
			+\infty   &\text{if } y\in (1,+\infty),
		\end{cases}
	\end{aligned}
\end{equation}
where    $ \varphi_1(y) $  is a solution to the following BVP:
\begin{equation}\label{BVP_1_case_1}
	\begin{aligned}
		\normalfont\textsf{BVP}_1(\omega,u,\alpha) 
		\begin{cases}
			\varphi_1'(y) = \alpha\cdot \frac{\varphi_1(y)  - f'(y) }{f'^{-1}\left(\varphi_1(y)\right)}, y\in(\omega,u);\\
			\varphi_1(\omega) = \ubar{p}, \varphi_1(u) =\bar{c},	
		\end{cases}
	\end{aligned}
\end{equation}
and $ \varphi_2(y) $ is a solution to the following BVP:
\begin{equation}\label{BVP_2_case_1}
	\begin{aligned}
		\normalfont\textsf{BVP}_2(u,\alpha)
		\begin{cases}
			\varphi_2'(y) = \alpha\cdot \left(\varphi_2(y) - f'(y) \right),  y\in (u,1);\\
			\varphi_2(u) = \bar{c},  \varphi_2(1) \geq \bar{p}.
		\end{cases}
	\end{aligned}
\end{equation}
On the other hand, if there exists an $ \alpha $-competitive online algorithm, then there must exist some pair of $ (\omega,u) \in [F_{\ubar{p}}^{-1}\big(\frac{1}{\alpha} h(\ubar{p})\big),\ubar{\rho}]\times(\omega,1) $ such that $ \normalfont\textsf{BVP}_1(\omega, u,\alpha) $ and $ \normalfont\textsf{BVP}_2(u,\alpha) $ are well-defined, and each of them has  a strictly-increasing solution. 
\end{corollary}

In Corollary \ref{theorem_case_1_two_BVP}, $ F_{\ubar{p}}^{-1} $ represents the inverse\footnote{Based on Eq. \eqref{profit_function}, $ F_{\ubar{p}}(y) $ is strictly increasing in $ y\in [0,\ubar{\rho}] $ since $ f'(y)\leq \ubar{p} $ holds for all $ y\in [0,\ubar{\rho}] $. Thus, the inverse of $ F_{\ubar{p}} $, denoted by $ F_{\ubar{p}}^{-1} $,  is well-defined.} of the profit function $ F_{\ubar{p}} $ defined in Eq. \eqref{profit_function}, and $ F_{\ubar{p}}^{-1}\big(\frac{1}{\alpha} h\big(\ubar{p}\big)\big) $ is derived from Eq. \eqref{flat_sufficiency_omega}. The flat-segment $ [0,\omega] $ of $ \phi $ directly follows Theorem \ref{sufficiency}. The two BVPs in Eq. \eqref{BVP_1_case_1} and Eq. \eqref{BVP_2_case_1}, as well as their corresponding solutions $ \varphi_1(y) $ and $ \varphi_2(y) $, follow Theorem \ref{necessity} after substituting $ h'(\varphi(y)) $ from Eq. \eqref{h_prime} into $ \normalfont\textsf{BVP}(\omega,\alpha)  $.  The sufficiency and necessity of Corollary \ref{theorem_case_1_two_BVP} thus follow.  

Corollary \ref{theorem_case_1_two_BVP} not only shows how to justify whether a given pricing function is competitive or not, but also argues that, to find the optimal competitive ratio, we simply need to find the minimum $ \alpha $ so that both $ \textsf{BVP}_1(\omega,u,\alpha) $  and $ \textsf{BVP}_2(u,\alpha) $ have  strictly-increasing solutions. Below we give Proposition \ref{two_lower_bounds_case_1} which shows the existence conditions of solutions to both $ \textsf{BVP}_1(\omega,u,\alpha) $ and $ \textsf{BVP}_2(u,\alpha) $.
\textbf{Case-1}. 

\begin{proposition}\label{two_lower_bounds_case_1}
Given a convex setup $ \mathcal{S} $ in \textbf{Case-1}, the following claims regarding $\normalfont \textsf{BVP}_1(\omega,u,\alpha)  $ and $\normalfont \textsf{BVP}_2(u,\alpha)  $ are true:
\begin{itemize}
	\item For each given pair of $ (\omega,u) \in [F_{\ubar{p}}^{-1}\big(\frac{1}{\alpha} h\big(\ubar{p}\big)\big),\ubar{\rho}]\times(\omega,1) $, there exists a well-defined function $ \Gamma_1(\omega,u) $ so that  $ \normalfont\textsf{BVP}_1(\omega,u,\alpha)  $ has a unique strictly-increasing  solution if and only if $ \alpha = \Gamma_1(\omega,u) $.
	\item For each given $ u\in (0,1) $, $\normalfont \textsf{BVP}_2(u,\alpha)  $ has a unique strictly-increasing solution if and only if $ \alpha \geq  \Gamma_2(u) $, where $ \Gamma_2(u) $ is the unique root to the following equation in variable $ \Gamma_2 $:
	\begin{equation}\label{Gamma_2}
		\int_u^{1}\frac{\Gamma_2 f'(y)}{ \exp(y\Gamma_2) } dy = \frac{\bar{p}}{\exp(\Gamma_2)} - \frac{\bar{c}}{\exp(u\Gamma_2)}.
	\end{equation}
	Moreover, for each given $ u\in (0,1) $, when $ \alpha =  \Gamma_2(u) $, the unique solution to $ \textsf{BVP}_2(u,\alpha) $ satisfies $ \varphi_2(1) = \bar{p} $.
\end{itemize}
\end{proposition}

The proof of Proposition \ref{two_lower_bounds_case_1} is based on the two ODEs given in Eq. \eqref{BVP_1_case_1} and  Eq. \eqref{BVP_2_case_1}, and the details are deferred to Appendix \ref{proof_of_two_lower_bounds_case_1}. Based on the above Proposition \ref{two_lower_bounds_case_1}, below we give the optimal  competitive ratio in \textbf{Case-1}.

\begin{theorem}\label{existence_uniqueness_opt_case_1}
	Given a convex setup $ \mathcal{S} $ in \textbf{Case-1}, the optimal  competitive ratio achievable by all online algorithms is given by:
	\begin{equation}\label{equation_three_case_1}
		\alpha_*({\mathcal{S}}) =   \frac{h\big(\ubar{p}\big)}{F_{\ubar{p}}(\omega_*)}  =  \Gamma_1(\omega_*,u_*) = \Gamma_2(u_*),
	\end{equation}
	where $ u_*\in (0,1) $ is the unique root that satisfies
	\begin{equation}\label{system_of_equations_case_1}
		\Gamma_1\left(F_{\ubar{p}}^{-1}\bigg(\frac{h(\ubar{p})}{\Gamma_2(u_*)}\bigg),u_*\right) = \Gamma_2(u_*),
	\end{equation}
	and  $ \omega_*\in [0,\ubar{\rho}] $ is given by 
	\begin{equation*}
		\omega_* =  F_{\ubar{p}}^{-1}\bigg(\frac{h\big(\ubar{p}\big)}{\Gamma_2(u_*)}\bigg).
	\end{equation*}
	Meanwhile, $ \textsf{PM}_{\phi_*}  $ is $ \alpha_*({\mathcal{S}}) $-competitive if and only if $ \phi_* $ is given by Eq. \eqref{phi_case_1} with $ (\omega,u,\alpha) = (\omega_*,u_*, \alpha_*({\mathcal{S}})) $. Moreover, we have $ \phi_*(1) =  \bar{p} $.   
\end{theorem}

\begin{proof}
	This theorem shows the unique existence of $ \omega_* $,  $ u_* $, and $ \alpha_*(\mathcal{S}) $. Based on these parameters, we can calculate the unique optimal pricing function $ \phi_* $ based on Eq. \eqref{phi_case_1}. The proof of the unique existence of $ \omega_* $ and $ u_* $ is based on Proposition \ref{two_lower_bounds_case_1}. 
	For the detailed proof, please refer to Appendix \ref{proof_of_existence_uniqueness_opt_case_1}.
\end{proof}

Based on Theorem \ref{existence_uniqueness_opt_case_1}, Theorem \ref{major_results} directly follows in \textbf{Case-1}. Due to space limitations, we defer the proof of Theorem \ref{major_results} in \textbf{Case-2} and \textbf{Case-3} to Appendix \ref{missing_proofs_case_2_case_3}.

Fig. \ref{three_cases}(b) illustrates the optimal pricing function $ \phi_*  $  in \textbf{Case-1}. We can see that in addition to the flat-segment $ [0,\omega_*] $, the increasing-segment is further divided into two parts, namely, $ [\omega_*,u_*] $ and $ [u_*,1] $. Meanwhile, Fig. \ref{three_cases}(c) and Fig. \ref{three_cases}(d) illustrate the optimal pricing functions in \textbf{Case-2} and \textbf{Case-3}, respectively. Note that by  $ \textsf{PM}_{\phi_*} $, the highest-possible resource utilization levels in \textbf{Case-1} and \textbf{Case-2} are 1, while in \textbf{Case-3} is $ \bar{\rho}$ and $ \bar{\rho}<1 $. This is consistent with our definition of $ \bar{\rho} $ in Eq. \eqref{def_of_rho}. 

For all the three cases, the optimal pricing function $ \phi_* $ cannot be given in analytical forms. We argue that the computation of $ \omega_*, u_* $, and $ \alpha_*(\mathcal{S}) $ is light-weight and can be performed efficiently via various numerical methods  such as  bisection searching.  More importantly, $\omega_*$ and $ u_* $ can be computed offline before the start of $ \textsf{PM}_\phi $. For more detailed discussions of how to quantify $\omega_*$, $u_*$, and $\alpha_*(\mathcal{S}) $, please refer to Appendix \ref{appendix_compute_alpha}.

\subsection{An Example: Quadratic Supply Costs}
\label{section_case_study}
To better show how to calculate $ \omega_* $, $ u_* $, $ \alpha_*(\mathcal{S}) $, and $ \phi_* $ in Theorem \ref{existence_uniqueness_opt_case_1}, in this subsection we perform a case study based on quadratic supply costs. Let us assume $ f(y) = \frac{1}{2}y^2 $ (i.e.,  $ f'(y) =  y $), and thus $ \ubar{c} = f'(0) = 0 $ and $ \bar{c} = f'(1) = 1 $. Based on the definitions of $ h $ and $ F_p $ in Section \ref{primal_dual_offline_setting}, we have
\begin{equation*}
h(\ubar{p}) = 
\begin{cases}
\frac{1}{2}\ubar{p}^{2} & \text{if } p\in [\ubar{c},\bar{c}]\\
\ubar{p} - \frac{1}{2} & \text{if } p\in (\bar{c},+\infty)
\end{cases}
,\ 
F_{\ubar{p}}(\omega) = \ubar{p}\omega -\frac{1}{2}\omega^2.
\end{equation*} 
Therefore, the lower bound of the critical threshold $ \omega $ can be written as a function of $ \alpha $ as $ F_{\ubar{p}}^{-1}(\frac{1}{\alpha}h(\ubar{p})) =  \ubar{p}\big(1-\sqrt{1-1/\alpha}\big) $.

We next show how to compute $ \omega_* $, $ u_* $, $ \alpha_*(\mathcal{S}) $ in \textbf{Case-1}. When $ f'(y) = y $,  $ \textsf{BVP}_1(\omega,u,\alpha)  $ in Eq. \eqref{BVP_1_case_1} can be written as
\begin{equation}\label{BVP_1_example}
\begin{aligned}
\begin{cases}
\varphi_1'(y) = \alpha\cdot \frac{\varphi_1(y)  - y }{\varphi_1(y)}, y\in(\omega,u);\\
\varphi_1(\omega) = \ubar{p}, \varphi_1(u) = 1.
\end{cases}
\end{aligned}
\end{equation}
The ODE in Eq. \eqref{BVP_1_example} can be solved by separation of variables \cite{ODE_book}. Substituting the two boundary conditions into the solution of the ODE in Eq. \eqref{BVP_1_example} leads to the following equation of $ \alpha$, $ \omega $, and $ u $:
\begin{equation}\label{Gamma_1_example}
\int_{\ubar{p}/\omega}^{1/u} \frac{-\eta}{\eta^2-\alpha\eta+\alpha} d\eta = \ln\left(\frac{u}{\omega}\right).
\end{equation}
Based on Proposition \ref{two_lower_bounds_case_1}, Eq. \eqref{Gamma_1_example} has a unique root in $ \alpha\geq 1 $ for each given pair of $ (\omega,u) \in [F_{\ubar{p}}^{-1}\big(\frac{1}{\alpha} h\big(\ubar{p}\big)\big),\ubar{\rho}]\times(\omega,1) $. Therefore, $ \alpha = \Gamma_1(\omega,u) $ is well-defined through Eq. \eqref{Gamma_1_example}.

Similarly, $ \textsf{BVP}_2(u,\alpha)  $ in Eq. \eqref{BVP_2_case_1} can be written as
\begin{equation}\label{BVP_2_example}
\begin{aligned}
\begin{cases}
\varphi_2'(y) = \alpha\cdot \left(\varphi_2(y) - y \right),  y\in (u,1);\\
\varphi_2(u) = \bar{c},  \varphi_2(1) =  \bar{p},
\end{cases}
\end{aligned} 
\end{equation}
where the second boundary condition is written as an equality since the optimal design should have $ \phi_*(1) = \bar{p} $, as indicated by Theorem \ref{existence_uniqueness_opt_case_1}.
Solving the ODE in Eq. \eqref{BVP_2_example} is elementary. Substituting the two boundary conditions into the analytical solution of the ODE in Eq. \eqref{BVP_2_example} leads to the following equation of $ \alpha $ and $ u $:
\begin{equation}\label{Gamma_2_example}
\exp\left(\alpha(1-u)\right)  = \frac{\alpha(1+\bar{p})+1}{\alpha(u+ \bar{c})+1}.
\end{equation}
Based on Proposition \ref{two_lower_bounds_case_1},  Eq. \eqref{Gamma_2_example} has a unique root in $ \alpha $ for each given $ u\in (0,1)$, and thus $ \alpha = \Gamma_2(u) $ is well defined through Eq. \eqref{Gamma_2_example}. Note that Eq. \eqref{Gamma_2_example} is a simplification of Eq. \eqref{Gamma_2} when $ f'(y) = y $. 

Combining the above results, when $ f'(y) = y $, the calculation of $ \omega_* $, $ u_* $ and $ \alpha_*(\mathcal{S}) $ (which is denoted as $ \alpha_* $ below for simplicity) in Eq. \eqref{equation_three_case_1} can be written as the following system of three equations:
\begin{subequations}\label{system_of_equations_example}
	\begin{align}
	\omega_* =  \ubar{p}\left(1-\sqrt{1-1/\alpha_*}\right),\label{system_of_equations_example_1} \\
	\int_{\ubar{p}/\omega_*}^{1/u_*} \frac{-\eta}{\eta^2-\alpha_*\eta+\alpha_*} d\eta = \ln\left(\frac{u_*}{\omega_*}\right),\label{system_of_equations_example_2}\\
	\exp\left(\alpha_*(1-u_*)\right)  = \frac{\alpha_*(1+\bar{p})+1}{\alpha_*(u_* + 1)+1}.\label{system_of_equations_example_3}
	\end{align}
\end{subequations}
Based on Theorem \ref{existence_uniqueness_opt_case_1}, the above system of equations has a unique pair of solutions $ (\omega_*, u_*, \alpha_*) $. In particular, the optimal competitive ratio for the quadratic cost setup can be written as
\begin{equation*}
\alpha_*(\mathcal{S}) = \frac{h\big(\ubar{p}\big)}{F_{\ubar{p}}(\omega_*)} 
= 
\begin{cases}
\frac{\ubar{p}^2}{2\ubar{p}\omega_* - \omega_*^2} & \text{if } p\in [\ubar{c},\bar{c}],\\
\frac{2\ubar{p} - 1 }{2\ubar{p}\omega_* - \omega_*^2}& \text{if } p\in (\bar{c},+\infty).
\end{cases}
\end{equation*}
Recall that for the linear cost setup in Corollary \ref{major_results_log}, $ \alpha_*(\mathcal{S}) = \frac{1}{\omega_*} $. 

When calculating the unique solution to Eq. \eqref{system_of_equations_example}, we can substitute $ \omega_* $ given by Eq. \eqref{system_of_equations_example_1} into Eq. \eqref{system_of_equations_example_2}, and then leverage the monotonicity of both sides of Eq. \eqref{system_of_equations_example_2} and Eq. \eqref{system_of_equations_example_3} to locate the unique values of $ \alpha_* $ and $ u_* $. 
Based on $ \alpha_* $ and $ u_* $, the optimal critical threshold $ \omega_* $ can be easily calculated based on  Eq. \eqref{system_of_equations_example_1}. 
Theorem \ref{existence_uniqueness_opt_case_1} shows that the optimal pricing function  $ \phi_* $ can be obtained as long as $ \omega_*, u_* $ and $ \alpha_* $ are given. For example, in \textbf{Case-1}, we can solve $ \textsf{BVP}_1(\omega_*,u_*,\alpha_*) $ in Eq. \eqref{BVP_1_example} to get the optimal pricing function $ \varphi_1^* $ in the interval $ [\omega_*,u_*] $, and solve $ \textsf{BVP}_2(u_*,\alpha_*) $ in Eq. \eqref{BVP_2_example} to get the optimal pricing function $ \varphi_2^* $ in the interval $ [u_*, 1] $. The complete optimal pricing function $ \phi_* $ (in \textbf{Case-1}) can thus be determined. 

We numerically solve the system of equations in Eq. \eqref{system_of_equations_example} and illustrate the relationship between $ \alpha_* $ and $ \bar{p} $ in Fig. \ref{fig_computation_alpha}. Specifically, we fix $ \ubar{p} = 0.3 $ and vary $ \bar{p} $ from $ 0.3 $ to $ 7.1 $ in Fig. \ref{fig_computation_alpha}(a) so that our parameter setting covers both \textbf{Case-1} (i.e., $ \bar{p} > 1 $) and \textbf{Case-3} (i.e., $ \bar{p} \leq 1 $). Meanwhile, we fix $ \ubar{p} = 1.1 $ and vary $ \bar{p} $ from $ 1.1 $ to $ 12.1 $ in Fig. \ref{fig_computation_alpha}(b) to simulate \textbf{Case-2} (i.e., $ \ubar{p} \geq  1 $). Fig. \ref{fig_computation_alpha} shows that $ \alpha_*(\mathcal{S}) = 1 $ when $ \bar{p} = \ubar{p} $, and is strictly increasing in $ \bar{p}\in [\ubar{p}, +\infty) $. This validates our results in Corollary \ref{property_of_alpha_S}. Note that the systems of equations for the calculation of $ \alpha_* $ in \textbf{Case-2} and \textbf{Case-3} are different from Eq. \eqref{system_of_equations_example}. For the details, please refer to Appendix \ref{missing_proofs_case_2_case_3}.

\begin{figure}
	\centering
	\subfigure[$ \ubar{p} = 0.3 $]{\includegraphics[width= 5 cm]{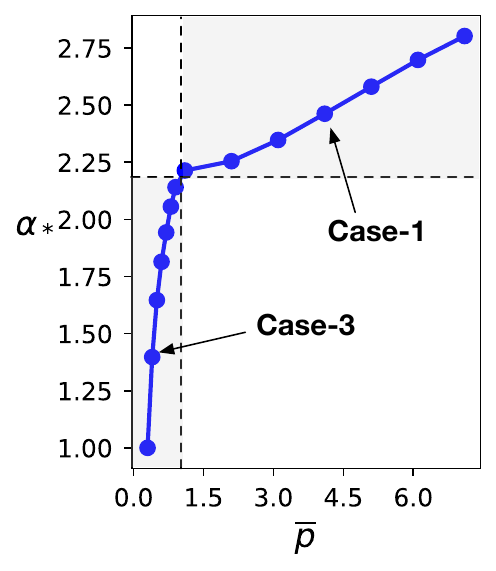}}
	\qquad 
	\qquad
	\subfigure[$ \ubar{p} = 1.1 $]{\includegraphics[width= 5.2 cm]{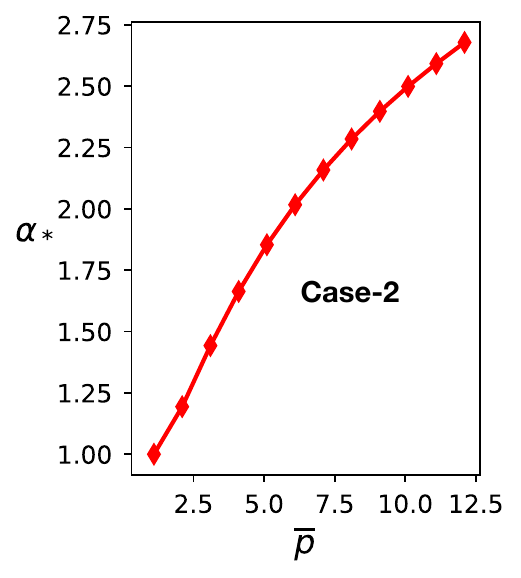}}
	\caption{Relationship between  $ \alpha_*(\mathcal{S}) $ and $ \bar{p} $ when $ f(y) = \frac{1}{2}y^2 $.}
	\label{fig_computation_alpha}
\end{figure}

\section{Conclusion and Future Work}
We studied the mechanism design for general online resource allocation problems in the presence of supply costs and capacity limits. In this setting, the supplier can produce additional units of resources at increasing marginal cost, but with a stringent capacity limit; agents are strategic and may deliberately misreport their preference to be better off. The major contribution of the paper is the development of a unified incentive compatible online mechanism, with a principled method of constructing a universal ``pricing function," that leads to the optimal competitive ratios for different setups including zero, linear and strictly-convex supply costs.  

One of the drawbacks of the current design is that the optimal competitive ratio, as well as other key design parameters such as the optimal critical threshold $ \omega_* $,
cannot be given in analytical forms. This brings difficulty in analyzing the properties of the competitive ratio w.r.t. important system parameters such as the lower and upper bounds of valuation densities (i.e., $ \ubar{p} $ and $ \bar{p} $).  It also remains unclear if there exists any analytical (and thus simple) pricing function so that the optimal competitive ratio can be approximately achieved. Meanwhile, extending the proposed techniques to other online optimization problems is also an interesting future direction.

\bibliographystyle{plain}
\bibliography{MD_unified_arXiv}

\newpage
\appendix

\section{Proof of Theorem \ref{sufficiency}}
\label{proof_sufficiency_alpha_competitive}
The sufficient conditions in Theorem \ref{sufficiency} are derived based on Proposition \ref{OPD_principle}. Specifically, we need to prove that under $ \textsf{PM}_{\phi} $, i) the sequences of $ \{P_n\}_{\forall n} $ and $ \{D_n\}_{\forall n} $  are feasible, and ii) there exists some $ k\in\mathcal{N} $ so that the initial inequality and the incremental inequality hold. 	Based on Proposition \ref{feasibility}, the feasibility of $ \{P_n\}_{\forall n} $ and $ \{D_n\}_{\forall n} $ are trivial since the pricing function $ \phi $ given by Eq. \eqref{pricing_function} is monotone. In the following we focus on the proof of the initial inequality and the incremental inequality. 

(\textbf{Initial Inequality}) Let us assume $\omega =  \sum_{n=1}^kr_n $, namely, $ k $ is the number of agents such that the total resource allocated equals $ \omega $. In the following we prove that the initial inequality $ P_k\geq \frac{1}{\alpha}D_k $ holds. Based on Eq. \eqref{flat_sufficiency_omega}, $ \omega $ is the resource utilization threshold such that 
$ F_{\ubar{p}}(\omega) = \ubar{p} \omega - \bar{f}\left(\omega\right)  \geq \frac{1}{\alpha}  h(\ubar{p}) $. Thus, we have
\begin{align*}
\ubar{p}\cdot \left(\sum_{n=1}^kr_n\right)-  \bar{f}\left(\sum_{n=1}^kr_n\right) \geq \frac{1}{\alpha} h(\ubar{p}),
\end{align*}
Since $ \alpha\geq 1 $ and $ \hat{\gamma}_n \geq 0 $ for all $ n\in\mathcal{N} $, the above inequality leads to the following one
\begin{align*}
\left(1 - \frac{1}{\alpha}\right)\sum_{n=1}^k \hat{\gamma}_n + \sum_{n=1}^k \ubar{p} r_n-  \bar{f}\left(\sum_{n=1}^k r_n\right) \geq \frac{1}{\alpha} h(\ubar{p}).
\end{align*}
The pricing function given by Eq. \eqref{pricing_function} indicates that the requirements of all the requests will be satisfied as long as the resource utilization is below $ \omega $. For this reason, we have $ \hat{y}_k = \sum_{n=1}^k r_n  = \omega $, which represents the total resource utilization after processing request $ k $.  Therefore, we have 
\begin{align}\label{proof_initial_alpha}
\sum_{n=1}^k \hat{\gamma}_n  +\sum_{n=1}^k \ubar{p} r_n- \bar{f}(\hat{y}_k) \geq \frac{1}{\alpha}\left(\sum_{n=1}^{k} \hat{\gamma}_n +  h(\ubar{p})\right),
\end{align}
Note that the left-hand-side of Eq. \eqref{proof_initial_alpha} is equal to $ P_k $ since 
$ \sum_{n=1}^k \hat{\gamma}_n  +\sum_{n=1}^k  \ubar{p} r_n- \bar{f}(\hat{y}_k) =  \sum_{n=1}^k v_n - \bar{f}(\hat{y}_k) = P_k $,  
where the first equality comes from $ v_n = \hat{\gamma}_n + \hat{p}_{n-1}r_n = \hat{\gamma}_n + \ubar{p}r_n $. Here, $\hat{p}_{n-1} = \ubar{p}  $ holds for all $ n = \{1,2,\cdots,k\} $ because the posted prices for all these agents are $ \ubar{p} $ (i.e., the flat-segment). 
Meanwhile, based on the objective of Problem \eqref{dual_problem}, the terms inside the parenthesis of Eq. \eqref{proof_initial_alpha} equal $ D_k $, namely, $ D_k = \sum_{n=1}^{k} \hat{\gamma}_n +  h(\ubar{p}) $.  
Therefore, Eq. \eqref{proof_initial_alpha} indicates that $ P_k\geq \frac{1}{\alpha}D_k $ holds. We thus complete the proof of the initial inequality. 

{\color{black}
	(\textbf{Incremental Inequality}) We now prove that the incremental inequality $ P_n- P_{n-1}\geq \frac{1}{\alpha}\left(D_n -D_{n-1}\right) $ holds for all $ n\in \{k+1, \cdots, N\} $. Note that if agent $ n $ is rejected (i.e., the demand of agent $ n $ is not satisfied), then $ P_n = P_{n-1} $ and $ D_n = D_{n-1} $. In this case, the incremental inequality $ P_n- P_{n-1}\geq \frac{1}{\alpha}\left(D_n -D_{n-1}\right) $ holds for all $ \alpha\geq 1 $. Therefore, in the following we only focus on the case when the requirement of agent $ n $ is satisfied. Eq. \eqref{ODI_principle_sufficiency} shows that  the following differential inequality  holds
	\begin{align}\label{ODE_appendix}
	\phi(\hat{y}_{n-1}) - \bar{f}'(\hat{y}_{n-1}) \geq \frac{1}{\alpha}\cdot h'\big(\phi(\hat{y}_{n-1})\big)\cdot\phi'(\hat{y}_{n-1})
	\end{align}
	for all $ \hat{y}_{n-1}\in [\omega,\bar{\rho}) $.  Since $ \hat{y}_n - \hat{y}_{n-1} = r_n $, based on the Taylor series we have
	\begin{align*}
	&\bar{f}\left(\hat{y}_n\right) = \bar{f}\left(\hat{y}_{n-1}\right) + \bar{f}'\left(\hat{y}_{n-1}\right)(\hat{y}_n-\hat{y}_{n-1}) + O(r_n^2),\\
	&h\big(\phi(\hat{y}_n)\big) = h\big(\phi(\hat{y}_{n-1})\big)  + h'\big(\phi(\hat{y}_{n-1})\big)\big(\phi(\hat{y}_n) - \phi(\hat{y}_{n-1})\big) + O\left(\big(\phi(\hat{y}_n) - \phi(\hat{y}_{n-1})\big)^2\right),\\
	&\phi\left(\hat{y}_n\right) = \phi\left(\hat{y}_{n-1}\right) + \phi'\left(\hat{y}_{n-1}\right)(\hat{y}_n-\hat{y}_{n-1}) + O(r_n^2).
	\end{align*}
	Based on Eq. \eqref{ODE_appendix} and the above Taylor series, we have
	\begin{align*}
	&\phi\left(\hat{y}_{n-1}\right)-\frac{\bar{f}\left(\hat{y}_n\right)-\bar{f}\left(\hat{y}_{n-1}\right)}{\hat{y}_n-\hat{y}_{n-1}} -O(r_n)\\
	\geq\ & \frac{1}{\alpha} \cdot \left(\frac{h\left(\phi\left(\hat{y}_n\right)\right)-h\left(\phi\left(\hat{y}_{n-1}\right)\right)}{\phi\left(\hat{y}_n\right)-\phi\left(\hat{y}_{n-1}\right)} - O(r_n)\right)\cdot \left(\frac{\phi\left(\hat{y}_n\right)-\phi\left(\hat{y}_{n-1}\right)}{\hat{y}_n-\hat{y}_{n-1}}- O(r_n)\right),
	\end{align*}
	Assumption \ref{Assumption_d_n} states that $ r_n $ is sufficiently small, i.e., $ r_n \rightarrow 0 $. Thus, the following inequality holds:
	\begin{align*}
	\phi\left(\hat{y}_{n-1}\right)\left(\hat{y}_n-\hat{y}_{n-1}\right)-\left(\bar{f}\left(\hat{y}_n\right)-\bar{f}\left(\hat{y}_{n-1}\right)\right)
	\geq \frac{1}{\alpha}\left( h\left(\phi\left(\hat{y}_n\right)\right)-h\left(\phi\left(\hat{y}_{n-1}\right)\right)\right).
	\end{align*}
	Since $ \hat{\gamma}_n \geq 0$ and $ 1-\frac{1}{\alpha} \geq 0 $, we have 
	\begin{align}\label{proof_incremental_alpha}
	\left(1-\frac{1}{\alpha}\right)\hat{\gamma}_n +   \bigg(\phi(\hat{y}_{n-1})(\hat{y}_n-\hat{y}_{n-1})- \Big(\bar{f}(\hat{y}_n)-\bar{f}(\hat{y}_{n-1})\Big)\bigg)
	\geq\frac{1}{\alpha} \left( h\left(\phi\left(\hat{y}_n\right)\right)-h\left(\phi\left(\hat{y}_{n-1}\right)\right)\right),
	\end{align}
	
	Note that after processing agent $ n $, the change in the primal objective is given by
	$P_n - P_{n-1}
	= v_n - \left(\bar{f}(\hat{y}_n)-\bar{f}(\hat{y}_{n-1})\right) 
	= \hat{\gamma}_n +\phi(\hat{y}_{n-1})\left(\hat{y}_n-\hat{y}_{n-1}\right)- \left(\bar{f}(\hat{y}_n)-\bar{f}(\hat{y}_{n-1})\right)
	$,
	where the second equality is due to $ v_n = \hat{\gamma}_n + \hat{p}_{n-1} r_n  =  \hat{\gamma}_n + \phi(\hat{y}_{n-1})(\hat{y}_n-\hat{y}_{n-1})$. Similarly, the change in the dual objective is given by $ D_n-D_{n-1} = \hat{\gamma}_n +  h(\hat{p}_{n})-h(\hat{p}_{n-1}) = \hat{\gamma}_n +  h\left(\phi\left(\hat{y}_n\right)\right)-h\left(\phi\left(\hat{y}_{n-1}\right)\right)  $. Therefore, Eq. \eqref{proof_incremental_alpha} indicates that the incremental inequality $ P_n- P_{n-1}\geq \frac{1}{\alpha}\left(D_n -D_{n-1}\right) $ holds  for all $ n\in \{k+1, \cdots, N\} $ when $ \hat{y}_{n-1}\in [\omega,\bar{\rho}) $.
}	

Note that we need to guarantee that the incremental inequality holds at each round for all $ \hat{y}_{n-1}\in [\omega,\bar{\rho}] $,
where $ \bar{\rho} $ represents the maximum resource utilization level. 
Therefore, we need to consider the boundary condition of $ \phi(y) $ when $ y = \bar{\rho} $. 
First, $ \phi(\omega) = \ubar{p} $ is obvious as the pricing function $ \phi $ must be continuous.  Taking integration of both sides of Eq. \eqref{ODE_appendix} leads to
\begin{align*}
\int_{\omega}^{\rho}\phi(y)dy - \bar{f}(\rho) + \bar{f}(\omega) \geq  \frac{1}{\alpha}\Big(h\big(\phi(\rho)\big) - h\big(\ubar{p}\big)\Big).
\end{align*}
Meanwhile, based on the flat-segment, we have 
$ \ubar{p} \omega - \bar{f}(\omega) \geq \frac{1}{\alpha} h(\ubar{p}) $. Thus, we have
\begin{align}\label{integral_version_expanding_single_appendix}
\ubar{p}\omega + \int_{\omega}^\rho\phi(y)dy - \bar{f}(\rho) \geq  \frac{1}{\alpha}h\big(\phi(\rho)\big),
\end{align}
which is the same as Eq. \eqref{integral_version_expanding}. When $ \rho = \bar{\rho} $, the above inequality also holds, and thus we have
\begin{align}\label{booundary_rho_bar}
\ubar{p}\omega + \int_{\omega}^{\bar{\rho}}\phi(y)dy - f(\bar{\rho}) \geq  \frac{1}{\alpha}h\big(\phi(\bar{\rho})\big).
\end{align}

On the other hand, $ \textsf{PM}_\phi $ is $ \alpha $-competitive indicates that the following inequality must hold:
\begin{align}\label{boundary_p_bar}
\ubar{p}\omega + \int_{\omega}^{\bar{\rho}}\phi(y)dy - f(\bar{\rho}) \geq  \frac{1}{\alpha}h(\bar{p}).
\end{align}
Therefore, if we have
$ \phi(\bar{\rho})\geq \bar{p} $, then  Eq. \eqref{booundary_rho_bar} guarantees that Eq. \eqref{boundary_p_bar} always holds. 
This means, the incremental inequality $ P_n- P_{n-1}\geq \frac{1}{\alpha}\left(D_n -D_{n-1}\right) $ holds at each round for all $ \hat{y}_{n-1}\in [\omega,\bar{\rho}] $ as long as the pricing function satisfies 
\begin{align*}
\begin{cases}
\phi(y) - f'(y) \geq \frac{1}{\alpha} \cdot h'\left(\phi(y)\right)\cdot{\phi'(y)}, y\in (\omega,\bar{\rho}), \\
\phi(\omega) = \ubar{p}, \phi(\bar{\rho})\geq \bar{p}.
\end{cases}
\end{align*}
We thus complete the proof of the incremental inequality in Eq. \eqref{ODI_principle_sufficiency}.

Summarizing our above proofs regarding the initial and incremental inequalities, we complete the proof of Theorem \ref{sufficiency}.

\section{Proof of Theorem \ref{necessity}} \label{proof_of_necessity_alpha_competitive}
Our proof of Theorem \ref{necessity} is organized as follows. We first give a set of less restrictive conditions in Theorem \ref{necessity_alpha_competitive}, based on which we give  Theorem \ref{strictly_increasing_psi} to show the existence of strictly-increasing allocation functions and pricing functions. Based on Theorem \ref{necessity_alpha_competitive} and Theorem \ref{strictly_increasing_psi},  Theorem \ref{necessity} directly follows. 

\subsection{General Necessary Conditions}
Theorem \ref{necessity_alpha_competitive} below gives a set of necessary conditions to guarantee the existence of an $ \alpha $-competitive online algorithm. 

\begin{theorem}
	\label{necessity_alpha_competitive}
	Given a convex setup $ \mathcal{S} $, if there exists an $ \alpha $-competitive online algorithm, then there must exist a parameter $\omega $ which satisfies \begin{align}\label{flat_necessary_omega}
	F_{\ubar{p}}(\omega)   \geq    \frac{1}{\alpha}  h\big(\ubar{p}\big) \text{ and }
	0\leq \omega \leq \ubar{\rho}
	\end{align}
	so that the following claims hold simultaneously: 
	\begin{itemize}
		\item \textbf{NC1}): There exists a case when the total resource utilization is $ \omega $ and all the accepted agents have the same valuation density $ \ubar{p} $.
		\item \textbf{NC2}): There exists a non-decreasing function $ \psi(p) $ that satisfies
		\begin{align}\label{ODI_psi}
		\begin{cases}
		\ubar{p}\omega+
		\int_{\ubar{p}}^{p}\eta \psi'(\eta) d\eta - f\big(\psi(p)\big) 
		\geq  \frac{1}{\alpha} h\left(p\right), p\in (\ubar{p},\bar{p}),\\
		\psi\big(\ubar{p}\big) = \omega, \psi\big(\bar{p}\big)\leq \bar{\rho},
		\end{cases}
		\end{align}
		where $ \psi(p) $ is a monotone mapping from $ p\in  [\ubar{p},\bar{p}] $ to $ \psi(p)\in [\omega,\bar{\rho}]$.
	\end{itemize}
\end{theorem}
\begin{proof}
	The proof of this theorem is based on constructing a resource utilization level $ \omega$ and a monotone function $ \psi $ for any $ \alpha $-competitive online algorithm under two special arrival instances.  We organize our proofs of \textbf{NC1} and \textbf{NC2} as follows. 
	
	(\textbf{Proof of NC1}) By definition,  an $ \alpha $-competitive online algorithm must achieve an online social welfare that is no less than $ \frac{1}{\alpha} $ of its offline counterpart for all possible arrival instances. To prove this, we first construct an arrival instance $ \mathcal{A}_{\ubar{p}} $ as follows: \textit{there are $ K $ agents in $ \mathcal{A}_{\ubar{p}} $ with the same valuation density  $ \ubar{p} $, and all the requests ask for resources with the same requirement $ \Delta $, where $ \Delta $ is infinitely small and $ K\Delta= \ubar{\rho} $}. 
	Based on  Eq. \eqref{def_of_rho}, the optimal offline strategy under this arrival instance is to accept all the requests and the resulting optimal social welfare is given by
	$ S_{\textsf{offline}}(\mathcal{A}_{\ubar{p}}) = K\ubar{p}\Delta - \bar{f}(K\Delta)= \ubar{p}\cdot \ubar{\rho} - \bar{f}(\ubar{\rho}) = h(\ubar{p})
	$. For any $ \alpha $-competitive online algorithm (including $ \textsf{PM}_\phi $), let us denote the final resource utilization level by $ \omega\in [0,1]$, then the online social welfare is given by
	$S_{\textsf{online}}(\mathcal{A}_{\ubar{p}}) = \ubar{p}\omega - \bar{f}(\omega) =  F_{\ubar{p}}(\omega)
	$.  Based on $ S_{\textsf{offline}}(\mathcal{A}_{\ubar{p}}) $ and $ S_{\textsf{online}}(\mathcal{A}_{\ubar{p}}) $, the online algorithm is $ \alpha $-competitive indicates that 
	$ S_{\textsf{online}}(\mathcal{A}_{\ubar{p}}) = F_{\ubar{p}}(\omega) \geq \frac{1}{\alpha}S_{\textsf{offline}}(\mathcal{A}_{\ubar{p}}) = \frac{1}{\alpha} h\big(\ubar{p}\big) $. 
	Meanwhile, $ \omega $ must be less than or equal to $ \ubar{\rho} $ since otherwise the marginal cost is larger than the valuation density $ \ubar{p} $, i.e., $ \omega\in [0,\ubar{\rho}] $. Combining these results, we can conclude that under the arrival instance $ \mathcal{A}_{\ubar{p}} $, if there is an $ \alpha $-competitive online algorithm, then this algorithm must accept all the agents so that the total resource utilization is $ \omega $. We thus complete the proof of \textbf{NC1}. 
	
	(\textbf{Proof of NC2}) Our proof of \textbf{NC2} is based on constructing a special arrival instance $ \mathcal{A}_{p} $ such that any $ \alpha $-competitive online algorithm must satisfy the inequality in Eq. \eqref{ODI_psi} in order to achieve at least  $ \frac{1}{\alpha} $ of the offline optimal social welfare in hindsight. Specifically,  for any $ p \in  [\ubar{p}, \bar{p}] $, $ \mathcal{A}_{p} $ is constructed as follows.  \textit{We assume that $ \mathcal{A}_p $ consists of two groups of arrivals.  First, there is a group of identical agents whose  valuation densities are $ \ubar{p} $ and the total resource requirements of these identical agents are $ \omega $. In the following, we refer to the identical agents in this group by \textbf{Group-}$ (\ubar{p},\omega) $. Second, let us assume for each $\eta \in  [\ubar{p}, p] $, there is a group of agents parameterized by $ \eta $, where all the agents in group $ \eta $ are identical and has a total requirement of $ h'(\eta) $ (i.e., each customer's requirement is infinitely small and the total of them is $ h'(\eta) $). Meanwhile, the valuation densities of all the agents in group $ \eta $ are the same value of $ \eta $, namely, the total valuation of all the agents in this group is given by $ v_{\eta} = \eta  h'(\eta) $. We refer to the agents in the second continuum of groups by \textbf{Groups}-$ (\ubar{p},p) $}. Recall that based on Eq. \eqref{dual_bar_f_star}, $ h' $ is given by
	\begin{align}\label{h_prime_appendix}
	h'\left(\eta\right) =
	\begin{cases}
	f'^{-1}\left(\eta\right) &\text{if } \eta\in [\ubar{c},\bar{c}],\\
	1 &\text{if } \eta\in (\bar{c},+\infty),
	\end{cases}
	\end{align}
	Therefore, $ h'(\eta) $ is the maximum units of resource that can be supplied when the marginal cost is $ \eta $ per unit. 
	
	For a given arrival instance $ \mathcal{A}_{p} $ with any $ p\in [\ubar{p}, \bar{p}] $, the social welfares achieved by the optimal offline algorithm and the $ \alpha $-competitive online algorithm are given as follows:
	\begin{itemize}
		\item \textbf{Offline}: the optimal offline result in hindsight is to allocate $ h'(p) $ units of resources to the requests in the last group, i.e., group $ p $ in \textbf{Groups}-$ (\ubar{p},p) $, and none to all the previous requests. The optimal social welfare is thus
		\begin{align*}
		S_{\textsf{offline}}(\mathcal{A}_p) = p  h'(p) - f\left(h'(p)\right) = h(p),
		\end{align*}
		where the second equality comes from the Fenchel duality of $ h $. In fact, if we substitute $ h' $ from Eq. \eqref{h_prime_appendix} into the above $ S_{\textsf{offline}}(\mathcal{A}_p) $, the equation follows.
		
		\item \textbf{Online}: for the $ \alpha $-competitive online algorithm, let $ y = \psi(\eta) $ denote the total resource utilization after processing the customers in group  $ \eta \in  [\ubar{p}, p] $, and thus $ \psi(\eta) $ represents the resources allocated to all the groups of customers in $ [\ubar{p},\eta] $.
		\textit{According to \textbf{NC1}, all the requests in \textbf{Group-}$ (\ubar{p},\omega) $ must be satisfied. Therefore, we have $ \psi(\ubar{p}) = \omega $, meaning that when $ \eta = \ubar{p} $, the total resource utilization is $ \omega $, which equals the total resource requirements of all the requests in \textbf{Group-}$ (\ubar{p},\omega) $}.
		Meanwhile, $ \psi(\eta) $ is monotonically non-decreasing in $ \eta\in [\ubar{p}, p] $ since the resource allocations are irrevocable. The social welfare achieved by this online algorithm is thus the total valuation minus the total cost, namely,
		\begin{align*}
		S_{\textsf{online}}(\mathcal{A}_{\ubar{p}})
		= \ubar{p}\omega + \int_{\psi(\ubar{p})}^{\psi(p)} \eta d(\psi(\eta)) -   
		f\big(\psi(p)\big)  = \ubar{p}\omega + \int_{\ubar{p}}^{p} \eta \psi'(\eta)d\eta - f\big(\psi(p)\big).
		\end{align*}
	\end{itemize}
	
	Based on $ S_{\textsf{offline}}(\mathcal{A}_p) $ and $ S_{\textsf{online}}(\mathcal{A}_p) $, the online algorithm is $ \alpha $-competitive means that the following  inequality
	\begin{align}\label{integral_appendix}
	\ubar{p}\omega + \int_{\ubar{p}}^{p} \eta \psi'(\eta)d\eta - f\big(\psi(p)\big)  \geq \frac{1}{\alpha} h(p)
	\end{align}
	holds for all $ p \in [\ubar{c},\bar{p}] $. According to the definition of $ \psi $ and $ \bar{\rho} $, we have  $ \psi(p)\leq \bar{\rho} $, $ \forall p\in [\ubar{c},\bar{p}] $, and thus $ \psi(\bar{p})\leq \bar{\rho} $ holds as well. Therefore, if there exists an $ \alpha $-competitive online algorithm, then there must exist a non-decreasing function $ \psi(p) $ that satisfies Eq. \eqref{ODI_psi}.
\end{proof}

\subsection{Existence of Strictly-Increasing Allocation Functions and Pricing Functions}
Based on \textbf{NC2} in Theorem \ref{necessity_alpha_competitive},  Theorem \ref{strictly_increasing_psi} below shows that it is necessary to have a strictly-increasing segment $ [\omega,\bar{\rho}] $  for $ \phi $.
\begin{theorem}\label{strictly_increasing_psi}
	Given a convex setup $ \mathcal{S} $, if there exists an $ \alpha $-competitive online algorithm, then there must exist a positive value  $ \omega\in [F_{\ubar{p}}^{-1}\big(\frac{1}{\alpha} h(\ubar{p})\big),\ubar{\rho}] $ such that the following claims hold: 
	\begin{itemize}
		\item There exists a strictly-increasing function $ \psi(p) $ that satisfies
		\begin{align}\label{BVP_psi_p}
		\begin{cases}
		\psi'(p) = \frac{1}{\alpha}\cdot \frac{h'(p)}{p - f'(\psi(p))}, p\in (\ubar{p},\bar{p}),\\
		\psi(\ubar{p}) = \omega, \psi(\bar{p})\leq \bar{\rho}.
		\end{cases}
		\end{align}
		\item There exists a strictly-increasing function $ \varphi(y) $ that satisfies
		\begin{align}\label{BVP_necessary_appendix} 
		\begin{cases}
		\varphi'(y) =  \alpha \cdot\frac{\varphi(y) - f'(y)}{h'\left(\varphi(y)\right)}, y\in (\omega,\bar{\rho}),\\
		\varphi(\omega) = \ubar{p}, \varphi(\bar{\rho})\geq \bar{p}.
		\end{cases}
		\end{align}
		\item  $ \psi(p) $ and $ \varphi(y) $ are inverse to each other, i.e.,  $ \psi = \varphi^{-1} $ or $ \varphi = \psi^{-1} $.
	\end{itemize}
\end{theorem}

\begin{proof}
	We organize our proofs of Eq. \eqref{BVP_psi_p} and Eq. \eqref{BVP_necessary_appendix} as follows.
	
	{\color{black}	 
		(\textbf{Proof of Eq. \eqref{BVP_psi_p}}) We first prove the first bullet of Theorem \ref{strictly_increasing_psi}, namely, there exists a strictly-increasing function $ \psi(p) $ that satisfies Eq. \eqref{BVP_psi_p} as long as there exists an $ \alpha $-competitive online algorithm. Note that when $ \psi(\ubar{p}) = \omega$, we can derive based on the left-hand-side of Eq. \eqref{integral_appendix} that
		\begin{align*} 
		\ubar{p}\omega+
		\int_{\ubar{p}}^{p}\eta \psi'(\eta) d\eta - f\big(\psi(p)\big) = p\psi(p) - \int_{\ubar{p}}^{p} \psi(\eta)d\eta - f\big(\psi(p)\big). 
		\end{align*}
		Based on Theorem \ref{necessity_alpha_competitive} and the above equivalent transformation, we know that there always exists a non-decreasing function $ \psi(p) $ that satisfies the following inequality
		\begin{align}\label{gronwell_initial}
		\psi(p) \geq \frac{h(p) +\alpha f(\psi(\ubar{p}))}{\alpha p} +  \frac{1}{p}\int_{\ubar{p}}^{p}\Big(\psi(\eta)+f'\big(\psi(\eta)\big)\Big) d\eta 
		\end{align}
		and the boundary conditions: $ \{\psi\big(\ubar{p}\big) = \omega, \psi\big(\bar{p}\big)\leq \bar{\rho}\} $. For ease of exposition, let us write the right-hand-side of Eq. \eqref{gronwell_initial} in a more compact form as follows:
		\begin{align}\label{gronwall_inequality}
		\psi(p) \geq A(p) +  B(p)\int_{\ubar{p}}^{p} Q\big(\psi(\eta)\big) d\eta, p\in [\ubar{p},\bar{p}],
		\end{align}
		where $ A(p) $, $ B(p) $, and $ Q(\psi(\eta)) $ are given by
		\begin{align*}
		& A(p) = \frac{h(p) +\alpha f(\psi(\ubar{p}))}{\alpha p},\\
		& B(p) = \frac{1}{p},\\
		& Q(\psi(\eta)) = \psi(\eta)+f'(\psi(\eta)).
		\end{align*}
		
		It can be observed that Eq. \eqref{gronwall_inequality} is a nonlinear generalization of the standard Gronwell inequality \cite{Gronwell}. Thus, there must exist a tightest lower bound function $ \psi_L(p) $ such that
		\begin{align*}
		\psi(p) \geq \psi_L(p), p\in [\ubar{p},\bar{p}],
		\end{align*}
		namely, $ \psi_L $ is defined to be the tightest lower bound of all the $ \psi $'s that satisfy Eq. \eqref{gronwall_inequality} with $ \{\psi\big(\ubar{p}\big) = \omega, \psi\big(\bar{p}\big)\leq \bar{\rho}\} $. We emphasize that the function $ \psi_L $ is well-defined as there exists at least one non-decreasing function $ \psi(p) $  that satisfies Eq. \eqref{gronwall_inequality} with $ \{\psi\big(\ubar{p}\big) = \omega, \psi\big(\bar{p}\big)\leq \bar{\rho}\} $. Meanwhile, when $ Q $ is in some special forms (e.g., affine), $ \psi_L $ can be analytically given based on $ A, B $, and $ Q $ \cite{Gronwell}. Here, the nonlinearity of $ Q $ makes it difficult (if not impossible) to derive the exact form of $ \psi_L $. To aid our following proof, let us  define $ \Omega_\psi $ by
		\begin{align*}
		\Omega_\psi \triangleq \big\{\psi \big| \psi \text{ is a strictly-increasing function that satisfies Eq. \eqref{BVP_psi_p}}\big\}.
		\end{align*}
		The first bullet in Theorem \ref{strictly_increasing_psi} argues that $ \Omega_\psi $ is well-defined and non-empty. To prove this, in what follows  we show that $ \psi_L \in \Omega_\psi$, namely, $ \psi_L $ is a strictly-increasing function that satisfies Eq. \eqref{gronwall_inequality} by equality and $ \{\psi_L\big(\ubar{p}\big) = \omega, \psi_L\big(\bar{p}\big)\leq \bar{\rho}\} $. 
		
		Recall that $ \psi_L $ is defined to be the tightest lower bound of all the $ \psi $'s that satisfy Eq. \eqref{gronwall_inequality} with $ \{\psi\big(\ubar{p}\big) = \omega, \psi\big(\bar{p}\big)\leq \bar{\rho}\} $, and thus $ \psi_L(\bar{p}) \leq \bar{\rho} $ always holds. On the other hand, if $ \psi_L(\bar{p}) \leq \bar{\rho} $ does not hold, then it is obvious that there will be no solution to Eq. \eqref{ODI_psi} since $ \psi(p) \geq \psi_L(p) > \bar{\rho} $ holds for all $ p\in [\ubar{p},\bar{p}] $, which thus contradicts the necessary condition \textbf{NC2} given by Theorem \ref{necessity_alpha_competitive}. 
		For this reason, we only need to prove that $ \psi_L(p) $ is strictly-increasing in $ p\in [\ubar{p}, \bar{p}] $ and $ \psi_L(\ubar{p}) = \omega $. In fact, the definition of $ \psi_L $ indicates that the following equation of $ \psi_L(p) $ holds for all $ p\in [\ubar{p}, \bar{p}] $:
		\begin{align}\label{equality_psi_L} 
		\psi_L(p) = A(p) +  B(p)\int_{\ubar{p}}^{p} Q\Big(\psi_L(\eta)\Big) d\eta.
		\end{align}
		This is because,  if there exists any $ \hat{p}\in[\ubar{p}, \bar{p}] $ such that Eq. \eqref{equality_psi_L} does not hold at $ p = \hat{p} $, then we can always lower the value of $ \psi_L(\hat{p}) $ to reach the equality and make the smaller value as the new output of $ \psi_L(\hat{p}) $. Based on Eq. \eqref{equality_psi_L}, we can prove  that $ \psi_L(p) $ is strictly-increasing in $ p\in [\ubar{p}, \bar{p}] $. Our proof is based on contradiction and the details are as follows. As illustrated by the dashed line in Fig. \ref{fig_construction_psi}, if we assume that $  \psi_L $ is not strictly monotone in $ p\in [\ubar{p}, \bar{p}]  $, then there must exist a point, say $ p_0 \in [\ubar{p}, \bar{p}] $, so that $ \psi_L'(p_0) = 0 $. Recall that $ h'(p)>0 $ holds for all $ p \in [\ubar{p}, \bar{p}] $ as the conjugate function $ h(p) $ is strictly-convex in $ p\in [\ubar{p}, \bar{p}] $. Thus, we can derive from Eq. \eqref{equality_psi_L} that
		\begin{align*}
		\psi_L'(p) = \frac{1}{\alpha}\cdot \frac{h'(p)}{p - f'(\psi_L(p))} \neq 0,  p \in [\ubar{p}, \bar{p}],
		\end{align*}
		which contradicts $ \psi_L'(p_0) = 0 $. Therefore, $ \psi_L $ is either strictly-increasing or strictly-decreasing. Note that our design of $ \omega\in [F_{\ubar{p}}^{-1}\big(\frac{1}{\alpha} h(\ubar{p})\big),\ubar{\rho}] $ guarantees that $ \ubar{p} - f'(\omega) > 0$. Therefore, based on the Picard–Lindel\"{o}f theorem  \cite{ODE1973, Perko2001}, there always exists a unique strictly-increasing function $ \psi_L $ that satisfies the following first-order initial value problem (IVP)\footnote{\color{black}We give a brief overview of first-order IVPs in Appendix \ref{background_BVP_IVP} and Appendix \ref{basic_lemmas}.}:
		\begin{align}\label{IVP_psi_L}
		\begin{cases}
		\psi_L'(p) = \frac{1}{\alpha}\cdot \frac{h'(p)}{p - f'(\psi_L(p))} \neq 0,  p \in (\ubar{p}, \bar{p}),\\
		\psi_L(\ubar{p}) =\omega.
		\end{cases}
		\end{align}
		As a result, $ \psi_L $ must be strictly-increasing.  Recall that $ \psi_L(\bar{p})\leq \bar{\rho} $ always holds, together with Eq. \eqref{IVP_psi_L},  we can conclude that $ \psi_L\in \Omega_\psi $, which means that $ \Omega_\psi $ is well-defined and non-empty. Thus, we complete the proof of Eq. \eqref{BVP_psi_p}.
		
		\begin{figure}
			\centering
			\includegraphics[height=6cm]{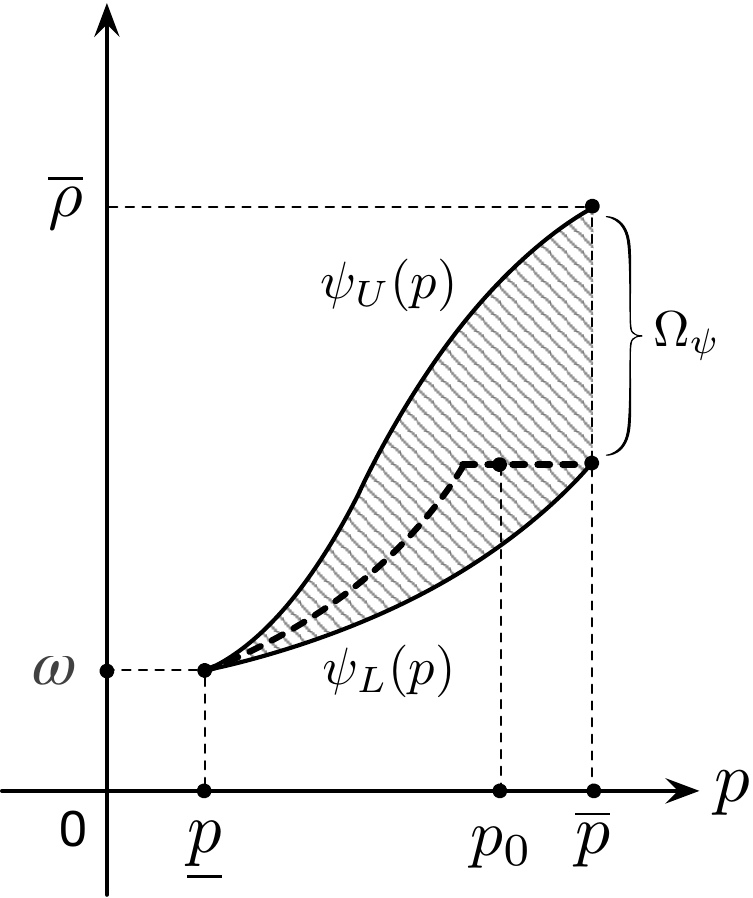}
			\caption{\color{black} Illustration of $ \Omega_\psi$ for given $ \ubar{p} $, $ \bar{p} $, $ \omega $, and $ \ubar{\rho} $. In the figure, $ \psi_L $ and $ \psi_U $ denote the tightest lower and upper boundaries of $ \Omega_\psi $, respectively. }
			\label{fig_construction_psi}
		\end{figure}
		
		\begin{remark}\color{black}
			Our above proof shows that the tightest lower bound function $ \psi_L(p) $ is strictly-increasing in $ p\in [\ubar{p},\bar{p}] $ and satisfies Eq. \eqref{BVP_psi_p}. If we define similarly that $ \psi_U $ is the tightest upper bound of all the $ \psi $'s that satisfy Eq. \eqref{gronwall_inequality} with $ \{\psi\big(\ubar{p}\big) = \omega, \psi\big(\bar{p}\big)\leq \bar{\rho}\} $, then we can prove a stronger result: for any given $ \alpha\geq 1 $, if there exists an $ \alpha $-competitive online algorithm,  then there exist infinitely-many strictly-increasing functions that satisfy Eq. \eqref{BVP_psi_p}, unless $ \psi_L = \psi_U $ since in this case $ \Omega_\psi $ is a singleton, i.e., $ \Omega_\psi = \{\psi_L\} $.  
			Geometrically, when $ \alpha $ is larger, the shaded area in Fig. \ref{fig_construction_psi} will be expanded as $ \psi_U(\bar{p}) -  \psi_L(\bar{p}) $ becomes larger. In contrast, a smaller $ \alpha $ indicates that $ \psi_U $ and $ \psi_L $ are closer to each other, i.e., $ \psi_U(\bar{p}) -  \psi_L(\bar{p}) $ is smaller, and in the extreme case when $ \psi_U = \psi_L $, the corresponding value of $ \alpha $ is the minimum competitive ratio that any online algorithm could possibly achieve.  For the detailed discussion of the optimal/minimum competitive ratio, please refer to the proof of Theorem \ref{major_results} in Section \ref{proof_of_major_results}. 
		\end{remark}
	}
	
	(\textbf{Proof of Eq. \eqref{BVP_necessary_appendix}}) Now we are ready to prove Eq. \eqref{BVP_necessary_appendix}.     
	Since $ \psi(p) $ is strictly-increasing, the inverse of $ \psi(p) $ is well-defined. Let us denote the inverse of $y =  \psi(p) $  by $ p = \psi^{-1}(y) \triangleq \varphi(y) $. The derivative of $ \psi(p) = \varphi^{-1}(p) $ can be written as 
	\begin{align}\label{psi_p_prime}
	\psi'(p) = \frac{1}{\varphi'\left(\varphi^{-1}(p)\right)} = \frac{1}{\varphi'(y)}.
	\end{align} 
	Substituting Eq. \eqref{psi_p_prime} into Eq. \eqref{BVP_psi_p} leads to the following ODE
	\begin{align*}
	\frac{1}{\varphi'(y)}
	=  \frac{1}{\alpha}\cdot \frac{h'(\varphi(y))}{\varphi(y)-f'(y)}.
	\end{align*}
	Therefore, we have the following ODE in terms of $ \varphi $:
	\begin{align*}
	\varphi'(y) =  \alpha \cdot\frac{\varphi(y) - f'(y)}{h'\big(\varphi(y)\big)}, y\in (\omega,\bar{\rho}),
	\end{align*}
	which is exactly the ODE in Eq. \eqref{BVP_necessary_appendix}. 
	Based on the two boundary conditions of $ \psi(p) $ given in Eq. \eqref{ODI_psi}, the BVP in Eq. \eqref{BVP_necessary_appendix} directly follows. Based on the above proofs of Eq. \eqref{BVP_psi_p} and Eq. \eqref{BVP_necessary_appendix}, we thus complete the proof of Theorem \ref{strictly_increasing_psi}. 
\end{proof}

Based on \textbf{NC1} in Theorem \ref{necessity_alpha_competitive} and the two BVPs in Theorem \ref{strictly_increasing_psi}, Theorem \ref{necessity} follows.

\section{Proof of Corollary \ref{major_results_log}}\label{proof_of_major_results_appendix}
The proof is simple and intuitive by following our previous analysis of the ODE with boundary conditions. 

(\textbf{Zero Cost}: $ q = 0 $) When $ f(y) = 0 $, i.e., there is no supply cost, the extended cost function is given as follows:
\begin{align*}
\bar{f}(y) = 
\begin{cases}
0 & \text{if } y\in [0, 1],\\
+\infty  &  \text{if } y\in (1,+\infty),
\end{cases}
\end{align*}
Recall that the conjugate function $ h(p) $ is given by 
$ h(p) = \max_{y \geq 0}\ p y    - \bar{f}(y) = \max_{y\geq 0}\ F_p(y) $,
which indicates that $ h(p) 
= p, \forall p\in [0,+\infty) $.
Therefore, the BVP given by Eq. \eqref{BVP_necessary_varphi} can be written as follows:
\begin{align*}
\begin{cases}
\varphi'(y) =  \alpha\cdot \varphi(y), y\in (\omega,\bar{\rho}),\\
\varphi(\omega) = \ubar{p}, \varphi(1)\geq \bar{p}.
\end{cases}
\end{align*}
Based on Eq. \eqref{flat_necessary_omega}, we know that $ \omega\geq \frac{1}{\alpha} $.  
Solving the above ODE leads to the following solution
\begin{align*}
\varphi(y) = \delta \cdot \exp\left(\alpha y\right),
\end{align*}
where $ \delta $ can be any real scalar. Based on the two boundary conditions, to minimize $ \alpha $, we set
\begin{align*}
\delta\exp\left(\alpha \omega\right) = \ubar{p}, \delta \exp\left(\alpha \right) = \bar{p}, \alpha = \frac{1}{\omega}. 
\end{align*}
Solving the above equations leads to $ \delta = \ubar{p}/e $, where $ e $ denotes the base of the natural logarithm. Therefore, the final optimal pricing function is
\begin{align*}
\phi_*(y) = 
\begin{cases}
\ubar{p} &\text{if } y\in [0,\omega_*),\\
\ubar{p}\cdot\exp\left(y/\omega_* - 1\right)  &\text{if } y\in [\omega_*, 1],\\
+\infty &\text{if } y\in (1,+\infty),
\end{cases}
\end{align*}
where $ \omega_* $ and $ \alpha_* $ are given by
\begin{align*}
\alpha_* = 1 + \ln\left(\frac{\bar{p}}{\ubar{p}}\right), \omega_* = \frac{1}{1 + \ln\left(\bar{p}/\ubar{p}\right)}.
\end{align*}

(\textbf{Linear Cost}: $ q> 0$) 
When $ f'(y) = q $ and $ q>0 $, namely $ f(y) = qy $ when $ y\in [0,1] $. The conjugate function $ h(p) $ is given by
\begin{align*}
h(p) =
\begin{cases}
0  &\text{if } p\in [0,q],\\
p -q   &\text{if } p\in (q,+\infty).
\end{cases}
\end{align*}

Based on Eq. \eqref{flat_necessary_omega},  we  have $ \ubar{p} \omega -f(\omega) \geq \frac{1}{\alpha} \big(\ubar{p}-f(1)\big) $, 
which thus indicates that $ \omega\geq \frac{1}{\alpha} $. The differential equation of interest is thus  given by
\begin{align*}
\begin{cases}
\varphi(y) - f'(y) 
\geq \frac{1}{\alpha} \cdot   \varphi'(y)\cdot h'\left(\varphi(y)\right), \forall y\in (\omega,1),\\
\varphi(\omega) = \ubar{p},\varphi(1) = \bar{p},
\end{cases}
\end{align*}
which leads to the following general solution: $ \varphi(y) = \delta \cdot\exp\left(\alpha y\right)+q $. Substituting the two boundary conditions into the general solution
leads to $ \varphi(\omega) = \delta\cdot \exp(\alpha \omega) +q  = \ubar{p}, \varphi(1) = \delta e^{\alpha} + q  = \bar{p} $.
To minimize $ \alpha $, it suffices to set  $ \omega =  1/\alpha  $ and thus we have
\begin{align*}
\phi_*(y) = 
\begin{cases}
\ubar{p} &\text{if } y\in [0,\omega_*),\\
(\ubar{p}-q)\cdot\exp\left(y/\omega_* - 1\right) + q &\text{if } y\in [\omega_*, 1],\\
+\infty &\text{if } y\in (1,+\infty),
\end{cases}
\end{align*}
where $ \omega_* $ and $ \alpha_* $ are given by
\begin{align*}
\alpha_* = 1+ \ln\left(\frac{\bar{p}-q}{\ubar{p}-q}\right),\omega_* = \frac{1}{1+ \ln\left(\frac{\bar{p}-q}{\ubar{p}-q}\right)}.
\end{align*}

Combining the results regarding $ q = 0 $ and $ q>0 $, we complete the proof of Corollary \ref{major_results_log}.

\section{Proof of Corollary \ref{main_result_any}}\label{proof_of_main_result_any}
Based on Theorem \ref{sufficiency} and Eq. \eqref{alpha_larger_than_a_ratio}, the worst-case competitive ratio of $ \textsf{PM}_\phi $ satisfies:
\begin{itemize}
	\item When the final resource utilization $ \rho\in [0, \omega] $, the worst-case competitive ratio is
	\begin{align*}
	\alpha_1 =     \frac{h\big(\ubar{p}\big)}{F_{\ubar{p}}(\omega)}.  
	\end{align*}
	This corresponds to the arrival instance that all agents have the minimum valuation density $ \ubar{p} $. 
	\item When the final resource utilization $\rho =  \bar{\rho}_\varphi $, the worst-case competitive ratio is
	\begin{align*}
	\alpha_2 =   \frac{h(\bar{p})}{\ubar{p}\omega+\int_{\omega}^{\bar{\rho}_\varphi}\varphi(y) dy - f\left(\bar{\rho}_{\varphi}\right)}.
	\end{align*}
	This corresponds to the worst-case scenario when the first group of agents  are all satisfied by $ \textsf{PM}_\phi $ and all of them obtain zero utilities (i.e., the valuation equals the payment), while the offline optimal social welfare is to satisfy all the agents in the second group whose valuation densities are all $ \bar{p} $. 
	
	\item When the final resource utilization $ \rho\in [\omega,\bar{\rho}_\varphi) $, the worst-case competitive ratio is
	\begin{align*}
	\alpha_3 = \max_{\rho\in [\omega,\bar{\rho}_\varphi)}\  \frac{h(\varphi(\rho))}{\ubar{p}\omega+\int_{\omega}^{\rho}\varphi(y) dy - f\left(\rho\right)},
	\end{align*}
	which comes from Eq. \eqref{alpha_larger_than_a_ratio}. The rationality of $ \alpha_3 $ is based on analyzing the competitive ratio of  $ \textsf{PM}_\phi $ under a series of worst-case scenarios constructed in Section \ref{section_worst_case}.
\end{itemize}

Combining the above three cases, we have
\begin{align*}
\alpha(\omega,\varphi) =\ & \max\{\alpha_1,\alpha_2,\alpha_3\} \\
=\ & \max\Bigg\{\frac{h\big(\ubar{p}\big)}{F_{\ubar{p}}(\omega)}, \frac{h(\bar{p})}{\ubar{p}\omega+\int_{\omega}^{\bar{\rho}_\varphi}\varphi(y) dy -  f\left(\bar{\rho}_{\varphi}\right)},  \max_{\rho\in [\omega,\bar{\rho}_{\varphi}]}\  \frac{h\big(\varphi(\rho)\big)}{\ubar{p}\omega+\int_{\omega}^{\rho}\varphi(y) dy - f\left(\rho\right)} \Bigg\}.
\end{align*}
We thus complete the proof of Corollary \ref{main_result_any}. Note that when $ \varphi(1)\geq \bar{p} $, the second term (i.e., $ \alpha_2 $) is not needed as it is contained by  $ \alpha_3 $.

\section{Proof of Proposition \ref{two_lower_bounds_case_1}}\label{proof_of_two_lower_bounds_case_1}
\subsection{Background of First-Order Two-Point BVPs}\label{background_BVP_IVP}
For a better reference, we revisit the two BVPs in Corollary \ref{theorem_case_1_two_BVP} as follows:
\begin{align}\label{BVP_1_case_1_appendix}
\normalfont\textsf{BVP}_1(\omega,u,\alpha) 
\begin{cases}
\varphi_1'(y) = \alpha\cdot \frac{\varphi_1(y)  - f'(y) }{f'^{-1}\left(\varphi_1(y)\right)}, y\in(\omega,u);\\
\varphi_1(\omega) = \ubar{p}, \varphi_1(u) =\bar{c},	
\end{cases}
\end{align}
\begin{align}\label{BVP_2_case_1_appendix}
\normalfont\textsf{BVP}_2(u,\alpha)
\begin{cases}
\varphi_2'(y) = \alpha\cdot \left(\varphi_2(y) - f'(y) \right),  y\in (u,1);\\
\varphi_2(u) = \bar{c},  \varphi_2(1) \geq \bar{p}.
\end{cases}
\end{align}
A solution to the above first-order two-point BVPs is a function $ \varphi_1(y) $ or $ \varphi_2(y) $ that satisfies the corresponding ODE and also the two boundary conditions.

In the field of differential equations, when there is only one additional condition other than the ODE, the resulting problem is an initial value problem (IVP). For example, the following IVP is derived from $ \textsf{BVP}_1(\omega,u,\alpha) $ by excluding the first boundary condition:
\begin{align}\label{IVP_1_case_1_appendix}
\normalfont\textsf{IVP}_1(\omega,u,\alpha) 
\begin{cases}
\varphi_1'(y) = \alpha\cdot \frac{\varphi_1(y)  - f'(y) }{f'^{-1}\left(\varphi_1(y)\right)}, y\in(\omega,u);\\
\varphi_1(u) =\bar{c},	
\end{cases}
\end{align}
where $ \varphi_1(u) =\bar{c} $ is usually termed as the initial condition. We denote the solution to the above IVP by $ \varphi_{\mathrm{ivp}}\left(y; \omega, u, \alpha\right)  $ (if it exists). Meanwhile, we denote the first derivative of  $ \varphi_{\mathrm{ivp}}\left(y; \omega, u, \alpha\right)  $ w.r.t $ y\in (\omega,u) $ by  $ \varphi_{\mathrm{ivp}}'\left(y; \omega, u, \alpha\right)  $.

The key to the analysis of IVPs and BVPs is the existence and uniqueness of solutions. For first-order IVPs, the existence and uniqueness theorem is well understood. In particular,  the Picard–Lindel\"{o}f theorem  \cite{ODE1973, Perko2001} guarantees the unique exsitence of solutions as long as the ODE satisfies a certain Lipschitz continuity conditions. Meanwhile, there are plenty of iterative  methods off-the-shelf that can solve IVPs numerically.  However,  for BVPs, there is no general uniqueness and existence theorem. As argued by \cite{Perko2001}, it is even non-trivial to obtain numerical solutions for some BVPs in the most basic two-point case as Eq. \eqref{BVP_1_case_1_appendix} and Eq. \eqref{BVP_2_case_1_appendix}.  

Our optimal pricing function design in Theorem \ref{major_results} relies on the uniqueness and existence property of first-order two-point BVPs. On the one hand, based on the existence and uniqueness property, we can prove that there always exists a competitive pricing function as long as the competitive ratio $ \alpha $ satisfies a certain condition. On the other hand, we can prove that there exists no other online algorithms that  can outperform $ \textsf{PM}_{\phi} $ with a smaller competitive ratio since the corresponding BVP has no solution.

\subsection{Basic Lemmas}\label{basic_lemmas}
Our pricing function design is related to the existence, uniqueness  and monotonicity property of solutions to $ \textsf{IVP}_1(\omega, u,  \alpha) $.  Note that to simplify the analysis below, we temporarily assume that $ \omega\in [0,\ubar{\rho}] $ is a given parameter. Later in our proofs of the existence of $ \Gamma_1(\omega,u) $ and $ \Gamma_2(u) $, we will include the discussions of $ \omega $ as a design variable.

\begin{lemma}\label{existence_uniquess}
	For each $ (\alpha, u)\in [1,+\infty)\times (\omega,1] $,  $\normalfont\textsf{IVP}_1(\omega, u,  \alpha) $ has a unique solution  $ \varphi_{\mathrm{ivp}}\left(y; \omega, u, \alpha\right)  $ that is defined over $y\in [\omega,u] $.
\end{lemma}

Lemma \ref{existence_uniquess} follows  the  Picard–Lindel\"{o}f theorem for the existence and uniqueness of solutions to IVPs.  We refer the details to  \cite{uniqueness_book1993, ODE1973, Perko2001}. Basically the Picard–Lindel\"{o}f theorem guarantees that there always exists a unique solution to $ \textsf{IVP}_1(\omega, u,  \alpha)  $, defined on a small neighbourhood of the initial point $ \varphi(u) = \bar{c} $, as long as the ODE in $ \textsf{IVP}_1(\omega, u,  \alpha)  $ is Lipschitz continuous within that neighbourhood. Moreover, this unique solution extends to the whole region of $ y \in [\omega,u] $. Based on Lemma \ref{existence_uniquess}, we can prove the following Lemma \ref{monotonicity_omega} and Lemma \ref{monotonicity_alpha}.

\begin{lemma}\label{monotonicity_omega}
	Given $ \alpha\geq 1 $, $\varphi_{\mathrm{ivp}}\left(y; \omega, u, \alpha\right)$ is increasing in $ y\in[\omega,u] $ and lower bounded by $ f'(y) $ for all $ y \in  [\omega,u] $.
\end{lemma}
\begin{proof}
	We first prove that $  \varphi_{\mathrm{ivp}}\left(y; \omega, u, \alpha\right) \geq  f'(y) $ holds for all $ y\in [\omega,u] $. Note that when $ y = u\in (0,1) $, we have $  \varphi_{\mathrm{ivp}}(u;\omega,u, \alpha) >  f'(u) $ and $ \varphi_{\mathrm{ivp}}'(u;\omega,u, \alpha) >0 $. Therefore, if $  \varphi_{\mathrm{ivp}}\left(y; \omega, u, \alpha\right)\geq f'(y) $ does not hold for all $ y\in [\omega,u] $, then there must exists at least one point within $ (\omega,u) $, say $ y_0\in (\omega,u) $, such that $ \varphi_{\mathrm{ivp}}\left(y; \omega, u, \alpha\right) >  f'(y) $ for all $ y\in (y_0,u] $, $ \varphi_{\mathrm{ivp}}(y_0; \omega, u, \alpha) =   f'(y_0) $, and $ \varphi_{\mathrm{ivp}}\left(y; \omega, u, \alpha\right) < f'(y) $ for all $ \omega\in (y_0-\epsilon,y_0) $, where $ \epsilon $ is a small positive value. However, when $ \varphi_{\mathrm{ivp}}\left(y; \omega, u, \alpha\right) < f'(y) $, $ \varphi_{\mathrm{ivp}}'(y;\omega,u,\alpha) $ is negative according to the ODE, and thus $ \varphi_{\mathrm{ivp}}\left(y; \omega, u, \alpha\right) $ is decreasing in $ (y_0-\epsilon,y_0) $. This means that $ \varphi_{\mathrm{ivp}}\left(y; \omega, u, \alpha\right) > \varphi_{\mathrm{ivp}}(y_0;\omega,u,\alpha) = f'(y_0)>f'(y) $ for all $y\in (y_0-\epsilon,y_0)  $, leading to a contradiction. Therefore, $  \varphi_{\mathrm{ivp}}\left(y; \omega, u, \alpha\right) \geq  f'(y) $ always holds for all $ y\in [\omega,u] $, and the monotonicity directly follows. 
\end{proof}

Lemma \ref{monotonicity_omega} guarantees that for any $ \alpha\geq 1 $ and $ u\in (\omega,1) $, the unique solution to $ \textsf{IVP}_1(\omega, u,  \alpha)  $ is a monotone pricing function. Below we give Lemma \ref{monotonicity_alpha}, which states that $\varphi_{\mathrm{ivp}}\left(y; \omega, u, \alpha\right)  $ is also monotonic in $ (u,\alpha)\in (\omega,1)\times [1,+\infty) $.

\begin{lemma}\label{monotonicity_alpha}
	$\varphi_{\mathrm{ivp}}\left(y; \omega, u, \alpha\right) $ is  continuous and non-increasing   in $ (u,\alpha)\in (\omega,1)\times [1,+\infty) $.
\end{lemma}
\begin{proof}
	The continuity directly follows since $\varphi_{\mathrm{ivp}}\left(y; \omega, u, \alpha\right)$ is well defined for all  $ (u,\alpha)\in (\omega,1)\times [1,+\infty) $.
	
	We first prove the monotonicity in $ u\in(\omega,1)$  by contradiction. Suppose we have $ u_1\in (\omega,1) $ and $ u_2\in (\omega,1) $, and assume w.l.o.g. that $ u_1>u_2 $, we can prove that  $ \varphi_{\mathrm{ivp}}(y;\omega,u_1,\alpha) <  \varphi_{\mathrm{ivp}}(y;\omega,u_2,\alpha) $ holds for all $ y\in (\omega,u_2) $. The idea is that these two functions cannot have any intersection point, since otherwise the IVP with the same ODE as  $ \textsf{IVP}_1(\omega, u,  \alpha)  $ but with the initial condition defined at the intersection point will have at least two solutions, namely  $ \varphi_{\mathrm{ivp}}(y;\omega,u_1,\alpha) $ and $ \varphi_{\mathrm{ivp}}(y;\omega,u_2,\alpha) $, which is impossible due to the uniqueness property given by Lemma \ref{existence_uniquess}. Note that it is also impossible to have $ \varphi_{\mathrm{ivp}}(y;\omega,u_1,\alpha) > \varphi_{\mathrm{ivp}}(y;\omega,u_2,\alpha) $ since if this is the case, then $ \varphi_{\mathrm{ivp}}(y;\omega,u_1,\alpha)  $ is not monotonic in $ y\in (\omega,u_1) $. Therefore, when $ u_1>u_2 $, we always have $ \varphi_{\mathrm{ivp}}(y;\omega,u_1,\alpha) < \varphi_{\mathrm{ivp}}(y;\omega,u_2,\alpha) $. 
	
	We now prove the monotonicity in $ \alpha \in [1,+\infty) $. Suppose we have $ \alpha_1 $ and $ \alpha_2 $, and assume w.l.o.g. that $ \alpha_1>\alpha_2 $. We need to prove that $ \varphi_{\mathrm{ivp}}(y;\omega,u,\alpha_1) <  \varphi_{\mathrm{ivp}}(y;\omega,u,\alpha_2) $ for all $ y \in (\omega,u] $.  Based on the ODE of \eqref{IVP_1_case_1_appendix}, when $ \alpha _1>\alpha_2 $, the derivative of $ \varphi $ at $ y = u $ satisfies
	\begin{align*}
	\varphi_{\mathrm{ivp}}'(u;\omega, u,\alpha_1) > \varphi_{\mathrm{ivp}}'(u;\omega, u,\alpha_2).  
	\end{align*}
	Therefore, there must exist a small interval on the left-hand-side of $ u $, say $ [u-\epsilon,u] $, where $ \epsilon $ is a small positive value, such that $ \varphi_{\mathrm{ivp}}(y;\omega, u, \alpha_1)< \varphi_{\mathrm{ivp}}(y;\omega, u, \alpha_1)$ for all $ y \in  [u-\epsilon,u] $. This can be easily proved based on the definition of the derivative, which is omitted for brevity.
	
	Now suppose $ \varphi_{\mathrm{ivp}}(y;\omega,u,\alpha_1) <  \varphi_{\mathrm{ivp}}(y;\omega,u,\alpha_2) $ does not hold for all $ y\in [\omega,u] $, then there must exist an intersection point $ u_0 $, such that $ \varphi_{\mathrm{ivp}}(y;\omega,u,\alpha_1)< \varphi_{\mathrm{ivp}}(y;\omega,u,\alpha_2)$ when $ y\in (u_0,u] $, and $ \varphi_{\mathrm{ivp}}(y;\omega,u,\alpha_1)\geq  \varphi_{\mathrm{ivp}}(y;\omega,u,\alpha_2)$ when $ y \in (u_0-\epsilon,u_0] $, where $ \epsilon $ is a very small positive value. Now let us consider two new IVPs with the same initial condition defined at point $ y = u_0 $, and denote the unique solutions to these two new IVPs by $ \varphi_{\mathrm{new}}(y;\omega,u_0,\alpha_1) $ and $ \varphi_{\mathrm{new}}(y;\omega,u_0,\alpha_2) $, according to the uniqueness property, we must have
	\begin{align*}
	\varphi_{\mathrm{new}}(y;\omega,u_0,\alpha_1) = \varphi_{\mathrm{ivp}}(y;\omega,u,\alpha_1),\forall y\in (\omega,u_0),\\
	\varphi_{\mathrm{new}}(y;\omega,u_0,\alpha_2)= \varphi_{\mathrm{ivp}}(y;\omega,u,\alpha_2), \forall y\in (\omega,u_0).
	\end{align*}
	Since $ \varphi_{\mathrm{ivp}}(y;\omega,u,\alpha_1)\geq  \varphi_{\mathrm{ivp}}(y;\omega,u,\alpha_2)$ when $ y \in (u_0-\epsilon,u_0] $, which means that we cannot find a small interval on the left-hand-side of $ u_0 $, say $ [u_0-\hat{\epsilon},u_0] $, such that $ \varphi_{\mathrm{new}}(y;\omega,u_0,\alpha_1)< \varphi_{\mathrm{new}}(y;\omega,u_0,\alpha_2)$. However, this contradicts with the fact that
	\begin{align*}
	\varphi_{\mathrm{new}}'(u_0;\omega,u_0,\alpha_1) > \varphi_{\mathrm{new}}'(u_0;\omega,u_0, \alpha_2).  
	\end{align*}
	Therefore, we have $ \varphi_{\mathrm{ivp}}(y;\omega,u,\alpha_1) <  \varphi_{\mathrm{ivp}}(y;\omega,u,\alpha_2) $ for all $ y\in (\omega,u] $. We thus complete the proof.
\end{proof}

\subsection{Proof of the Existence of $ \Gamma_1(\omega,u) $}
We now prove the unique existence of a well-defined function $  \Gamma_1(\omega,u) $ in the first bullet of Proposition \ref{two_lower_bounds_case_1}.

Our proof is based on the idea of ``shooting method" in ODE \cite{ODE1973, Perko2001}.  Here we briefly explain the intuition. 
For each given $ \omega\in [0,\ubar{\rho}] $ and $ u\in (\omega, 1) $, based on the ODE in Eq. \eqref{IVP_1_case_1_appendix}, when $ \alpha $ approaches infinity, $ \varphi_1(y) $ must be very close to $ f(y) $ when $ y $ approaches $\omega $ from the right, namely, $ \varphi_1(\omega) $ approaches $ f'(\omega) $. In this case, $ \varphi_1(\omega)\leq \ubar{p} $.  Similarly, when $ \alpha $ is very small, say approaches zero (our definition of $ \alpha $ requires that $ \alpha\geq 1 $, here we purely focus on the numerical analysis of the IVP given by Eq. \eqref{IVP_1_case_1_appendix}), then $ \varphi_1(\omega) $ must be very close to $ \bar{c} $ since $ \varphi_1'(y) $ will be close to 0 and thus $ \varphi_1(y) $ will be almost horizontal. In this case, $ \varphi_1(\omega)\geq \ubar{p} $. The monotonicity property in Lemma \ref{monotonicity_alpha} indicates that there must exist a unique $ \alpha $ in between so that $ \varphi(\omega) $ is exactly equal to $ \ubar{p} $. In the ``shooting method", this indicates the value setting of $ \alpha $ hits the target of $ \varphi(\omega) = \ubar{p} $. In summary, for a given $ \omega\in [0,\ubar{\rho}] $ and $ u\in (\omega, 1) $, there exists a unique $ \alpha $ so that the unique solution to the IVP in Eq. \eqref{IVP_1_case_1_appendix} is the unique solution to the BVP in Eq. \eqref{BVP_1_case_1_appendix}. We thus complete the proof of the existence of $ \alpha = \Gamma_1(\omega,u) $.

\subsection{Proof of the Existence of  $ \Gamma_2(u) $}
Proving the unique existence of $ \alpha = \Gamma_2(u) $ is more straightforward as the ODE in Eq. \eqref{BVP_2_case_1_appendix} can be solved in analytical forms. To be more specific, the general solution can be written as follows:
\begin{align*}
\varphi_2(y) = \exp(\alpha y)\cdot \left( \int^{y} \frac{\alpha f'(\eta)}{\exp(\alpha\eta)} d\eta + C\right),
\end{align*}
where $ C $ is any real constant. Substituting the two boundary conditions into the above general solution leads to the following
\begin{align*}
\varphi_2(u) = \exp(\alpha\cdot u)\cdot \left( \int^{u} \frac{\alpha f'(\eta)}{\exp(\alpha\eta)} d\eta + C\right) = \bar{c},\\
\varphi_2(1) = \exp(\alpha\cdot 1)\cdot \left( \int^{1} \frac{\alpha f'(\eta)}{\exp(\alpha\eta)} d\eta + C\right) \geq \bar{p},
\end{align*}
which thus indicates that
\begin{align}\label{inequality_u_alpha_Gamma_2}
\int_u^1\frac{\alpha f'(\eta)}{\exp(\alpha\eta)}d\eta \geq \frac{\bar{p}}{\exp(\alpha)} - \frac{\bar{c}}{\exp(u\alpha)}
\end{align}

We can prove that for each given $ u\in (0,1) $ , there exists a unique $ \alpha $ that satisfies Eq. \eqref{inequality_u_alpha_Gamma_2} by equality as long as $ \bar{p}>\bar{c} $. The proof is similar to that of $ \alpha = \Gamma_1(\omega,u) $ as any given setup of $ \bar{p}>\bar{c} $ will guarantee the existence of a monotone function $ \varphi_2(y) $ in the interval $ [u,1] $ with an appropriate setting of $ \alpha $. Moreover, this $ \alpha $ must be unique as otherwise the uniqueness property in Lemma \ref{existence_uniquess} would be contradicted. Therefore, if we denote this unique $ \alpha $ by $ \alpha = \Gamma_2(u) $, then we complete the proof of the existence of $ \Gamma_2(u) $. However, note that in the second bullet of Proposition \ref{two_lower_bounds_case_1}, we argue that for any $ \alpha\geq \Gamma_2(u) $, the BVP in Eq. \eqref{BVP_2_case_1_appendix} will have a unique monotone solution. This can be proved by the monotonicity of $ \varphi_2(y) $ w.r.t. $ \alpha $, similar to our proof of Lemma \ref{monotonicity_alpha}, and thus we skip the details for brevity. In summary, for each given $ u\in (0,1) $ , there exists a unique solution to $ \textsf{BVP}_2(u,\alpha) $ as long as  $ \alpha \geq  \Gamma_2(u)$. We thus complete the proof of Proposition \ref{two_lower_bounds_case_1}.

\begin{figure}
	\centering
	\includegraphics[width= 6 cm]{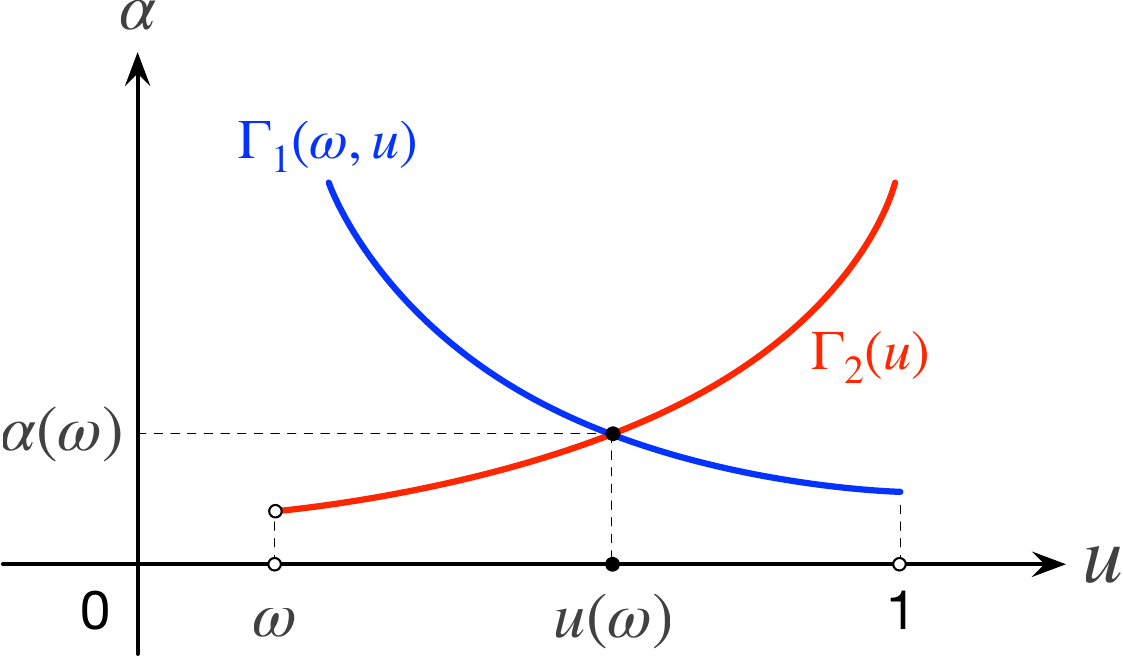}
	\caption{Illustration of the unique existence of the intersection point $ u(\omega) $ between $ \Gamma_1(\omega,u) $ and $ \Gamma_2(u) $ for a given $ \omega\in (0,\ubar{\rho}) $.}
	\label{two_alpha_intersection}
\end{figure}

\section{Proof of Theorem \ref{existence_uniqueness_opt_case_1}}\label{proof_of_existence_uniqueness_opt_case_1}
The proof of Theorem \ref{existence_uniqueness_opt_case_1} heavily relies on the monotonicity properties of $ \phi $ w.r.t $ \omega, u $, and $ \alpha $. To achieve the minimum $ \alpha $, we make the following claims:
\begin{itemize}
	\item \textbf{Claim 1}: $ \omega $ must be as small as possible. Since a larger $ \omega $ will shrink the interval of $ [\omega, 1] $, and the value of $ \alpha $ must be increased as the overall slope the pricing function will be steeper. For any $ \alpha $, to minimize the choice of $ \omega $, we set
	\begin{align*}
	\alpha = \frac{h(\ubar{p})}{F_{\ubar{p}}(\omega)} \text{ or } \omega = F_{\ubar{p}}^{-1}\left(\frac{1}{\alpha} h\big(\ubar{p}\big)\right).
	\end{align*}
	\item \textbf{Claim 2}: for each given $ \omega\in [F_{\ubar{p}}^{-1}\big(\frac{1}{\alpha} h\big(\ubar{p}\big)\big),\ubar{\rho}] $, $ \Gamma_1(\omega,u) $ is strictly decreasing in $ u\in (\omega,1)$. Meanwhile, $ \Gamma_2(u) $ is strictly increasing in $ u\in (\omega,1) $.  The proof of the monotonicity of $ \Gamma_1(\omega,u) $ and $ \Gamma_2(u) $ is elementary based on our analysis in Appendix \ref{proof_of_two_lower_bounds_case_1} above.  To visualize the intuition,  Fig. \ref{two_alpha_intersection} illustrates $ \Gamma_1(\omega,u) $ and $ \Gamma_2(u) $ w.r.t. $ u\in (\omega,1) $ for a given $ \omega $.  We can prove that for any $ \omega \in (0,\bar{\rho}] $, there exists a unique  $ u(\omega) $ such that $ \Gamma_1(\omega,u(\omega)) = \Gamma_2(u(\omega)) $, as illustrated by the intersection point in Fig. \ref{two_alpha_intersection}. The proof is straightforward by evaluating the values of $ \Gamma_1(\omega,u) $ and $ \Gamma_2(u) $ at their two boundaries, e.g., $ \Gamma_1(\omega,u)$ and $ \Gamma_2(u) $ approaches infinity if $ u $ approaches $ \omega  $ and $ 1 $, respectively, Therefore, the monotonicity of  $ \Gamma_1(\omega,u)$ and $ \Gamma_2(u) $ indicates that there must exist a unique intersection point $ u(\omega) $ in the interval $ [\omega,1] $.  To minimize $ \alpha $, we simply choose $ u= u(\omega) $ so that $ \alpha = \alpha(\omega) $ is minimized for a given $ \omega\in [\omega,\ubar{\rho}] $. 
\end{itemize}


Based on the above two claims, to minimize $ \alpha(\omega) $, we must set $ \omega =  F_{\ubar{p}}^{-1}\left(\frac{1}{\alpha} h\big(\ubar{p}\big)\right)$. In this case the unique intersection point between $ \Gamma_1(\omega,u) $ and $ \Gamma_2(u) $ is the optimal design of $ u $, and the corresponding $ \alpha $ is the minimum competitive ratio, i.e., $ \alpha_*(\mathcal{S}) $. In summary, the optimal design of $ \omega $ and $ u $ is the unique solution to the following system of equations
\begin{align*}
\frac{h\big(\ubar{p}\big)}{F_{\ubar{p}}(\omega_*)}  =  \Gamma_1(\omega_*,u_*) = \Gamma_2(u_*).
\end{align*}
We thus complete the proof of Theorem \ref{existence_uniqueness_opt_case_1}.

\section{Remaining Proofs of Theorem \ref{major_results} in Case-2 and Case-3} \label{missing_proofs_case_2_case_3}
In this section, we provide the remaining proofs of Theorem \ref{major_results} in \textbf{Case-2} and \textbf{Case-3}.

\subsection{Case-2: $ \protect \underline{c} < \protect\overline{c} \leq \protect\underline{p} \leq \protect\overline{p}$}\label{section_Case_III}

In this case, $  \bar{c}\leq \ubar{p}\leq \bar{p} $ indicates that $ \ubar{\rho} =  \bar{\rho}=1 $. Meanwhile, there exist  no such a resource utilization threshold $ u $ so that $ \phi(u) = \bar{c} $, as illustrated in Fig. \ref{fig_rationality_sufficiency}(c). Based on Theorem \ref{sufficiency} and Theorem \ref{necessity}, we give the following Corollary  \ref{theorem_case_2_two_BVP} which summarizes the sufficient and necessary conditions in \textbf{Case-2}. 

\begin{corollary}\label{theorem_case_2_two_BVP}
	Given a convex setup $ \mathcal{S} $ in \textbf{Case-2}, $ \textsf{PM}_\phi  $ is IC and $ \alpha $-competitive if there exists a  critical threshold $ \omega \in [F_{\ubar{p}}^{-1}\big(\frac{1}{\alpha} h\big(\ubar{p}\big)\big),1] $ such that $ \phi $ is given by
	\begin{align}\label{phi_case_2}
	\phi(y) = 
	\begin{cases}
	\ubar{p},          &\text{if } y\in [0,\omega],\\
	\hat{\varphi}_2(y), &\text{if } y\in [\omega,1],\\
	+\infty, &\text{if } y\in (1,+\infty),
	\end{cases}
	\end{align}
	where $ \hat{\varphi}_2(y) $ is a solution to the following BVP:
	\begin{align}\label{BVP_case_2}
	\normalfont\widehat{\textsf{BVP}}_2(\omega,\alpha) 
	\begin{cases}
	\hat{\varphi}_2'(y) = \alpha\cdot \big(\hat{\varphi}_2(y) - f'(y) \big),  y\in (\omega,1);\\
	\hat{\varphi}_2(\omega) = \ubar{p},  \hat{\varphi}_2(1) \geq \bar{p}.
	\end{cases}
	\end{align}
	On the other hand, if there exists an $ \alpha $-competitive online algorithm, then there must exist  a strictly-increasing solution  to  $ \normalfont\widehat{\textsf{BVP}}_2(\omega,\alpha)  $ for some $ \omega \in [F_{\ubar{p}}^{-1}\big(\frac{1}{\alpha} h\big(\ubar{p}\big)\big),1] $. 
\end{corollary} 
\begin{proof}
	The proof is similar to our previous analysis regarding Corollary \ref{theorem_case_1_two_BVP} in \textbf{Case-1}. 
\end{proof}

Similar to Proposition \ref{two_lower_bounds_case_1}, we give the following Proposition \ref{def_Gamma_2_hat} to show the existence of function $ \hat{\Gamma}_2(\omega) $.
\begin{proposition}\label{def_Gamma_2_hat}
	For each given $ \omega \in [0,1)$, $\normalfont \widehat{\textsf{BVP}}_2(\omega,\alpha)  $ has a unique strictly-increasing solution if and only if $ \alpha \geq  \widehat{\Gamma}_2(\omega) $, where $ \widehat{\Gamma}_2(\omega) $ satisfies the following equation in variable $ \widehat{\Gamma}_2 $:
	\begin{align}\label{Gamma_2_hat}
	\int_\omega^{1}\frac{\widehat{\Gamma}_2 f'(y)}{ \exp\big(y \widehat{\Gamma}_2\big) } dy = \frac{\bar{p}}{\exp\big(\widehat{\Gamma}_2\big)} - \frac{\ubar{p}}{\exp\big(\omega\widehat{\Gamma}_2\big)}.
	\end{align}
\end{proposition}
\begin{proof}
	Eq. \eqref{Gamma_2_hat} is similar to Eq. \eqref{Gamma_2} and thus we omit the proof for brevity.
\end{proof}

Based on Proposition \ref{def_Gamma_2_hat}, we give the following Theorem \ref{existence_uniqueness_opt_case_2} to show the unique existence of $ \omega_* $ as well as the calculation of $ \alpha_*(\mathcal{S}) $ in \textbf{Case-2}. 

\begin{theorem}\label{existence_uniqueness_opt_case_2}
	Given a convex setup $ \mathcal{S} $ in \textbf{Case-2}, there exists a unique $ \omega_*\in (0,1)$ such that the optimal competitive ratio $ \alpha_*(\mathcal{S}) $ is given by
	\begin{align}\label{system_of_equations_case_2}
	\alpha_*({\mathcal{S}}) = \frac{h\big(\ubar{p}\big)}{F_{\ubar{p}}(\omega_*)} = \widehat{\Gamma}_2(\omega_*),
	\end{align}
	where $ \omega_*\in [0,1) $ is the unique root that satisfies Eq. \eqref{system_of_equations_case_2}. 
	Meanwhile,  $ \textsf{PM}_\phi  $ is $ \alpha_*({\mathcal{S}}) $-competitive if and only if $ \phi $ is given by Eq. \eqref{phi_case_2} with $(\omega,\alpha) = (\omega_*,\alpha_*(\mathcal{S}))  $.  
\end{theorem}
\begin{proof}
	The result follows similarly as Theorem \ref{existence_uniqueness_opt_case_1}. 
\end{proof}

Note that in \textbf{Case-2}, the calculation of the unique $ \omega_* $ is equivalent to finding the unique root to the following equation
\begin{align*}
\int_{\omega_*}^{1} \frac{h\big(\ubar{p}\big)}{F_{\ubar{p}}(\omega_*)}\cdot \exp\left(- \frac{h\big(\ubar{p}\big)}{F_{\ubar{p}}(\omega_*)}\cdot  y\right)\cdot f'(y)\   dy = \bar{p}\cdot \exp\left(-\frac{h\big(\ubar{p}\big)}{F_{\ubar{p}}(\omega_*)}\right) - \ubar{p}\cdot\exp\left(-\frac{h\big(\ubar{p}\big)}{F_{\ubar{p}}(\omega_*)}\cdot \omega_*\right).
\end{align*}

\subsection{Case-3: $ \protect\underline{c}< \protect\underline{p}\leq \protect\overline{p}\leq \protect\overline{c}$}\label{section_Case_V}
In this case, $  \ubar{c}<\ubar{p}\leq \bar{p}\leq \bar{c} $ indicates that $ \ubar{\rho}\leq \bar{\rho} \leq 1 $. Meanwhile, there exist  no such a threshold $ u\in (0,1) $ so that $ \phi(u) = \bar{c} $, as illustrated in Fig. \ref{three_cases}(d). Based on Theorem \ref{sufficiency} and Theorem \ref{necessity}, we give the following Corollary  \ref{theorem_case_3_two_BVP} which summarizes the sufficient and necessary conditions in \textbf{Case-3}. 

\begin{corollary}\label{theorem_case_3_two_BVP}
	Given a convex setup $ \mathcal{S} $ in \textbf{Case-3}, $ \textsf{PM}_\phi  $ is IC and $ \alpha $-competitive if there exists a pair of resource utilization thresholds $ (\omega,\rho) \in [F_{\ubar{p}}^{-1}\big(\frac{1}{\alpha} h\big(\ubar{p}\big)\big),\ubar{\rho}]\times [\omega,\bar{\rho}] $ such that $ \phi $ is given by
	\begin{align}\label{phi_case_3}
	\phi(y) = 
	\begin{cases}
	\ubar{p},          &\text{if } z\in [0,\omega],\\
	\hat{\varphi}_1(y), &\text{if } z\in [\omega,\rho],\\
	+\infty,      &\text{if } z\in (\rho,1],
	\end{cases}
	\end{align}
	where  $ \hat{\varphi}_1(y) $ is a solution to the following BVP:
	\begin{align}\label{BVP_case_3}
	\normalfont\widehat{\textsf{BVP}}_1(\omega,\rho,\alpha) 
	\begin{cases}
	\hat{\varphi}_1'(y) = \alpha\cdot \frac{\hat{\varphi}_1(y)  - f'(y) }{f'^{-1}\big(\hat{\varphi}_1(y)\big)}, y\in(\omega,\rho);\\
	\hat{\varphi}_1(\omega) = \ubar{p}, \hat{\varphi}_1(\rho) \geq  \bar{p}.
	\end{cases}
	\end{align}
	On the other hand, if there exists an $ \alpha $-competitive online algorithm, then there must exist  a strictly-increasing solution  to $ \normalfont\widehat{\textsf{BVP}}_1(\omega,\rho,\alpha) $ for some pair of $ (\omega,\rho) \in [F_{\ubar{p}}^{-1}\big(\frac{1}{\alpha} h\big(\ubar{p}\big)\big),\ubar{\rho}]\times [\omega,\bar{\rho}] $.  
\end{corollary} 
\begin{proof}
	The proof is similar to our analysis of Theorem \ref{existence_uniqueness_opt_case_1}. 
\end{proof}

Similar to Proposition \ref{two_lower_bounds_case_1}, we give the following Proposition \ref{lower_bounds_case_3} to show the existence of function $ \hat{\Gamma}_1(\omega,\rho) $.

\begin{proposition}\label{lower_bounds_case_3}
	Given a convex setup $ \mathcal{S} $ in \textbf{Case-3}, for each given pair of $ (\omega,\rho) \in [0,\ubar{\rho})\times (\omega,\bar{\rho}]$, there exists a well-defined function $ \widehat{\Gamma}_1(\omega,\rho) $ so that $ \normalfont\widehat{\textsf{BVP}}_1(\omega,\rho,\alpha) $ has a unique strictly-increasing  solution if and only if $ \alpha \geq  \widehat{\Gamma}_1(\omega,\rho)$.
\end{proposition}
\begin{proof}
	The result follows the first bullet in Proposition \ref{two_lower_bounds_case_1}.
\end{proof}

Based on Proposition \ref{lower_bounds_case_3}, we give the following Theorem \ref{existence_uniqueness_opt_case_3} to show the calculation of $ \alpha_*(\mathcal{S}) $ in \textbf{Case-3}. 

\begin{theorem}\label{existence_uniqueness_opt_case_3}
	Given a convex setup $ \mathcal{S} $ in \textbf{Case-3}, there exists a unique $ \omega_*\in [0,\ubar{\rho}]$ such that the optimal competitive ratio $ \alpha_*(\mathcal{S}) $ is given by
	\begin{align}\label{system_of_equations_case_3}
	\alpha_*(\mathcal{S}) = \frac{h\big(\ubar{p}\big)}{F_{\ubar{p}}(\omega_*)} = \widehat{\Gamma}_1\left(\omega_*,\bar{\rho}\right),
	\end{align}
	where $ \omega_* $ is the unique root that satisfies Eq. \eqref{system_of_equations_case_3}. 
	Meanwhile,  $ \textsf{PM}_\phi  $ is $ \alpha_*({\mathcal{S}}) $-competitive if and only if $ \phi $ is given by Eq. \eqref{phi_case_3} with $(\omega,\rho,\alpha) = (\omega_*,\bar{\rho}, \alpha_*(\mathcal{S}))  $.  
\end{theorem}
\begin{proof}
	The proof is similar to our analysis of Theorem \ref{existence_uniqueness_opt_case_1}. 
\end{proof}

Based on Theorem \ref{existence_uniqueness_opt_case_1}, Theorem \ref{existence_uniqueness_opt_case_2}, and Theorem \ref{existence_uniqueness_opt_case_3}, Theorem \ref{major_results} follows. We thus complete the proof of Theorem \ref{major_results} in all the three cases.

\begin{remark}
	Based on Theorem \ref{existence_uniqueness_opt_case_1}, Theorem \ref{existence_uniqueness_opt_case_2}, and Theorem \ref{existence_uniqueness_opt_case_3}, we conclude that the upper bound of the optimal pricing function is always equal to the maximum valuation density, namely, $ \phi_*(\bar{\rho}) = \bar{p} $.  Note that compared to $ \phi_*(\bar{\rho}) > \bar{p} $, by setting $ \phi_*(\bar{\rho}) = \bar{p} $ we can have a higher resource utilization level while still guarantee the competitive ratio of $ \textsf{PM}_\phi $. The intuition here is that it is always beneficial to have a higher resource utilization level as long as the resources are sold at the right price. 
\end{remark}

\section{Proof of Corollary \ref{property_of_alpha_S}}
\label{proof_of_property_of_alpha_S}
The properties of $ \alpha_*(\mathcal{S}) $ directly follow the monotonicity of $ \Gamma_1(\omega,u) $ and $ \Gamma_2(u) $. We briefly explain them as follows:
\begin{itemize}
	\item For a given $ \bar{p}\in (\ubar{c},+\infty) $, we claim that $ \Gamma_1(\omega,u) $ is strictly decreasing in $ \ubar{p}\in (\ubar{c},\bar{p}] $. We use \textbf{Case-1} to explain the intuition. A larger $ \ubar{p} $ indicates a smaller value of $ \bar{c} - \ubar{p} $, and thus for a given $ \omega $ and $ u $, the pricing function $ \varphi_1 $ can increase from $ \ubar{p} $ to $ \bar{c} $ with a smaller rate. A smaller increasing rate of $ \varphi_1 $ indicates that the value of $ \alpha $ can be smaller as well. Therefore, $ \Gamma_1(\omega,u) $ is decreasing in $ \ubar{p}\in (\ubar{c},\bar{p}] $. As we can see from Fig. \ref{two_alpha_intersection}, a smaller $ \Gamma_1(\omega,u) $ indicates a smaller value of $ u(\omega)$ (note that $ \Gamma_2(u) $ will not change w.r.t. the increasing of $ \ubar{p} $), as well as a smaller value of $ \alpha(\omega) $, which thus leads to a smaller value of $ \alpha_*(\mathcal{S}) $ in the end. Therefore, $\alpha_*(\mathcal{S})  $ is strictly decreasing in $ \ubar{p}\in (\ubar{c},\bar{p}] $ for a given $ \bar{p}\in (\ubar{c},+\infty) $. 
	
	\item  We can prove similarly that $\alpha_*(\mathcal{S})  $ is strictly increasing in $ \bar{p}\in [\ubar{p},+\infty) $ for a given $ \ubar{p}\in (\ubar{c},+\infty) $. We skip the details for brevity.
	
	\item Note that  $ \ubar{p} = \bar{p} \in  (\ubar{c},+\infty)$ happens only in \textbf{Case-2} and \textbf{Case-3}. For instance, in \textbf{Case-2}, $ \ubar{p} = \bar{p} \in  (\bar{c},+\infty)$ indicates that $ \hat{\varphi}_2(\omega_*)  = \hat{\varphi}_2(1) = \ubar{p} = \bar{p} $, where $ \hat{\varphi}_2 $ is the unique solution to the BVP in Eq. \eqref{BVP_case_2}. We argue that such a solution occurs only when $ \omega_* = 1 $ and in this case the optimal competitive ratio $ \alpha_*(\mathcal{S}) = 1 $ since both the online and offline strategies are to satisfy all the request until reaching the total capacity limit. We can prove similarly that in \textbf{Case-3}, $ \omega_* = \ubar{\rho} $ and $ \alpha_*(\mathcal{S})=1 $ as long as $ \ubar{p} = \bar{p} \in (\ubar{c},\bar{c}] $. Combining the results in both cases, we have $ \alpha_*(\mathcal{S}) = 1 $ and $ \omega_* = \ubar{\rho} $ when $ \ubar{p} = \bar{p} \in (\ubar{c},+\infty) $.	
\end{itemize}

We thus complete the proof of Corollary \ref{property_of_alpha_S}.

\section{Proof of Theorem \ref{sufficiency_multiple_paper}}
\label{section_extension}
Before presenting the proof of Theorem \ref{sufficiency_multiple_paper}, we first give the extension of $ \textsf{PM}_\phi $ in Algorithm \ref{PriMe_multiple}.

\begin{algorithm}
	\caption{Posted Price Mechanism ($ \textsf{PM}_{\bm{\phi}} $)}	
	\begin{algorithmic}[1]
		\STATE \textbf{Inputs:} A given setup $ \mathcal{S} = \{f_t,\ubar{p}_t,\bar{p}_t\}_{\forall t\in\mathcal{T}} $ and $ \bm{\phi} = \{\phi_t\}_{\forall t\in\mathcal{T}}$. 
		
		\STATE \textbf{Initialize}: $ \hat{y}_t^{(0)} = 0 $ and $ \hat{p}_t^{(0)} = \phi(\hat{y}_t^{(0)})$ for all $ t\in\mathcal{T} $.  
		
		\WHILE{a new agent $ n $ arrives}
		
		\STATE Supplier publishes the price $ \hat{p}_t^{(n-1)} $ for $ t\in\mathcal{T} $. 
		
		\IF {$ v_n -  \sum_{t\in\mathcal{T}_n} r_n^t \cdot \phi_t\big(y_t^{(n-1)}\big) < 0 $}
		
		\STATE Agent $ n $ leaves (i.e., set $ \hat{x}_n =0 $)
		
		\ELSIF{$\hat{y}_t^{(n-1)}+ r_n^t >1$ for any $ t\in \mathcal{T} $}
		\STATE Request $ n $ is rejected (i.e., set $ \hat{x}_n =0 $)
		
		\ELSE 
		\STATE Request $ n $ is satisfied (i.e., set $ \hat{x}_n = 1 $)
		
		\STATE Set  the payment $ \hat{\pi}_n $ by $ \hat{\pi}_n =  \sum_{t\in\mathcal{T}_n} r_n^t \cdot \phi_t\Big(y_t^{(n-1)}\Big)  $. 
		
		\STATE Update the total resource utilization by  $ \hat{y}_t^{(n)} =  \hat{y}_t^{(n-1)} + r_n^t $, $ \forall t\in\mathcal{T} $. 
		
		\STATE Update the price by $
		\hat{p}_t^{(n)} = \phi_t\left(\hat{y}_t^{(n)}\right), \forall t\in\mathcal{T}$. 
		
		\ENDIF
		\ENDWHILE
	\end{algorithmic}
	\label{PriMe_multiple}
\end{algorithm}

Based on Algorithm \ref{PriMe_multiple}, the sufficient conditions of Theorem \ref{sufficiency_multiple_paper} can be derived based on Proposition \ref{OPD_principle} as follows.  We first prove that there indeed exists an index $ k\in\mathcal{N} $ such that the initial inequality hold. The major argument is that in the worst-case, each agent may only require resource for a single time slot, leading to the existence of a critical threshold $ \omega_t $ for each $ t\in\mathcal{T} $. Specifically, suppose $ k\in \mathcal{N} $ and after processing all the agents within $ \{1,2,\cdots,k\} $, we have
\begin{align*}
P_k =  \sum_{n=1}^k v_n -  \sum_{t\in\mathcal{T}}\bar{f}_t(\hat{y}_t^{(k)}) = \sum_{n=1}^k \left(\hat{\gamma}_n  + \sum_{t\in\mathcal{T}} r_n^t \cdot \phi_t(y_t^{(n-1)})\right)- \sum_{t\in\mathcal{T}}\bar{f}_t(\hat{y}_t^{(k)}).
\end{align*}
The dual objective after processing request $ k $ is given by
\begin{align*}
D_k = \sum_{n=1}^{k} \hat{\gamma}_n +   \sum_{t\in\mathcal{T}} h_t(\hat{p}_t^{(k)}).
\end{align*}
Based on the pricing function given by  Eq. \eqref{sufficient_phi_multiple_paper}, the initial inequality $ P_k\geq \frac{1}{\alpha_t}D_k $ leads to the following
\begin{align*}
\left(1 - \frac{1}{\alpha_t}\right)\sum_{n=1}^k \hat{\gamma}_n +   \sum_{t\in\mathcal{T}} \ubar{p}_t\cdot \left(\sum_{n=1}^kr_n^t\right)- \sum_{t\in\mathcal{T}}\bar{f}_t\left(\sum_{n=1}^k r_n^t\right)  \geq \frac{1}{\alpha_t} \sum_{t\in\mathcal{T}} h_t\big
(\ubar{p}_t\big).
\end{align*}
Since $ \alpha_t\geq 1 $ and $ \hat{\gamma}_n\geq 0 $, it suffices to have the following inequality for each time slot:
\begin{align}\label{p_j}
\ubar{p}_t\cdot \left(\sum_{n=1}^kr_n^t\right)- \bar{f}_t\left(\sum_{n=1}^k r_n^t\right)  \geq \frac{1}{\alpha_t} h_t\big(\ubar{p}_t\big).
\end{align}
The above inequality Eq. \eqref{p_j}
holds if the following one holds:
\begin{align*}
\ubar{p}_t \omega_t - \bar{f}_t\left(\omega_t\right) = F_{\ubar{p}_t}(\omega_t) \geq \frac{1}{\alpha_t}  h_t\big(\ubar{p}_t\big),
\end{align*}
where $\omega_t =  \sum_{n=1}^kr_n^t $. Meanwhile, for the same reason as the previous basic resource allocation model, it is obvious that $ \omega_t $ must be less than or equal to $ \ubar{\rho}_t $ because the selling price must be larger than or equal to the marginal cost, i.e,. $ \phi_t(\omega_t) = \ubar{p}_t \geq f_t'(\omega_t) $. Therefore, Eq. \eqref{flat_sufficiency_omega_multiple_paper} indicates that the initial inequality $ P_k\geq \frac{1}{\alpha_t}D_k $ holds for some $ k\in\mathclap{N} $. 

We next prove the  ordinary differential inequality in Eq. \eqref{ODI_principle_sufficiency_multiple_paper}.  The change in the primal objective is given as follows:
\begin{align*}
P_n - P_{n-1}
=\ & v_n - \sum_{t\in\mathcal{T}}\left(\bar{f}_t(\hat{y}_t^{(n)})-\bar{f}_t(\hat{y}_t^{(n-1)})\right)\\
=\ & \hat{\gamma}_n +\sum_{t\in\mathcal{T}}\phi_t(\hat{y}_t^{(n-1)})\left(\hat{y}_t^{(n)}-\hat{y}_t^{(n-1)}\right)-\sum_{t\in\mathcal{T}} \left(\bar{f}_t(\hat{y}_t^{(n)})-\bar{f}_t(\hat{y}_t^{(n-1)})\right).
\end{align*}
The change in the dual objective is given as follows:
\begin{align*}
D_n-D_{n-1} = \hat{\gamma}_n +  \sum_{t\in\mathcal{T}}\left( h_t(\hat{p}_t^{(n)})-h_t(\hat{p}_t^{(n-1)})\right).
\end{align*}
To guarantee that $ P_n- P_{n-1}\geq \frac{1}{\alpha}\left(D_n -D_{n-1}\right) $ holds, it is equivalent to having the following
\begin{align*}
\left(1-\frac{1}{\alpha_t}\right)\hat{\gamma}_n + \sum_{t\in\mathcal{T}}  \left(\phi_t(\hat{y}_t^{(n-1)})\big(\hat{y}_t^{(n)}-\hat{y}_t^{(n-1)}\big)- \big(\bar{f}_t(\hat{y}_t^{(n)})-\bar{f}_t(\hat{y}_t^{(n-1)})\big)\right)\\
\geq \frac{1}{\alpha_t}  \sum_{t\in\mathcal{T}}\left( h_t(\hat{p}_t^{(n)})-h_t(\hat{p}_t^{(n-1)})\right),
\end{align*}
where the first term $ \hat{\gamma}_n \geq 0$ always holds. Therefore, if the following inequality holds at each round, then the incremental inequality holds as well:
\begin{align}\label{phi_j}
\sum_{t\in\mathcal{T}}  \left(\phi_t(\hat{y}_t^{(n-1)})\big(\hat{y}_t^{(n)}-\hat{y}_t^{(n-1)}\big)- \big(\bar{f}_t(\hat{y}_t^{(n)})-\bar{f}_t(\hat{y}_t^{(n-1)})\big)\right)
\geq\frac{1}{\alpha_t} \sum_{t\in\mathcal{T}}\left( h_t(\hat{p}_t^{(n)})-h_t(\hat{p}_t^{(n-1)})\right).
\end{align}
To guarantee that Eq. \eqref{phi_j}
holds, it suffices to have the following per-slot inequality
\begin{align*}
\phi_t(\hat{y}_t^{(n-1)})\left(\hat{y}_t^{(n)}-\hat{y}_t^{(n-1)}\right)- \left(\bar{f}_t(\hat{y}_t^{(n)})-\bar{f}_t(\hat{y}_t^{(n-1)})\right)
\geq\frac{1}{\alpha_t}  \left( h_t(\hat{p}_t^{(n)})-h_t(\hat{p}_t^{(n-1)})\right),
\end{align*}
which is equivalent to the following:
\begin{align*}
\phi_t(\hat{y}_t^{(n-1)})-\frac{\bar{f}_t(\hat{y}_t^{(n)})-\bar{f}_t(\hat{y}_t^{(n-1)})}{\hat{y}_t^{(n)}-\hat{y}_t^{(n-1)}}
\geq \frac{1}{\alpha_t} \cdot \frac{h_t\left(\phi_t(\hat{y}_t^{(n)})\right)-h_t\left(\phi_t(\hat{y}_t^{(n-1)})\right)}{\phi_t(\hat{y}_t^{(n)})-\phi_t(\hat{y}_t^{(n-1)})}\cdot \frac{\phi_t(\hat{y}_t^{(n)})-\phi_t(\hat{y}_t^{(n-1)})}{\hat{y}_t^{(n)}-\hat{y}_t^{(n-1)}}.
\end{align*}
Therefore, we have the following inequality in the differential form:
\begin{align}\label{ODE_multiple_appendix}
\phi_t(y) - \bar{f}_t'(y) \geq \frac{1}{\alpha_t}\cdot h_t'\big(\phi_t(y)\big)\cdot\phi_t'(y), \forall y\in [\omega_t,1).
\end{align}
This means, if Eq. \eqref{ODE_multiple_appendix} holds for all $ y\in [\omega_t,1) $, then the incremental inequality holds at each round for all $ y\in [\omega_t,1) $. 

Similar to our previous analysis, we need to guarantee that the incremental inequality holds at each round for all $ y\in [\omega_t,1] $. Therefore, we need to consider the boundary condition of $ \phi_t(y) $ when $ y = 1 $.  First, it is obvious that $ \phi_t(\omega_t) = \ubar{p}_t $  as the pricing function $ \phi_t $ must be continuous.  Taking integration of both sides of Eq. \eqref{ODE_multiple_appendix} leads to
\begin{align*}
\int_{\omega_t}^{\rho_t}\phi_t(y)dy - \bar{f}_t(\rho_t) + \bar{f}_t(\omega_t) \geq  \frac{1}{\alpha_t}\Big(h_t\big(\phi_t(\rho_t)\big) - h_t\big(\ubar{p}_t\big)\Big),
\end{align*}
where $ \rho_t $ denotes the final utilization level at time slot $ t\in\mathcal{T} $.
Meanwhile, based on the flat-segment, we have 
$ \ubar{p}_t \omega_t - \bar{f}_t(\omega_t) \geq \frac{1}{\alpha_t} h_t(\ubar{p}_t) $. 
Thus, we have
\begin{align*}
\ubar{p}_t\omega_t + \int_{\omega_t}^{\rho_t}\phi_t(y)dy - \bar{f}_t(\rho_t) \geq  \frac{1}{\alpha_t}h_t\big(\phi_t(\rho_t)\big),
\end{align*}
which is the same as Eq. \eqref{integral_version_expanding_single_appendix}. 
When $ \rho_t = 1 $, the above equality also holds, and thus we have
\begin{align*}
\ubar{p}_t\omega_t + \int_{\omega_t}^{1}\phi_t(y)dy - \bar{f}_t(1) \geq  \frac{1}{\alpha_t}h_t\big(\phi_t(1)\big).
\end{align*}

On the other hand, $ \textsf{PM}_{\bm{\phi}} $ is $ \alpha_t $-competitive indicates that the following inequality must hold:
\begin{align}\label{boundary_p_bar_j}
\ubar{p}_t\omega_t + \int_{\omega_t}^{1}\phi_t(y)dy - \bar{f}_t(1) \geq  \frac{1}{\alpha_t} \sum_{t\in \mathcal{T}}h_t(\bar{p}_t),
\end{align}
where the right-hand-side represents the optimal social welfare in the offline setting when the resource at time $ t $ has been used up. \textit{Note that here the worst-case scenario is different from the single time slot case. The rationality of Eq. \eqref{boundary_p_bar_j} is as follows. When the resource at time $ t $ is used up by agents whose valuation equals the payment, let us assume that there is  a follow-up arrival instance  which consists of  two groups of agents as follows: i) the first group of agents ask for resource at time $ t $ with the maximum valuation density $ \bar{p}_t $, and ii) the second group of agents ask for resources at time $ t $ as well as at time $ \hat{t} $ with the maximum valuation density $ \bar{p}_{\hat{t}} $, $ \forall \hat{t}\in\mathcal{T}\backslash \{t\} $. Meanwhile, we assume that the number of agents in these two groups are large enough so that the total capacity of each time slot can be used up. In the online setting, the social welfare is given by the left-hand-side of Eq. \eqref{boundary_p_bar_j}, meaning that the valuation of the accepted agents are all  equal to their payments, while the follow-up agents are all rejected due to the capacity limit at time $ t $. In comparison, in the offline setting, the optimal social welfare is to allocate all the capacity at time $ t $ to those 
	agents in the follow-up instance whose valuation densities are $ \bar{p}_t $, $ \forall t\in\mathcal{T} $, as a result, the optimal optimal social welfare in hindsight is $ \sum_{t\in \mathcal{T}}h_t(\bar{p}_t) $, as given by the right-hand-side of Eq. \eqref{boundary_p_bar_j}.} Therefore, to guarantee that the above Eq. \eqref{boundary_p_bar_j} holds, it suffices to have the following boundary condition:
$ \phi_t(1) - f_t(1) \geq \sum_{t\in \mathcal{T}}h_t(\bar{p}_t) \Longrightarrow \phi_t(1) \geq \bar{p}_t + \sum_{\hat{t}\in \mathcal{T}\backslash\{t\}}h_{\hat{t}}(\bar{p}_{\hat{t}}) $. 

Summarizing our above proofs regarding the critical threshold, the differential inequality and the boundary conditions, we complete the proof of Theorem \ref{sufficiency_multiple_paper}.

(\textit{Impacts of Multiple Time Slots on the Competitive Ratio}) Based on the second boundary condition in Eq. \eqref{ODI_principle_sufficiency_multiple_paper}, $ \phi_t(1) $ follows roughly linearly w.r.t. the total number of time slots. For example, when we assume $ \bar{p}_t\geq \bar{c}_t $ for all $ t\in\mathcal{T} $, then the second boundary condition in Eq. \eqref{ODI_principle_sufficiency_multiple_paper} indicates that $\phi_t(1)\geq  \bar{p}_t + \sum_{\hat{t}\in \mathcal{T}\backslash\{t\}}h_{\hat{t}}(\bar{p}_{\hat{t}}) = \sum_{t\in \mathcal{T}} \bar{p}_{t} - \sum_{\hat{t}\in \mathcal{T}\backslash\{t\}} f_{\hat{t}}(1)$. 
Our previous analysis regarding the basic model shows that the calculation of the optimal competitive ratio $ \alpha_*(\mathcal{S}) $ depends on the boundary condition of $ \phi_* $ (i.e., the monotonicity of $ \alpha_*(\mathcal{S}) $ in Corollary \ref{property_of_alpha_S}). Therefore, in the case with multiple time slots, the final competitive ratio $ \alpha = \max_t\{\alpha_t\} $  depends on $ |\mathcal{T}| $ as well. As we mention in the conclusion of the paper, it is difficult to  analyze the property  of $ \alpha_*(\mathcal{S}) $ w.r.t. $ \ubar{p} $ and $ \bar{p} $. Similarly, here it is also difficult to analyze the dependency of $ \alpha $ on the total number of time slots. We leave this for future research. 

(\textit{Necessary Conditions}) The necessary conditions in Theorem \ref{necessity} can also be generalized to consider the case with multiple time slots. We skip the details for brevity as most of the proofs are similar to our previous analysis in the paper, except that the worst-case analysis regarding the boundary conditions should be performed in a different manner, as demonstrated by Eq. \eqref{boundary_p_bar_j}.

\section{Computation of $\omega_*$, $u_*$, and $\alpha_*(\mathcal{S}) $}
\label{appendix_compute_alpha}

In this section, we present the algorithms for computing $\omega_*$, $u_*$, and $\alpha_*(\mathcal{S}) $ in \textbf{Case-1}. We skip the computational details regarding \textbf{Case-2} and \textbf{Case-3} for brevity.

We give Algorithm \ref{Gamma_1_function} below to show the calculation of $ \Gamma_1(\omega,u) $ for a given $ \omega $ and $ u $. We also give Algorithm \ref{Gamma_2_function} to show the calculation of $ \Gamma_2(u) $ for a given $ u $. Based on $\texttt{\textbf{GAMMA}}_1(\omega,u) $ in Algorithm \ref{Gamma_1_function} and $\texttt{\textbf{GAMMA}}_2(u) $ in Algorithm \ref{Gamma_2_function}, Algorithm \ref{compute_alpha} shows the calculation of $\omega_*$, $u_*$, and $\alpha_*(\mathcal{S}) $ for a given convex setup $ \mathcal{S} $.

\begin{algorithm}
	\caption{$\texttt{\textbf{GAMMA}}_1(\omega,u) $} 	
	\begin{algorithmic}[1]
		\STATE \textbf{Inputs:} A given setup $ \mathcal{S} = \{f,\ubar{p},\bar{p}\} $, $ \omega$, $ u $, and $ \epsilon $.  Set $ \alpha_1 = 1 $, $ \alpha_2 =  R_1$, and $ \varphi_1(u) = R_2 $. \label{R_12}
		
		\WHILE{$ |\varphi_1(u) - \bar{c}| > \epsilon $}
		
		\STATE Compute $ \alpha_0 = \frac{\alpha_1+\alpha_2}{2} $.
		
		\STATE Solve the following IVP and obtain the solution $ \varphi_1(y) $: \label{ivp_solver}
		\begin{align*}
		\normalfont\textsf{IVP}(\omega,u,\alpha_0) 
		\begin{cases}
		\varphi_1'(y) = \alpha_0 \cdot \frac{\varphi_1(y)  - f'(y) }{f'^{-1}\left(\varphi_1(y)\right)}, y\in(\omega,u);\\
		\varphi_1(\omega) = \ubar{p}.
		\end{cases}
		\end{align*}
		
		\IF {$ \varphi_1(u) > \bar{c} $}
		\STATE $ \alpha_2 = \alpha_0 $.
		\ELSE
		\STATE $ \alpha_1 = \alpha_0 $.
		\ENDIF
		
		\ENDWHILE
		\STATE $ \Gamma_1(\omega,u)  = \alpha_0 $.
		\STATE \textbf{Outputs:} $ \Gamma_1(\omega,u)$.  
	\end{algorithmic}
	\label{Gamma_1_function}	
\end{algorithm}

We can see that Algorithm \ref{Gamma_1_function} performs a bisection search for the unique output of $ \Gamma_1(\omega,u) $.
Note that $ \epsilon $ is a constant that is sufficiently small. 
In both Algorithm \ref{Gamma_1_function} and Algorithm \ref{Gamma_2_function}, $ R_1 $ and $ R_2 $ are two large positive constants that can be initialized appropriately based on the setup $ \mathcal{S} $. In line \ref{ivp_solver} of Algorithm \ref{Gamma_1_function}, we need to solve the IVP parameterized by $ \omega, u $, and $ \alpha_0 $. We emphasize that solving $ \textsf{IVP}(\omega,u,\alpha_0) $ in this step is light-weight and  various off-the-shelf methods/solvers can be used. Algorithm \ref{compute_alpha} searches between $ [0,1] $ for the unique $ u_* $. Based on the unique $ u_* $, the optimal critical threshold $ \omega_* $ can be given by line \ref{omega_star} of Algorithm \ref{compute_alpha}. Recall that in Algorithm \ref{compute_alpha},  $ F_{\ubar{p}}^{-1} $ denotes the inverse of $ F_{\ubar{p}} $, which is defined by Eq. \eqref{profit_function}.

\begin{algorithm}
	\caption{$\texttt{\textbf{GAMMA}}_2(u) $} 	
	\begin{algorithmic}[1]
		\STATE \textbf{Inputs:} A given setup $ \mathcal{S} = \{f,\ubar{p},\bar{p}\} $, $ u$, and $ \epsilon $.  Set $ \alpha_1 = 1 $, $ \alpha_2 =  R_1$, and $ GAP  = R_2 $.
		
		\WHILE{$ |GAP| > \epsilon $}
		
		\STATE Compute $ \alpha_0 = \frac{\alpha_1 + \alpha_2}{2} $.
		
		\STATE Set $ GAP  = \int_u^{1}\frac{\alpha_0 f'(y)}{ \exp(y\alpha_0) } dy - \left(\frac{\bar{p}}{\exp(\alpha_0)} + \frac{\bar{c}}{\exp(u\alpha_0)}\right) $.
		
		\IF {$ GAP  > 0 $}
		\STATE $ \alpha_2 = \alpha_0 $.
		\ELSE
		\STATE $ \alpha_1 = \alpha_0 $.
		\ENDIF
		
		\ENDWHILE
		\STATE $ \Gamma_2(u) = \alpha_0 $. 
		\STATE \textbf{Outputs:} $ \Gamma_2(u) $. 
	\end{algorithmic}
	\label{Gamma_2_function}	
\end{algorithm}

\begin{algorithm}
	\caption{Calculation of $\omega_*$, $u_*$, and $\alpha_*(\mathcal{S}) $} 	
	\begin{algorithmic}[1]
		\STATE \textbf{Inputs}: A given setup $ \mathcal{S} = \{f,\ubar{p},\bar{p}\} $ and $ \epsilon $.  Set $ u_1 = 0 $, $ u_2 =  1$, and $ GAP = 1 $.
		
		\WHILE{$ |GAP| > \epsilon $}
		
		\STATE Compute $ u_* = \frac{u_1+u_2}{2} $.
		
		\STATE $ \Gamma_2(u_*) = \texttt{\textbf{GAMMA}}_2(u_*) $.
		
		\STATE Set $ \omega_*  = F_{\ubar{p}}^{-1}\big(\frac{h(\ubar{p})}{\Gamma_2(u_*)}\big) $. \label{omega_star}
		
		\STATE $ \Gamma_1(\omega_*,u_*) =  \texttt{\textbf{GAMMA}}_1(\omega_*,u_*) $.
		
		\STATE Set $ GAP =  \Gamma_2(u_*) - \Gamma_1(\omega_*,u_*) $.
		
		\IF {$ GAP > 0 $}
		\STATE $ u_2 = u_* $.
		\ELSE
		\STATE $ u_1 = u_* $.
		\ENDIF
		
		\ENDWHILE
		
		\STATE Set $ \alpha_*(\mathcal{S}) =  \texttt{\textbf{GAMMA}}_1(\omega_*,u_*) = \texttt{\textbf{GAMMA}}_2(u_*) $.
		
		\STATE \textbf{Outputs:} $ \omega_*, u_*, \alpha_*(\mathcal{S})$. 
	\end{algorithmic}
	\label{compute_alpha}	
\end{algorithm}

Before ending this section, we give the following remark regarding the implementation of $ \textsf{PM}_{\phi_*} $.
\begin{remark}
	It is worth pointing out that the calculation of $ \phi_* $ in a numerical manner should not be a concern for the online implementation of $ \textsf{PM}_{\phi_*} $. On the one hand, our above algorithms shows that the computation of $ \omega_*,u_* $ and $ \alpha_*(\mathcal{S}) $ can be performed efficiently via various numerical root-finding methods  such as  bisection searching. On the other hand (and perhaps more importantly), for a given convex setup $ \mathcal{S} $, the parameters $ \alpha_*(\mathcal{S}),  \omega_*, $ and $ u_* $ can be computed offline before the start of $ \textsf{PM}_{\phi_*} $. This means, we only need to solve either $ \textsf{BVP}_1(\omega_*,u_*,\alpha_*) $ or $ \textsf{BVP}_2(u_*,\alpha_*) $ ``on-the-fly", depending on which segment the total resource utilization level is located.  
\end{remark}

\end{document}